\documentclass{article}

\usepackage[preprint]{neurips_2021}

\usepackage[utf8]{inputenc} %
\usepackage[T1]{fontenc}    %
\usepackage{hyperref}       %
\usepackage{url}            %
\usepackage{booktabs}       %
\usepackage{amsfonts}       %
\usepackage{nicefrac}       %
\usepackage{microtype}      %
\usepackage{xcolor}         %

\usepackage{times}
\usepackage{soul}
\usepackage{graphicx}
\usepackage{amsmath}
\usepackage{amsthm}
\usepackage{amssymb}
\usepackage{algorithm}
\usepackage{algorithmic}
\usepackage{natbib}
\usepackage{color}
\usepackage{wrapfig}
\usepackage{bm}
\usepackage{enumitem}
\usepackage{float}
\usepackage[small]{caption}
\usepackage{caption}
\usepackage{subcaption}
\usepackage{wrapfig}
\newtheorem{thm}{Theorem}
\newenvironment{thmbis}[1]
  {\renewcommand{\thethm}{\ref{#1}}%
   \addtocounter{thm}{-1}%
   \begin{thm}}
  {\end{thm}}
\newtheorem{dfn}{Definition}

\newtheorem{lemma}{Lemma}
\newtheorem{ex}{Example}

\newtheorem{asm}{Assumption}

\newtheorem{prop}{Proposition}

\newcommand{\bbR}{\mathbb{R}}

\newcommand{\pdfs}{p}

\newcommand{\PSdist}{\mathbf{Q}}

\newcommand{\exppost}{p}

\newcommand{\exppay}{\text{EE}}

\newcommand{\nonpay}{NN^{2}}
\newcommand{\expminussemi}{F}

\newcommand{\semiA}{A^1}

\newcommand{\semiB}{B^1}

\newcommand{\alev}{almost everywhere }

\newcommand{\newtruthful}{strictly posterior truthful}

\newcommand{\newtruthfulness}{strict posterior truthfulness}
\newcommand{\Newtruthfulness}{Strict posterior truthfulness}

\newcommand{\expreporterA}{E_A}
\newcommand{\expreporterB}{E_B}
\newcommand{\semireporterA}{N^1_A}
\newcommand{\semireporterB}{N^1_B}
\newcommand{\nonreporterA}{N^2_A}
\newcommand{\nonreporterB}{N^2_B}

\newcommand{\expset}{\Theta^{0}}
\newcommand{\exprepset}{\mathcal{N}^{0}}
\newcommand{\expsig}{s^{0*}}
\newcommand{\expSig}{S^{0}}
\newcommand{\exprep}{s^{0}}
\newcommand{\expele}{s^{0}}
\newcommand{\expelek}[2]{\theta^{0}_{#1, #2}}

\newcommand{\noninum}{n}
\newcommand{\nonisig}[1]{s^{#1*}}
\newcommand{\noniSig}[1]{S^{#1}}
\newcommand{\nonirep}[1]{s^{#1}}
\newcommand{\nonirepset}[1]{\mathcal{N}^{#1}}
\newcommand{\nonipost}{p}
\newcommand{\nonidist}[1]{f^{#1}}
\newcommand{\noniele}[1]{\theta^{#1}}
\newcommand{\noniset}[1]{\Theta^{#1}}
\newcommand{\nonisetnum}[1]{m^{#1}}
\newcommand{\noniA}[1]{A^{#1}}
\newcommand{\noniB}[1]{B^{#1}}
\newcommand{\PSI}[1]{PS^{#1}}
\newcommand{\SSIdist}[1]{\mathbf{Q}^{#1}}
\newcommand{\ssidist}[1]{\mathbf{q}^{#1}}
\newcommand{\ssidisti}[1]{{q}^{#1}}
\newcommand{\nonipay}[1]{\text{NN}^{#1}}

\title{Truthful Information Elicitation from Hybrid Crowds}

\author{
    Qishen Han \thanks{Equal contributions.}\ \ \footnotemark[2] 
    \And
    Sikai Ruan \footnotemark[1]\ \ \footnotemark[2]  
    \And
    Yuqing Kong  \footnotemark[2] 
    \And 
    Ao Liu \footnotemark[3]
    \And
    Farhad Mohsin  \footnotemark[3]
    \And
    Lirong Xia \footnotemark[3]
\\ \ \\
\footnotemark[2]\ \  Department of Computer Science, Peking University, Beijing, China\\
\footnotemark[3]\ \  Department of Computer Science, RPI, Troy, NY, USA\\
Email: \texttt{\{hnick2017, sikairuan, yuqing.kong\}@pku.edu.cn, }\\
\texttt{\{liua6,mohsif\}@rpi.edu, xial@cs.rpi.edu}
}

\begin{document}

\maketitle
 
\begin{abstract}
    Suppose a decision maker wants to predict weather tomorrow by eliciting and aggregating information from crowd. How can the decision maker incentivize the crowds to report their information truthfully? Many truthful peer prediction mechanisms have been proposed for {\em homogeneous} agents, whose types are drawn from the same distribution.  However, in many situations, the population is a {\em hybrid} crowd of different types of agents with different forms of information, and the decision maker has neither the identity of any individual nor the proportion of each types of agents in the crowd. Ignoring the heterogeneity among the agent may lead to inefficient of biased information, which would in turn lead to suboptimal decisions.

In this paper, we propose the first framework for information elicitation from hybrid crowds, and two mechanisms to motivate agents to report their information truthfully. The first mechanism combines two mechanisms via linear transformations and the second is based on mutual information. With two mechanisms, the decision maker can collect high quality information from hybrid crowds, and learns the expertise of agents.  %

\end{abstract}

\section{Introduction}
Information elicitation is a classical, significant, and ubiquitous problem in our society. A decision maker wants to predict the probability of a random event by eliciting information about its likelihood from agents, who receive payments based on their answers and sometimes also on the outcome of the event.  For example, in the classical weather forecasting problem, tomorrow's weather is modeled by a random variable $Y\in\{0 \text{ (rainy)},1 \text{ (sunny)}\}$, and the decision maker pays agents to report their probabilities for $Y=1$. The payment an agent receives depends on his/her answers.

There is a large literature on information elicitation in statistics, economics, operations research, and computer science. Most previous work focused  on two settings: elicitation {\em with} or {\em without} verification. In the elicitation setting with verification, agents' payments are allowed to depend on the outcome of the event. Then, {\em proper scoring rules}~\citep{bickel2007some} were proposed to incentivize agents to reporting truthfully,  so that their expected payment is maximized according to their believes. 

This paper focuses on elicitation without verification, a.k.a.~{\em peer prediction}, where agents' payments must be calculated before the event happens and are often based on comparisons with other agents' reports. Many previous work assumed a {\em homogeneous crowd}, i.e., agents' types are drawn i.i.d.~from a distribution, to design truth mechanisms~\citep{miller2005eliciting, shnayder2016informed}. %

However, in practice, agents often come from a {\em hybrid crowd}, which consists of different types of agents, and neither the agents' type nor the population of any type is known.  This happens particularly in crowdsourcing, where the expertise of agents can vary significantly depending on the experiments and it is often hard to estimate the competence level of an agent when the ground truth information is unavailable~\citep{DellaPenna:2012,8372965}.

\begin{ex}[\bf Information elicitation from a hybrid crowd]\label{ex:motivation}
In the weather forecasting problem, agents with most expertise (referred to as {\em experts}) may  receive a noisy {\em continuous} signal about $Y$, for example the probability for $Y=1$. Agents with less expertise (referred to as {\em non-experts} in this paper) may receive a noisy {\em discrete} signal, i.e.~whether it is more likely that 
$Y=1$ than $Y=0$, but not the values of their likelihood.  Suppose the decision maker knows neither the type of any agent, nor the population of each type of agents. How should the decision maker incentivize experts and non-experts to report their information truthfully?
\end{ex}

Motivated by Example~\ref{ex:motivation}, we ask the following question in this paper.

{\bf  \hfill How can we design   truthful mechanisms to elicit information from hybrid crowds? \hfill}

Despite a large body of literature in information elicitation from homogeneous crowds, we argue that mechanisms designed for homogeneous crowds do not work well for hybrid crowds. To see this, let us examine  two natural  ideas and briefly discuss why they are suboptimal.

The first natural and often-applied idea is to simply ignore the heterogeneity  and apply a mechanism  designed for homogeneous crowd. We see two problems with this idea, even when all agents are sincere. First, the mechanism may become hard to use for some agents, and may not be truthful anymore.  
For example, suppose  a truthful mechanism for continuous signals is used, e.g.~{the Common Ground Mechanism}~\citep{kong2018water}, in the context of Example~\ref{ex:motivation}. This works well for the experts, but even sincere non-experts would find it hard to use, because  it is  unclear how can they easily ``translate'' their discrete signals to continuous reports. Second, information might be lost, which leads to inefficient predictions.  For example, suppose a truthful mechanism for discrete signals is used, e.g.~the {\em Shifted Peer Prediction Mechanism (SPPM)}~\citep{kong2019information}. Then, the experts are forced to translate their continuous signals to discrete reports, and information is lost during the transition. When the agents are strategic, it is unclear how truthfulness should be defined, especially for non-experts. 

The second natural idea is to let the agents choose their favorite type of signals (i.e.~continuous or discrete) to report. This works well for truthful agents and can lead to more accurate predictions as shown in our MTurk experiments on 2020 US presidential election in Section~\ref{appsec:mturk}. However, strategic agents may still have incentive to misreport their preferences, and the analysis is more challenging than that for homogeneous crowds, because an agent can now choose to report a signal of any type. For example, a non-expert might be motivated to report a continuous signal for a higher payoff, and an expert might be motivated to report a  discrete signal for higher payoff and for better privacy. Technically, this cannot be handled well by modeling it as a homogeneous crowd, where the type space is the union of all agents' type space, because the population of different types of agents is unknown.

\subsection{Our Contributions}
\noindent{\bf The conceptual contribution of this paper} is a framework for information elicitation from hybrid crowds, where different types of agents receive different types of signals, and the population of any type of agent is unknown. Our setting is similar to the second idea discussed above, except that we will design mechanisms so that agents are incentivized to truthfully report their signals (and the types of signals).  See Appendix~\ref{appsec:mturk} for illustrative MTurk experiments on hybrid crowds. We also  propose the {\emph{\newtruthfulness}} criterion (Definition~\ref{dfn:truthfulness}), which states that reporting truthful signals is a Bayesian Nash Equilibrium (BNE), and the incentives are  strict unless a misreport leads to the same posterior distribution. This means that any  misreport would not   lead to a strictly higher expected payoff compared to the truthful report, and when the expected payoffs are the same, the misreport contains the same amount of information in  the posterior distribution.

\noindent{\bf The technical contributions of this paper} are two \newtruthful{}  mechanisms for information elicitation from hybrid crowds. The first mechanisms combines peer prediction mechanisms for different types of agents via linear transformations (Section~\ref{sec:comb}), and the second is  based on mutual information (Section~\ref{sec:mutual}).  While the main contribution of this paper is theoretical, numerical verifications of the proposed mechanisms on synthetic data can be found in Appendix~\ref{sec:exp}. %

\subsection{\bf Related Work and Discussions}  
\label{sec:related-work}
{\bf Information elicitation without verification.} There is a large body of literature in truthful information elicitation {without verification}, a.k.a.~{peer prediction}, in which the outcome of the event is not observed for calculating agents' payments. This is in sharp contrast to information elicitation {with verification}, in particular the literature in proper scoring rules~\citep{bickel2007some}.
Since the seminal work of \citet{miller2005eliciting}, a series of works have focused on designing truthful peer prediction mechanisms. For example, \citet{prelec2004bayesian, witkowski2012peer} relaxed the common prior assumption by asking extra questions about beliefs. \citet{radanovic2014incentives} proposed a mechanism for continuous private signals.  \citet{shnayder2016informed} proposed a new criterion of {\em informed truthfulness} in multi-task peer prediction mechanism, where truthful Bayesian Nash Equilibrium provides the highest payment. \citet{kong2018water} revealed a natural connection between learning problem and peer prediction, and designed truthful mechanisms to elicit forecasts. %
These works focus  on  homogeneous crowds and do not explicitly handle the heterogeneity of agents as we do in this paper. \citet{agarwal2020peer} focuses on peer prediction with agents of heterogeneous signal distributions, and proposes a truthful mechanisms by clustering agents based on their report behavior. However, their work still restricts in a single signal space while our work assumes multiple signal spaces.

{\bf \Newtruthfulness{}.}   Clearly, \newtruthfulness{} is stronger than standard truthfulness and weaker than strict truthfulness. We argue that a \newtruthful{} mechanism is not significantly weaker than a strictly truthful mechanism for information elicitation. Notice that a \newtruthful{} mechanism is strictly truthful when  agents' reports are converted to posterior distributions. Therefore, if the posterior distributions of agents' reports contain all information for aggregation and decision-making, such as in maximum a posteriori (MAP) estimators, then agents have strict incentive to report their truthful information contained in a signal that may or may not be the same as the type of signal she received.

{\bf Crowdsourcing.}   Our work is also related to the literature in crowdsourcing on estimating agents' competence and eliciting agents' effort, where various statistical algorithms are designed~\citep{Whitehill2009WhoseVS,10.5555/2540128.2540496,10.2307/2346806,10.1007/978-3-642-20161-5_17,BABA20142678,DUAN20145723, chen2014eliciting, witkowski2013dwelling}. \citet{10.5555/2540128.2540496} proposed a mechanism where workers report a self-confidence level together with the answer. \citet{witkowski2013dwelling} designed a mechanism to incentivize agents with high competence to participate and low competence to leave. \citet{cai2015optimum} proposed an optimum mechanism to incentivize effort, collecting high-quality information with low cost.  Practically, crowdsourcing platforms such as Amazon Mechanical Turk and CrowdFlower provide reputation  systems to assess the quality of the crowd worker. These platforms do not incentivize workers to provide truthful information. Moreover, they ask all  agents to report the same type of signals, and therefore fail to distinguish agents of different signal types. Our work provides theoretical foundations to  incentivize  workers from hybrid crowds  to provide truthful information via appropriate user interfaces corresponding to the types of signals they receive. 

{\bf Other applications of truthful information elicitation.}   
Truthful information elicitation has a wide range of applications in computer science and economics. For example, information elicitation has been extensively considered in machine learning recently:   \citet{liu2020peer} introduced a {\em peer loss function } based on peer prediction approach. Learning with peer loss function can learn from noisy labels without knowing the noise rates. \citet{liu2020online} use peer prediction techniques on a special case in online learning lacking direct feedback on loss. \citet{hu2018inference} considered collecting high quality labels via truthful information elicitation, and proposed inference aided reinforcement mechanism to incentivize agents. 
Another important application of truthful information elicitation is {\em prediction markets}~\citep{Hanson03:Combinatorial}, which collect and aggregate agents' believes about an event by allowing them to trade in contracts whose payments depend on the outcome of the event. Large literature has focused on eliciting high quality information via prediction markets~\citep{frongillo2018bounded, anunrojwong2019computing, kong2018optimizing,abernethy2013efficient, frongillo2012interpreting, frongillo2015convergence,waggoner2015market}. In particular,  \citet{abernethy2013efficient} proposed a framework for designing prediction markets via convex optimization. \citet{kong2018optimizing} studied the optimal strategy of agents in a simplified prediction market model. \citet{anunrojwong2019computing} extended the result into some more complicated models. \citet{waggoner2015market} proposed a prediction market that can protect the differential privacy of participants, and \citet{frongillo2018bounded} bounded the financial cost for the private market.

\section{Preliminaries}
\label{sec:prelim}
{

For any positive integer $n$, let $[n]=\{1,2,\cdots ,n\}$. %
We focus on the binary predicting problem in this paper, where  $Y\in\{0,1\}$ denotes the binary random variable that represents the ground truth. Let $P=\Pr[Y=1]$ denote the prior probability for $Y$ to be $1$. Each agent $j$ receives a private signal $\theta_j$ from the {\em signal space} $\Theta$, and chooses a report from the {\em report space} $\Omega$. Assuming a joint distribution over all agents' signals and the ground truth, we let $S_j$ denote the random variable representing the private signal agent $j$ receives. 
A {\em peer prediction mechanism} ${\cal M}: \Omega^n \to \bbR ^n$ then maps all $n$ agents' reports, denoted by $\vec\omega\in\Omega^n$, to the payments they receive. Often, ${\cal M}$ is a randomized mechanism, whose range becomes distributions over payment vectors. 

For example, the {\em Common Ground Mechanism (CGM)}~\citep{kong2018water} is a peer prediction mechanism for two agents, Alice and Bob.  Suppose the prior distribution (i.e., $P$) is common knowledge. Alice receives a private signal $\theta_A$ from a finite signal space $\Theta_A$, and Bob receives a private signal $\theta_B$ from a finite signal space $\Theta_B$. The report space is $\Omega = [0,1]$, and Alice and Bob are asked to report their posterior distributions on $Y$  conditioning on their private signals. For simplicity, let $p_A^*=\Pr[Y=1|S_A=\theta_A]$ and $p_B^*=\Pr[Y=1|S_B=\theta_B]$ denote the truthful reports for Alice and Bob respectively, and let $ {p}_A$ and $ {p}_B$ denote the actual reports by Alice and Bob. The payments to both Alice and Bob under CGM are the same, and are calculated as follows:

$\hfill R_\text{CGM}({p}_A,{p}_B)= \log\left(\frac{{p}_A\cdot {p}_B}{P}+\frac{(1-{p}_A)(1-{p}_B)}{1-P}\right)\hfill$

As another example, the {\em Shifted Peer Prediction Mechanism (SPPM)}~\citep{kong2019information} is a peer prediction mechanism for discrete signal space $\Theta$, and $\Omega = \Theta$. SPPM leverages  {\em proper scoring rules} that was designed for information elicitation with verification.
A {\em scoring rule}, denote by $PS$, is a mapping $PS:\Omega\times\Delta_{\Omega}\mapsto\mathbb{R}$, where $\Delta_{\Omega}$ is the set of all probability distributions over $\Omega$. The report space is $\Omega$, that is, an agent is asked to report  a probability distribution $p\in \Delta_{\Omega}$. Suppose the event turns out to be $\omega\in \Omega$, then the agent receives payment $PS(\omega, p)$.   A scoring rule is \textit{proper} if any agent is incentivized to report her truthful belief, that is, for any $p^*\in \Delta_{\Omega}$, we have
 
$\hfill p^* \in\arg\max_{p\in \Delta_{\Omega}}{\mathbb E}_{\omega\in p^*}(PS(\omega, p))\hfill$

A scoring rule is {\em strictly proper} if $p^*$ is the unique maximizer in the equation above. For example, the commonly used {\em log scoring rule} is strictly proper. The {log scoring rule} uses the payment function  $Log(\omega, p)=\log(p(\omega))$.

We are now ready to define SPPM that works for discrete signal space $\Omega$. Let $PS$ denote a  proper scoring rule. For any $r\le n$,  let $\PSdist^r$ denote the prior distribution of $r$'s private signal. Furthermore, for any agent $r\le n$ and her report $s_r$, and for any $j\le n$, let $\PSdist^j(s_r)$ denote agent $r$'s posterior prediction of $j$'s private signal conditioning on $r$'s private signal being $s_r$. For every agent $r$, the SPPM first choose another agent $j$ \emph{uniformly at random (u.a.r.)}, and then pay agent $r$'s according to the following formula, where $s_r$ denotes agent $r$'s  report. 
$$R_\text{SPPM}\left(s_r, s_j\right)=PS\left(s_j,\PSdist^j(s_r)\right)-PS\left(s_j, \PSdist^j\right)$$ 
That is, agent $r$ is paid by the  contribution (according to $PS$) of  her signal contributes to predicting (the randomly chosen) agent $j$'s signal. Notice that while the SPPM payment can be negative, its expectation is non-negative if agent $r$ reports truthfully.

A peer prediction setting naturally leads to a Bayesian game where  agents' type spaces is $\Theta^n$, their report space is $\Omega^n$, their strategy space is $\Omega^{\Theta}$, and for any $\vec\omega\in \Omega^n$, agents' utilities are ${\cal M}(\vec\omega)$.  A peer prediction mechanism where $\Omega = \Theta$ is said to be {\em truthful} (respectively, {\em strictly truthful}), if truth reporting is a {\em Bayesian Nash Equilibrium (BNE)}  (respectively, {\em strict BNE}). Formally, $\cal M$ is truthful if for every $\vec s =(s_1,\ldots,s_j)\in \Theta^n$, every $j\le n$, and every $s_j'\in \Omega$, we have
$${\mathbb E}_{\vec s_{-j}|s_j}[{\cal M}(\vec s)]_j \ge {\mathbb E}_{\vec s_{-j}|s_j}[{\cal M}(\vec s_{-j},s_j')]_j,$$
where $\vec s_{-j} = (s_1,\ldots,s_{j-1},s_{j+1},\ldots,s_n)$ and $[{\cal M}(\vec s)]_j$ is agent $j$'s payment in ${\cal M}(\vec s)$. $\cal M$ is strictly truthful, if the inequality is always strict. Both CGM and SPPM are strictly truthful.

}

\section{Our Model}
\label{sec:model}
{
\noindent\textbf{Ground Truth and Hybrid Crowd.} In this paper we focus on binary ground truth, denoted by $Y\in\{0,1\}$.  The prior distribution over $Y$ is common knowledge, and we let $P=\Pr[Y=1]$. 
In our model, there are $\noninum+1$ types of agents: the {\em experts}, and {\em group $l$ non-experts} ($l=1,2,\cdots, \noninum)$. 

\noindent\textbf{Private Signals.} The signal space is $\Theta = \expset \cup \noniset{1}\cup\cdots\cup \noniset{n}$. Different types of agents receive different types of signals defined as follows. Within each type, each agent's signal is i.i.d.~generated conditioning on the ground truth. The process is illustrated in Figure~\ref{fig:signals} and formally defined as follows.

\begin{figure}[htp]
\centering
\includegraphics[width=\linewidth]{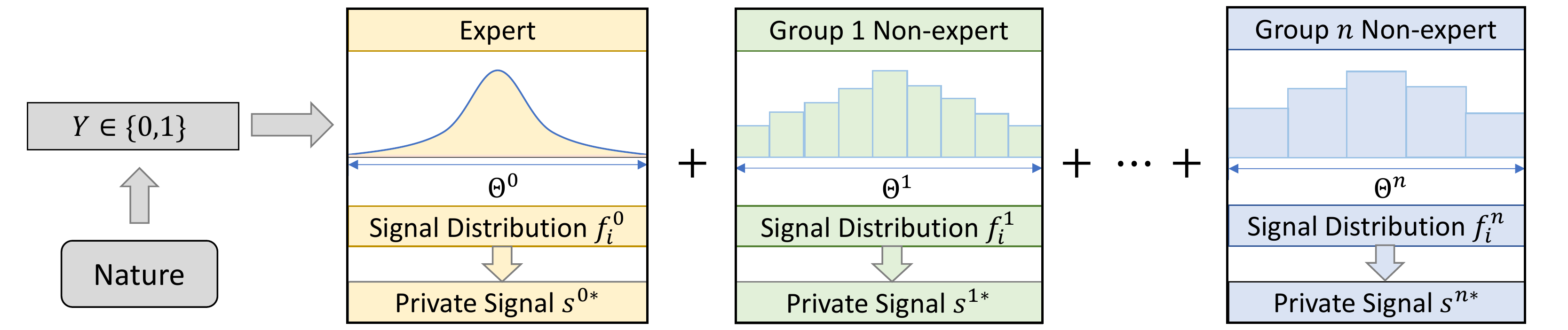}
\caption{Illustration of our model.}
\label{fig:signals}
\end{figure}

\noindent $\bullet$ {\bf Experts} receive cardinal signals, i.e., $\expset$ is an interval in $\bbR$. We let $\expsig\in\expset$ denote the private signal received by an expert that is represented by a random variable  $\expSig$. Let $f_0^0$ and $f_1^0$  denote the PDFs of $\expSig$ given $Y=0$ and $1$,  respectively. 

 \noindent $\bullet$ {\bf  Group $l$ Non-experts} $\left(l\in[\noninum]\right)$ receive ordinal signals. Let $\noniset{l}=\{\noniele{l}_1,\noniele{l}_2,\cdots, \noniele{l}_{\nonisetnum{l}}\}$ (for some $\nonisetnum{l}\in\mathbb N$) denote their signal space and let $\nonisig{l}\in\noniset{l}$ denote a private signal.  Let $\noniSig{l}$ denote the random variable of a group $l$ non-expert's signal, and let  $\nonidist{l}_0$ and $\nonidist{l}_1$ denote the PMFs of $\noniSig{l}$ given $Y=0$ and $Y=1$, respectively. \footnote{In this paper, the naming of experts vs.~non-experts reflects agents' expertise in receiving finer signals instead of their competence in providing accurate information. This is in accordance with many real-world scenarios, where a bad expert (e.g., a doctor) may provide less accurate information than a good non-expert (e.g., a mom).}

\begin{ex}
\label{ex:sig} In the setting of Example~\ref{ex:motivation}, we assume the uniform prior, i.e., %
$P=\Pr[Y=1]=0.5$.  Suppose $n=2$, i.e., the hybrid crowd consists of one group of experts and two groups of non-experts.

{\bf \emph{Experts.}} Let $\expset=\bbR$, $f_0^0=N(0, 4)$ (Gaussian distribution with mean 0 and variance 4) and $f_1^0= N(1, 4)$.  

{\bf \emph{Group 1 non-expert.}} $\noniset{1}=\{\text{Very likely rainy}, \text{Likely rainy}, \text{Likely sunny}, \text{Very likely sunny}\}$. That is, $m^1=4$. For their signal distributions, we first generate an expert's signal and then check which of the following four interval it falls into:  $\{(-\infty, 0],(0,0.5],(0.5,1],(1,\infty)\}$, and let the interval be the group 1 non-expert's signal.
It follows that  $\nonidist{1}_0=[0.50, 0.05,0.05,0.40]$ and $\nonidist{1}_1=[0.40,0.05,0.05,0.50]$.

{\bf \emph{Group 2 non-expert.}} $\noniset{2}=\{\text{Likely rainy, Likely sunny}\}$. That is, $m^2=2$.  Their signals are generated similarly except that we use a more coarse set of two intervals: $\{(-\infty, 0.5],(0.5,\infty)\}$. It follows that $\nonidist{2}_0=[0.55,0.45]$ and $\nonidist{2}_1=[0.45,0.55]$.
\end{ex}

\noindent {\textbf{Remarks on Notation.}} We often use superscript to denote the group number (where the experts correspond to group $0$). For example,  $ \expset$ represents the experts' signal space, and for any $1\le l\le n$, $\noniset{l}$ represents the signal space of group $l$ non-expert.  For signals, we use star on superscript to denote the (true) private signal. For example,  $\nonisig{l}$ represents the signal received by a group $l$ non-expert.

\noindent\textbf{Reports.}
Agents can choose a report in $\Omega=\Theta=\bigcup_{l=0}^{\noninum}\noniset{l} $. For any $0\le l\le n$,  we  let $\nonirep{l}$ denote a report in $\noniset{l}$. Notice that the agent's signal may come from a different signal space. Given agents' reports $\vec s$, we let $\nonirepset{l}\subseteq \mathbb N$ denote the set of agents whose reports are in $\noniset{l}$ where $0\le l\le\noninum$.

\noindent\textbf{Posterior Distributions.} 
Let $\exppost(\expele) \triangleq \Pr[Y=1\mid \expSig=\expele]$ denote the posterior distribution over the ground truth given $ \expele\in \expset$, as if $\expele$ is the (truthful) signal received by an expert. Similar, let $ \nonipost(s^{l})$ denote the 
posterior distribution over the ground truth given a group $l$ non-expert's  signal 
$ s^{l}$.

\noindent\textbf{Common Knowledge.} We assume the mechanism, the prior distribution, i.e.~$(1-P,P)$, and the signal distributions, i.e.~$\nonidist{l}_j$, for $l\in\{0,1,\ldots,n\}, j\in\{0,1\}$, are common knowledge. Agents' private signals are only known to themselves.

\noindent\textbf{Additional Assumptions.}
Throughout the paper, we make the following assumptions on the model. 
The first assumption states that agents' private signals are independent given the ground truth.
\begin{asm}
\label{asm:ci}
Condition on ground truth $Y$, the private signals of all  agents are independent. 
\end{asm}

The second assumption states that agents who receive different private signals have different predictions on others' private signal. 
\begin{asm}
\label{asm:ip}
For any agent $r$ and $j$, the posterior distribution of $j$'s private signal conditional on $r$'s private signal is different for different $r$'s private signal. See Appendix~\ref{sec:ip} for the formal definition.
\end{asm}
The  two assumption are  standard in the literature. For example, the first was made in \citep{kong2018water}, and the second was made in \citep{miller2005eliciting,kong2019information}.

The third assumption below has two parts. The first part guarantees that PDFs of experts' signal distributions are strictly positive on $\expset$, which we believe to be quite mild. The second part requires that the experts' signals are ``richer'' than non-experts in the sense that for any group $l$ non-expert's signal, there exists an expert's signal with the same posterior distribution. 
\begin{asm}
\label{asm:wise}
The experts' signal satisfies the following assumptions:
\begin{enumerate}
    \item for any $\expele \in \expset$ and any $i\in \{0,1\}$, we have $f_i^0(\expele)>0$;
    \item for any $l\in [\noninum]$ and any $ s^{l}\in\noniset{l}$, there exists $ \expele \in \expset$, s.t. $\exppost(\expele) = \nonipost(s^{l})$.  
\end{enumerate}
\end{asm}
We believe the second part of Assumption~\ref{asm:wise}, while being quite technical and primarily used in the proof of the theorem in Section~\ref{sec:comb}, is also very mild, because it does not put any constraints on shapes of the distributions---for example it does not require the experts' signals to be more accurate.

\textbf{Truthfulness.} One main goal of our mechanism is to elicit truthful information from agents, who can report any signal in $\bigcup_{l=0}^{\noninum}\noniset{l}$. In our model,  truthfulness means that an agent is incentivized to report the signal he/she received, which reveals which group he/she is in. In this paper, we introduce the following new notion of truthfulness, called  \textbf{\newtruthfulness}, which is stronger than standard truthfulness and slightly weaker than strict truthfulness. 

\begin{dfn}[\bf \Newtruthfulness]
\label{dfn:truthfulness}
 A mechanism $\mathcal{M}$  is  {\newtruthful} if
(1) truth reporting is a BNE, and (2) every agent has a strictly higher expected payment for truthful report than another report $s$ if $s$ has a different posterior   distribution over the ground truth. 
\end{dfn}
As discussed in Section~\ref{sec:related-work}, we believe that a \newtruthful{} mechanism is not significantly weaker than a strictly truthful mechanism w.r.t.~information elicitation. Notice that a \newtruthful{} mechanism is strictly truthful when  agents' reports are converted to posterior distributions. Therefore, if the posterior distributions of agents' reports contain all information for aggregation and decision-making, such as in maximum a posteriori (MAP) estimators, then agents have strict incentive to report their truthful information (though the reported signals might be different from the signals they receive).

}

\section{Composite Elicitation Mechanisms (CEM)}
\label{sec:comb}
{
In this section, we introduce composite elicitation mechanism for the model presented in Section~\ref{sec:model}, which combines two  information elicitation mechanisms:  ECGM (which is a natural extension of CGM for experts defined below) and SPPM with log scoring rule (one for each group of non-experts) via a set of  linear transformations, to achieve \newtruthfulness.

ECGM is a natural extension of  CGM   for two agents to more than two agents. 
Suppose agent $r$ reports  $\exprep_r$. ECGM first chooses another agent $j$ \emph{u.a.r} and then pay agent $r$ according to the payment function in CGM. More precisely, the payment to agent $r$ is:
\begin{equation}
\label{eq:ecg}
    R_\text{ECGM}(\exprep_r,\exprep_j)= \log\left(\frac{\exppost(\exprep_r)\cdot \exppost(\exprep_j)}{P}+\frac{(1-\exppost(\exprep_r))(1-\exppost(\exprep_j))}{1-P}\right)
\end{equation}
where we recall that $\exppost(\exprep_r)\triangleq\Pr[Y=1|\expSig_r=\exprep_r]$.

Like CGM, ECGM is also strictly truthful as shown in Proposition~\ref{thm:ecgm} in Appendix~\ref{proofecg}.   We are now ready to define   \emph{composite elicitation mechanisms (CEM)}.  
\begin{dfn}
\label{dfn:CEM} Given a hybrid crowd model in Section~\ref{sec:model} and a set of real numbers $\{ \noniA{l}_k, \noniB{l}_k:l\in [\noninum], 1\le k\le \nonisetnum{l} \}$, the \emph{composite elicitation mechanisms (CEM)} is defined as follows.

\noindent{\bf \boldmath A $\expset$-report agent} (with report $\exprep_r$) is paid  $R_\text{ECGM}(\exprep_r, \exprep_j),$ where $j$ is an agent sampled from $\exprepset\setminus\{r\}$   u.a.r.

\noindent{\bf \boldmath  A $\noniset{l}$-report agent} (with report $\nonirep{l}_r=\noniele{l}_k$) is paid $\noniA{l}_k\cdot R_\text{SPPM}^{l}(\nonirep{l}_r, \nonirep{l}_j)+\noniB{l}_k,$
where $j$ is an agent sampled from $\nonirepset{l}\setminus\{r\}$ u.a.r.
\end{dfn}
That is, a $\expset$-report agent  is paid her ECGM payment with another randomly selected $\expset$-report agent. (As a common assumption,  there are at least 2 reporters in each report set.)  A $\noniset{l}$-report agent  is paid a linear transformation of her SPPM  payment $R_\text{SPPM}^{l}$ with log scoring rule\footnote{In fact, all $PS$ satisfying Assumption~\ref{asm:ps} (See Appendix~\ref{sec:ps}), including the log scoring rule, can be used in CEM to make it \newtruthful.} %
with another randomly selected $\noniset{l}$-report agent. An example of CEM  can be found in  Appendix~\ref{sec:runexp}. 

The main theorem of this section is the following.
\begin{thm}[\bf \Newtruthfulness{} of CEM]
\label{thm:composite}
Under Assumptions~\ref{asm:ci}--\ref{asm:wise}, there exists a set of linear transformation coefficients $\left(\noniA{l}_k,\noniB{l}_k\right)_{l\in[\noninum], k\in[\nonisetnum{l}]}$, under which CEM is \newtruthful. 
\end{thm}
Although Theorem~\ref{thm:composite} only claims the existence of coefficients $\left(\noniA{l}_k,\noniB{l}_k\right)_{l\in[\noninum], k\in[\nonisetnum{l}]}$, the coefficients can be efficiently calculated in advance involving some integration. The explicit form of the coefficient (Equation~\ref{eq:noniA} and~\ref{eq:noniB}) can be found in Appendix~\ref{sec: proofce}.  

\noindent\textbf{Proof Sketch.}
To prove \newtruthfulness, it suffices prove the following two types of truthfulness for every agent, given that all other agents report truthfully.

 \textbf{1.~Interior truthfulness.} If the agent's report is in the same $\Theta^l$ as her private signal, her expected payment is uniquely maximized when she reports truthfully.  

\textbf{2.~Exterior Truthfulness.} If the agent's report is in a different $\Theta^l$, her expected payment will not exceed the payment under truthful report (and equal only for signals with same posterior distribution).

The proof proceeds with giving explicit  formulas of the  coefficients $\{\noniA{l}_k, \noniB{l}_k\}$, then proving that  both interior and exterior truthfulness hold for all agents.
Intuitively, the linear transformation coefficients is generated from and uniquely determined by the exterior truthfulness among expert group and non-expert groups (see Figure~\ref{fig:AB} and explanation below). On the other hand, we can prove that, with these coefficients, interior truthfulness in one non-expert group, and exterior truthfulness among different non-expert groups can be satisfied. Interior truthfulness of the expert group naturally comes from the truthfulness of ECGM. 

The truthfulness will be achieved by taking advantage of the second part of Assumption~\ref{asm:wise} to construct some special cases of signal, with which we bridge the gap between experts and non-experts.  For any $\noniele{l}_k\in \Theta^l$, let $\expelek{l}{k}\in\Theta^0$  denote a signal that is guaranteed by Assumption~\ref{asm:wise} (2), that is, $\exppost(\expelek{l}{k})=\nonipost(\noniele{l}_k)$, which is used to construct the linear coefficients. The proof relies on calculating and comparing expected payments of different groups of agents. While it is hard to directly compare expected payment of different groups, $\expelek{l}{k}$ provides a special case, where different expected payments will be equal or comparable to each other. In this way, we can convert all different expected payments into one group and compare them.

\begin{minipage}{0.4\textwidth} We now present the high-level intuition behind the linear transformation in Figure~\ref{fig:AB}.  The figure shows relationship between different kinds of expected payments. The X-axis is the posterior distribution of expert's signal. Y-axis is the value of expected payment. We draw the following expected payment in the figure. %

{\bf \boldmath Curve 1: $\exppay_\text{exp}(\expele_r,\expele_r)$.} The curve represents an expert's expected payment when reporting truthfully. 

{\bf \boldmath Line 1:  $\nonipay{l}_\text{exp}(\expele_r, \noniele{l}_k)$.} The line represents an expert's expected payment when reporting $\noniset{l}$ signal $\noniele{l}_k$ (before linear transformation).  

{\bf \boldmath Curve 2: $\exppay_{\text{non-}l}(\noniele{l}_k, \expele_r)$.} The curve represents a non-expert's expected payment when reporting $\expset$ signal $s_r$. 

{\bf \boldmath  Line 2:   $\nonipay{l}_{\text{non-}l}(\noniele{l}_k, \noniele{l}_k)$.}  The line represents a non-expert's expected payment when reporting truthfully (before linear transformation). 
\end{minipage}
\begin{minipage}{0.6\textwidth}
    \centering
    \captionsetup{type=figure}
    \includegraphics[width=\textwidth]{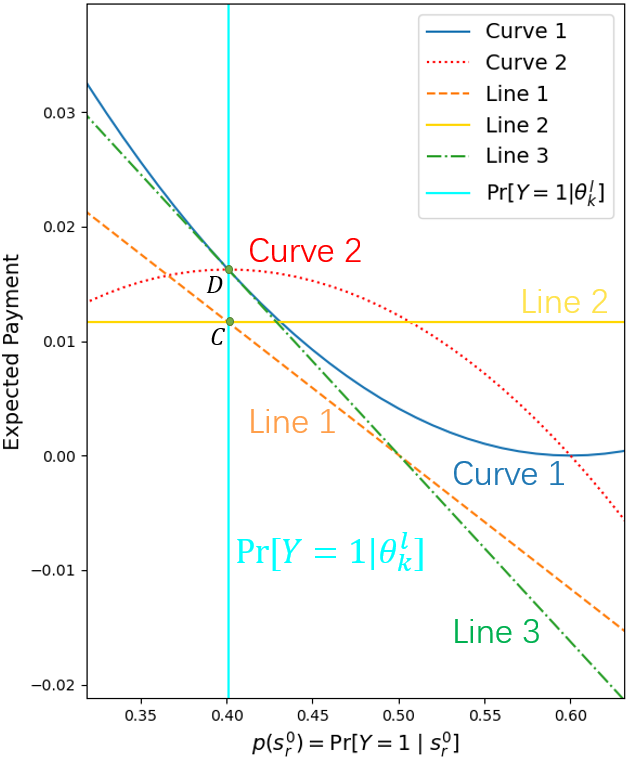}
    \captionof{figure}{Intuition of linear transformation.\label{fig:AB}}
\end{minipage}

The linear transformation should satisfy the following two conditions (which is the exterior truthfulness between expert group and non-expert group): 
\begin{enumerate}
    \item \textbf{Expert's exterior truthfulness towards Non-expert}:\\ $\exppay_{exp}(\expele_r,\expele_r) \ge \noniA{l}_k \cdot \nonipay{l}_{exp}(\expele_r, \noniele{l}_k) + \noniB{l}_k$. The expected payment of expert when report truthfully (Curve 1) should exceed her expected (transformed) payment when reporting $\noniele{l}_k$ (Line 1). 
    
    \item \textbf{Non-expert's exterior truthfulness towards expert}:\\ $\noniA{l}_k\cdot \nonipay{l}_{non-l}(\noniele{l}_k, \noniele{l}_k) + \noniB{l}_k\ge \exppay_{non-l}(\noniele{l}_k, \expele_r).$ The (transformed) expected payment of non-expert when report truthfully (Line 2) should exceed her expected (transformed) payment when reporting $\expele_r$ (Curve 2). 
\end{enumerate}

Meanwhile,  Line 1 and Line 2 are transformed using the same $\noniA{l}_k$ and $\noniB{l}_k$.

Now we look at point $C$, the intersection of Line 1 and Line 2. Note that $C$ lies on $\exppost(\expsig)=\Pr[Y=1\mid\noniele{l}_k]$ (the vertical line). Since Line 1 and Line 2 share the same $\noniA{l}_k$ and $\noniB{l}_k$, the transformed $C$ always lies on the vertical line no matter how $\noniA{l}_k$ and $\noniB{l}_k$ changes. Related to the restrictions (Curve 1 above Line 1, Curve 2 below Line 2), the only possible point for transformed $C$ is point $D$, the intersection of Curve 1 and Curve 2. Moreover, since Line 1 after transformation must lie below Curve 1, it must be tangent to the Curve 1 on $D$. Therefore, Line 1 is transformed the  Line 3. Note that since Line 3 is unique, the linear transformation coefficient is uniquely determined. 

Although the exterior truthfulness among expert group and non-expert groups has already determined the coefficients, we can still prove that interior truthfulness in one non-expert group and exterior truthfulness among different non-expert groups hold under this coefficients. The full proof can be found in Appendix~\ref{sec: proofce}.

}

\section{Mutual-Information-Based Mechanism}
\label{sec:mutual}
In this section, we introduce Mutual-Information-Based Mechanism (MIBM) and prove that it is also \newtruthful. Again, the key is to define its payment function.

\noindent\textbf{Payment Function for MIBM.} Recall that $\nonirepset{l}$ are the set of agents whose reports are in $\noniset{l}$ where $0\le l\le\noninum$. For each $\nonirepset{l}$, the payment function has $\noninum+1$ classes of non-negative coefficients: $\alpha^{l}_{i}\left(t_0,t_1,\cdots,t_{\noninum}\right)$, where $0\le i\le\noninum$ and $t_0,t_1\cdots,t_n\in \mathbb N$ represent $|\nonirepset{0}|,|\nonirepset{1}|,\cdots,|\nonirepset{n}|$ respectively.
\begin{ex}\label{ex:MIBM} In this example, all coefficients are the same, i.e., for every $0\le l,i\le\noninum$ and every $t_0,t_1,\cdots,t_\noninum\in\mathbb N$,  we let

$\hfill\alpha^{l}_{i}\left(t_0,t_1,\cdots,t_\noninum\right)=1
\hfill$
\end{ex}

Let $\PSI{i},0\le i\le\noninum$ be any strictly proper scoring rules. For any $0\le l_1,l_2\le\noninum$, let
$$R_{l_1l_2}(\nonirep{l_1},\nonirep{l_2})=PS^{l_2}(\nonirep{l_2}, \SSIdist{l_1}(\nonirep{l_1})) - PS^{l_2}(\nonirep{l_2}, \SSIdist{l_2}),$$
where $\SSIdist{l_1}(\nonirep{l_1})$ is the posterior signal distribution conditioning on signal $\nonirep{l_1}\in \noniset{l_1}$ and $\SSIdist{l_2}$ is the prior signal distribution on $\noniset{l_2}$. This function is inspired by Bregman Mutual Information~\citep{kong2019information}.

For each agent $r$ in $\nonirepset{l}$, we pick an agent $j^l\in\nonirepset{l}\setminus\{r\}$ u.a.r., and pick agents $j^i\in\nonirepset{i},0\le i\neq l\le\noninum$ u.a.r. The payment function for agent $r$ is 
\begin{align*}
    \sum_{i=0}^{\noninum}\alpha^{l}_i\left(|\nonirepset{0}|,\cdots,|\nonirepset{\noninum}|\right) R_{li}(\nonirep{l}_r,\nonirep{i}_{j^i})
\end{align*}

\begin{ex}
Assume we have only 1 group of non-expert, i.e. $\noninum=1$. If we use log scoring rule for $\PSI{0}$ and $\PSI{1}$ and the coefficients in example \ref{ex:MIBM}. The payment function for expert $r$ is
$$\log\left(\frac{\pdfs[\nonirep{0}_{j^0}\mid \noniSig{0}_r=\nonirep{0}_r]}{\pdfs[\nonirep{0}_{j^0}]}\right)+\log\left(\frac{\Pr[\noniSig{1}_{j^1}=\nonirep{1}_{j^1}\mid \noniSig{0}_r=\nonirep{0}_r]}{\Pr[\noniSig{1}_{j^1}=\nonirep{1}_{j^1}]}\right)$$
where $j^0$ is u.a.r. selected in $\nonirepset{0}\setminus\{r\}$, and $j^1$ is u.a.r. selected in $\nonirepset{1}$. Note that $\pdfs[\cdot]$ and $\pdfs[\cdot\mid S^0_r=s^0_r]$ are the PDF of $S^0$ and the PDF of $S^0$ conditioning on $S^0_r=s^0_r$ respectively.
\end{ex}

Another example of MIBM payments can be found in Appendix~\ref{sec:runexp}.

The main theorem of this section provides a sufficient conditions on the coefficients for MIBM to be \newtruthful.

\begin{thm}[\bf \Newtruthfulness{} of MIBM]
\label{thm:mi}
Given Assumption~\ref{asm:ci} and \ref{asm:ip}. 
Assume $|\nonirepset{l}|\ge3$ for all $0\le l\le n$.
If for $0\le l_1<l_2\le n$ and $0\le i\le n$, the following conditions holds for $t_0\ge2,\cdots,t_n\ge 2$:
$$ \alpha^{l_1}_i(t_0,\cdots,t_n) = \alpha^{l_2}_i(t_0,\cdots,t_{l_1}-1, \cdots,t_{l_2}+1,\cdots,t_n)$$
then MIBM is \newtruthful.
\end{thm}
It is not hard to verify that the coefficients in Example~\ref{ex:MIBM} satisfy the conditions in Theorem~\ref{thm:mi}.  The conditions may appear enigmatic, but we believe that they have intuitive meanings. Take $l_1=0,l_2=1$, i.e.~$\alpha^0_i(t_0,t_1,\cdots,t_n)=\alpha^1_i(t_0-1,t_1+1,\cdots,t_n)$, for example. The equation is used to guarantee that no expert has incentive to submit a report in $\noniset{1}$. And vice versa, no group 1 non-expert has incentive to submit a report in $\noniset{0}$. The $t_0-1$ and $t_1+1$ part in the right hand side   represents the hypothetical deviations.

\noindent\textbf{Proof Sketch.} Like the proof of Theorem~\ref{thm:composite}, we will prove  interior truthfulness and exterior truthfulness for MIBM whose coefficients satisfy the conditions described in the theorem.  The proof is done by taking an optimization-based argument. 

Take interior truthfulness for example. According to Assumption~$\ref{asm:ip}$ and the strictness of strictly proper scoring rule, we can verify that  the expected payment of an agent is uniquely maximized when she reports truthfully. Notice that under SPPM~\citep{kong2019information}, an agent's expected payoff can be viewed as the Bregman mutual information between her signal distribution and her peer's signal distribution.  However, this relationship does not hold anymore in MIBM. 
The full proof can be found in Appendix \ref{proofmi}.

\section{Conclusions and Future Work}
We formulate the problem of information elicitation from hybrid crowds and propose two \newtruthful{} mechanisms. Perhaps the most important open question is whether we can design strictly truthful mechanisms for hybrid crowds. Theoretical justifications of such mechanisms in terms accuracy/sample complexity, and real-world experiments that examines the strategic aspects of the agents are also promising and challenging directions.

\bibliographystyle{named}
\bibliography{references}

\clearpage
\appendix

\section{The Formal Definition and Properties of KL-divergence}\label{sec:KL}
\begin{dfn}[KL-Divergence]
For any finite set $\Sigma$. The KL-divergence $D_{KL}:\Delta_{\Sigma}\times\Delta_{\Sigma}\mapsto \mathbb{R}$ is a non-symmetric measure of the difference between distribution $\mathbf{P}\in\Delta_{\Sigma}$ and $\mathbf{Q}\in \Delta_{\Sigma}$. It is defined as 
$$ D_{KL}(\mathbf{P}\|\mathbf{Q})=\sum_{\sigma\in\Sigma} P(\sigma)\log\left(\frac{P(\sigma)}{Q(\sigma)}\right),$$
where $P(\sigma)$ and $Q(\sigma)$ is the probability function of distributions $\mathbf{P}$ and $\mathbf{Q}$. 

Similarly, for any continuous interval $I
 \subset \bbR$, when $\mathbf{P},\mathbf{Q}\in\Delta_{I}$, the KL-divergence is defined as
$$ D_{KL}(\mathbf{P}\|\mathbf{Q})=\int_{I} p(x)\log\left(\frac{p(x)}{q(x)}\right){\rm d}x, $$
where $p(x)$ and $q(x)$ are the PDFs of distributions $\mathbf{P}$ and $\mathbf{Q}$.
\end{dfn}

KL-divergence have the property of \emph{non-negativity}. That is, $D_{KL}(\mathbf{P}\|\mathbf{Q})\ge 0$ for any two distributions $\mathbf{P}$ and $\mathbf{Q}$ defined on the same probability space. The equality holds if and only if $\mathbf{P}=\mathbf{Q}$ (in continuous space $\mathbf{P}=\mathbf{Q}$ \alev). With using the non-negativity property, we can prove the following lemma.

\begin{lemma}
\label{lem:lsr}
(Discrete Version) $\mathbf{P,Q,R}$ are distribution on finite set $\Sigma$. Then we have 
$$ \sum_{\sigma\in\Sigma} P(\sigma)\log\left(\frac{Q(\sigma)}{R(\sigma)}\right) \le \sum_{\sigma\in\Sigma} P(\sigma)\log\left(\frac{P(\sigma)}{R(\sigma)}\right).$$

(Continuous Version) $\mathbf{P,Q,R}$ are distributions of continuous random variables on interval $I$. Then we have $$ \int_{I} p(x)\log\left(\frac{q(x)}{r(x)}\right){\rm d}x \le \int_{I} p(x)\log\left(\frac{p(x)}{r(x)}\right){\rm d}x.$$

For both version, the equality holds if and only if  $\mathbf{P}=\mathbf{Q}$ (\alev). 
\end{lemma}

\begin{proof}
For the discrete version:
\begin{align*}
    &\sum_{\sigma\in\Sigma} P(\sigma)\log\left(\frac{Q(\sigma)}{R(\sigma)}\right) - \sum_{\sigma\in\Sigma} P(\sigma)\log\left(\frac{P(\sigma)}{R(\sigma)}\right)\\
    = & \sum_{\sigma\in\Sigma}P(\sigma)\left(\log\left(\frac{Q(\sigma)}{R(\sigma)}\right)-\log\left(\frac{P(\sigma)}{R(\sigma)}\right)\right)\\
    =& \sum_{\sigma\in\Sigma} P(\sigma)\log\left(\frac{Q(\sigma)}{P(\sigma)}\right)\\
    =& -D_{KL}(\mathbf{P}\|\mathbf{Q})\\
    \le & 0.
\end{align*}
The inequality comes from the non-negativity of the KL-Divergence, and the equality holds only when $\mathbf{P}=\mathbf{Q}$

For the continuous version
\begin{align*}
 &\int_{I} p(x)\log\left(\frac{q(x)}{r(x)}\right){\rm d}x - \int_{I} p(x)\log\left(\frac{p(x)}{r(x)}\right){\rm d}x\\
 =&\int_{I} p(x)\left( \log\left(\frac{q(x)}{r(x)}\right) - \log\left(\frac{p(x)}{r(x)}\right) \right){\rm d}x\\
 = & \int_{I} p(x)\log\left(\frac{q(x)}{p(x)}\right){\rm d}x\\
 = & -D_{KL}(\mathbf{P}\|\mathbf{Q})\\
 \le & 0. 
\end{align*}
The inequality comes from the non-negativity of the KL-Divergence, and the equality holds only when $\mathbf{P}=\mathbf{Q}$ \alev.
\end{proof}

\section{Proper Scoring Rules for Shifted Peer Prediction Mechanism}

\label{sec:ps}
In this section we introduce the property of the scoring rules for shifted peer prediction mechanism in our text. This setting is necessary in the proof for truthfulness of Composite Elicitation Mechanism in Appendix~\ref{sec: proofce}

To introduce the property, we first introduce a lemma of proper scoring rule from \citep{mccarthy1956measures} and \citep{savage1971elicitation}. 
For convenience, for distribution $\mathbf{P}$ and $\mathbf{Q}$ we denote 
\begin{equation*}
    PS(\mathbf{P}; \mathbf{Q}) = \mathbb{E}_{s\sim \mathbf{P}} PS(s,\mathbf{Q}),
\end{equation*}
which is the expected score of an agent when she has belief $\mathbf{P}$ to the distribution and answers $\mathbf{Q}$.
The lemma is as follows:
\begin{lemma}
\label{lemma:ps}
A scoring PS is proper if an only if there exist a convex function $G:\Delta_{\Theta}\mapsto \mathbb{R}$, such that $G(\mathbf{Q})=PS(\mathbf{Q};\mathbf{Q})$, and 
\begin{equation*}
   PS(s,\mathbf{Q}) = G(\mathbf{Q}) + dG(\mathbf{Q}) \cdot (\delta_s-\mathbf{Q}),
\end{equation*}
where $dG(\mathbf{Q})$ is a subgradient of $G$ at $\mathbf{Q}$, and $\delta_s$ is the distribution putting probability 1 on $s$. 
\end{lemma}
This lemma give an equivalent form of a proper scoring rule with a convex function and its subgradient. 

Now we give the property in the form of assumption. 
\begin{asm}[Proper Scoring Rule Assumption]
\label{asm:ps}
For every proper scoring rule $PS^l$ in the Composite Elicitation mechanism (more specifically, in$R_{sppm}^{l}$ for all $l\in[\noninum]$), let $G^l$ to be its correlated convex function.
The Hessian Matrix of $G^l(\mathbf{Q}^l)$ exists. Moreover, $\frac{\partial^2 G^l(\mathbf{Q}^l)}{\partial q^l_i\partial q^l_j} \neq 0$ if and only if $i=j$ for all $i,j=1,2,\cdots,\nonisetnum{l}$. 
\end{asm}
Note that log scoring rule and quadratic scoring rule, both satisfy this assumption. 

\section{Proof of Truthfulness Proposition of Extended Common Ground Mechanism}
\label{proofecg}
\begin{prop}\label{thm:ecgm}
The Extended Common Ground mechanism is strictly truthful.
\end{prop}

\begin{proof}
To prove the truthfulness, we show that for any agent $i$, when she expects her peer $j$ to be truthful, her expected payment will be maximized when she give a truthful report $\exprep_i=\expsig_i$. 
Consider $i$'s expected payment when she receive $\expsig_i$ and report $\exprep_i$. For convenience, we denote $\exppost(\expsig_i)=\exppost_i^*$, $\exppost(\exprep_i)=\exppost_i$ and $\exppost(\expele_j)=\exppost_j$. Then the extended common ground mechanism payment of $i$ can be written as 
\begin{equation*}
    R_{ecgm}(\exprep_i,\expele_j) = \log \left( \frac{(1-\exppost_i)(1-\exppost_j^*)}{1-P}+\frac{\exppost_i\exppost_j}{P} \right)
\end{equation*}

Then we consider the case that $j$ receive different private signals $\exprep_j=\expele_j\in I$. 
$\pdfs[\expele_j\mid \expsig_i]$ is the PDF of the post distribution of $\expele_j$ conditioning on $\expsig_i$, and $\pdfs[\expele_j]$ is the PDF of the prior distribution of $\expele_j$. $i$'s expected payment is as follows, and we use Bayesian formula (third to fourth line) to change it form. 
\begin{align*}
    \mathbb{E}[R_{ecgm}(\exprep_i,\expele_j)]=& \int_{I} \pdfs[\expele_j\mid \expsig_i] \log \left( \frac{(1-\exppost_i)(1-\exppost_j^*)}{1-P}+\frac{\exppost_i\exppost_j^*}{P} \right){\rm d}\expele_j\\
    = &\int_{I} \pdfs[\expele_j\mid \expsig_i] \log \left( \frac{(1-\exppost_i)\Pr[Y=0\mid \expele_j]}{\Pr[Y=0]}+\frac{\exppost_i\Pr[Y=1\mid \expele_j]}{\Pr[Y=1]} \right){\rm d}\expele_j\\
    = & \int_{I} \pdfs[\expele_j\mid \expsig_i] \log \left( \frac{\left(\frac{(1-\exppost_i)\Pr[Y=0\mid \expele_j]}{\Pr[Y=0]}+\frac{\exppost_i\Pr[Y=1\mid \expele_j]}{\Pr[Y=1]}\right)\pdfs[\expele_j]}{\pdfs[\expele_j]} \right){\rm d}\expele_j\\
    = &\int_{I} \pdfs[\expele_j\mid \expsig_i] \log \left( \frac{(1-\exppost_i)\cdot \pdfs[\expele_j\mid Y=0]+\exppost_i\cdot \pdfs[\expele_j\mid Y=1]}{\pdfs[\expele_j]} \right){\rm d}\expele_j\\
    = & \int_{I}\left((1-\exppost_i^*)\cdot \pdfs[\expele_j\mid Y=0]+\exppost_i^*\cdot \pdfs[\expele_j\mid Y=1]\right)\\
    &\log \left( \frac{(1-\exppost_i)\cdot \pdfs[\expele_j\mid Y=0]+\exppost_i\cdot \pdfs[\expele_j\mid Y=1]}{\pdfs[\expele_j]}\right){\rm d}\expele_j. 
\end{align*}
Then we try to rewrite the formula into the following three functions:
\begin{itemize}
    \item $f_1^*(\expele_j)=(1-\exppost_i^*)\cdot \pdfs[\expele_j\mid Y=0]+\exppost_i^*\cdot \pdfs[\expele_j\mid Y=1]$
    \item $f_2^*(\expele_j)=(1-\exppost_i)\cdot \pdfs[\expele_j\mid Y=0]+\exppost_i\cdot \pdfs[\expele_j\mid Y=1]$
    \item $f_3^*(\expele_j)=\pdfs[\expele_j]$
\end{itemize}
The the formula can be rewritten as
$$\int_{I} f_1^*(\expele_j)\log \left( \frac{f_2^*(\expele_j)}{f_3^*(\expele_j)} \right){\rm d} \expele_j.$$

Note that both $f_1^*(\expele_j)$ and  $f_2^*(\expele_j)$ can be regard as continuous distribution on $I$, and $f_3^*(\expele_j) = \pdfs[\expele_j]$ is a continuous distribution of $I$ Therefore, by the continuous version of Lemma \ref{lem:lsr}, we have
\begin{align*}
    & \int_{I} f_1^*(\expele_j)\log \left( \frac{f_2^*(\expele_j)}{f_3^*(\expele_j)} \right){\rm d} \expele_j\\
    \le & \int_{I} f_1^*(\expele_j)\log \left( \frac{f_1^*(\expele_j)}{f_3^*(\expele_j)} \right){\rm d} \expele_j\\
    = & \int_{I}\pdfs[\expele_j\mid \expsig_i]\log \left( \frac{(1-{p}_i)(1-\exppost_j^*)}{1-P}+\frac{{p}_i\exppost_j^*}{P}\right){\rm d}\expele_j,
\end{align*}

which is the expected payment of agent $i$ when she reports truthfully. Therefore, truthful reporting is an equilibrium in Extended Common Ground mechanism.

Then we consider the strictness of the mechanism. Notice that $f_1^*(\expele_j)$ and $f_2^*(\expele_j)$ are posterior distribution of $\expele_j$ conditioning on $s_i$ and $\hat{s}_i$. Therefore, with Assumption~\ref{asm:ip}, $f_1^*(\expele_j)=f_2^*(\expele_j)$ \alev only if $\exprep_i=\expsig_i$. Therefore, according to the strictness of Lemma~\ref{lem:lsr} the equality holds if and only if $\exprep_i=\expsig_i$, and the strictness holds. 

\end{proof}

\section{Formal Version of Informative Prior}
\label{sec:ip}
In this section, we give a more formal version of Assumption~\ref{asm:ip}, i.e informative prior.

\noindent\textbf{Assumption \ref{asm:ip}. }(Informative Prior) \emph{Consider an agent $i$. Here we denote the random variable of $i$'s private signal as $S_i$, and  two possible realization of  $S_i$ as $s_i$ and ${s_i}'$, no matter what $i$'s signal space is. }

\emph{
(Expert Version) For any expert $j$, and for any $s_i\neq{s_i}'$, let $q(\expsig_j\mid s_i)$ and $q(\expsig_j\mid {s_i}')$ be PDF of the posterior distribution of $j$'s private signal conditional on $i$'s private signal being $s_i$ and ${s_i}'$ separately. Then $q(\expsig_j\mid s_i)$ does not equal to $q(\expsig_j\mid {s_i}')$ \alev (i.e. the set of $\expsig_j$ such that $q(\expsig_j\mid s_i)=q(\expsig_j\mid {s_i}')$ is not a zero-measure set}. 

\emph{
(Non-expert Version) For any group $l$ non-expert $j$, and for any $s_i\neq{s_i}'$, there exist at least one $k\in [\nonisetnum{l}]$ such that $\Pr[\noniSig{l}_j=\noniele{l}_k\mid S_i=s_i] \neq \Pr[\noniSig{l}_j=\noniele{l}_k\mid S_i={s_i}']$.}

\section{Proof Main Theorem of Composite Elicitation Mechanism }
\label{sec: proofce}
In this section we give the proof of the strict truthfulness of the composite elicitation mechanism. We first give the analyses of the expected payment of different agents, which appear in the definition of the mechanism and are necessary in the proof. Then we give the proof of the theorem. 

First we introduce notations for the proof. 

For group $l$ non-experts, we give a specific notation for signal distributions:  $\nonidist{l}_i(k) = \Pr[\noniSig{l}=\noniele{l}_k\mid Y=i], i=\{0,1\}$.

Then we define notations about proper scoring rule. We use $\PSI{l}$ for the proper scoring rule of $\noniset{l}$. 

For prior and posterior distribution of agent signals, we define as follows. For a group $l$ non-expert $r$, let $\SSIdist{l} = \left(\Pr[\noniSig{l}=\nonisig{l}] \right)_{\nonisig{l}\in\noniset{l}}$ be the prior distribution of $\noniSig{l}$. 
And let $\SSIdist{l}(\nonisig{l}_r) = \left(\Pr[\noniSig{l}=\nonisig{l}\mid \noniSig{l}_r=\nonisig{l}_r] \right)_{\nonisig{l}\in\noniset{l}}$ be $r$'s posterior prediction of another agent's private signal $\noniSig{l}$ conditioning on her private signal $\noniSig{l}_r=\nonisig{l}_r$

\subsection{Agent's Expected Payment}
\label{subsec: exppay}
In the section we give the analyses of the expected payment of the agents in the composite elicitation mechanism. Since our goal of \emph{\newtruthfulness} consider the case where all experts give truthful report in $\expset$, and for all $l\in[\noninum]$ all group $l$ non-experts give truthful report in $\noniset{l}$, an agent will expected her peer to be a certain expertise of agent based on her report space, and expect her peer to be truthful. We first give the explicit form of these expected payment, and then give the analysis. 
Here $\pdfs[\expsig\mid \expSig_r=\expsig_r]$ is the PDF of $\expSig$ conditioning on $\expSig_r=\expsig_r$.
\begin{align}
    \label{eq:CCexp}
    \exppay_{exp}(\expsig_r,\exprep_r)=&\int_{\expset} \pdfs[\expsig\mid \expSig_r=\expsig_r]\cdot R_{ecgm}(\exprep_r, \expsig){\rm d}\expsig. \\
    \label{eq:SSexp}
    \nonipay{l}_{exp}(\expsig_r, \nonirep{l}_r)=&\sum_{j=1}^{\nonisetnum{l}}\Pr[\noniSig{l}=\noniele{l}_j\mid \expSig_r=\expsig_r]\cdot R_{sppm}^{l}(\nonirep{l}_r, \noniele{l}_j).\\
    \label{eq:CCsemi}
    \exppay_{non-l}(\nonisig{l}_r,\exprep_r)=&\int_{\expset} \pdfs[\expsig\mid \noniSig{l}_r=\nonisig{l}_r]\cdot R_{ecgm}(\exprep_r, \expsig){\rm d}\expsig. \\
    \label{eq:SSsemi}
    \nonipay{l}_{non-l}(\nonisig{l}_r, \nonirep{l}_r)=&\sum_{j=1}^{\nonisetnum{l}}\Pr[\noniSig{l}=\noniele{l}_j\mid \noniSig{l}_r=\nonisig{l}_r]\cdot R_{sppm}^{l}(\nonirep{l}_r, \noniele{l}_j). \\
    \label{eq:DDsemi}
    \nonipay{h}_{non-l}(\nonisig{l}_r, \nonirep{h}_r)=&\sum_{j=1}^{\nonisetnum{h}} \Pr[\noniSig{k}=\noniele{h}_j\mid \noniSig{l}_r=\nonisig{l}_r]\cdot R_{sppm}^{h}(\nonirep{h}_r, \noniele{h}_j).
\end{align}

$\exppay_{exp}$ is the expected payment of expert $r$ when she receives private signal $\expSig_r=\expsig_r$, and gives continuous report $\exprep_r$. The other expected payments are defined in the similar way. Discrete expected payments are defined without linear transformations. 
 
\subsubsection{Expert's Expected Payment in Continuous Mechanism}
We'll start from the expert's expected payment in continuous mechanism (i.e. Extended Common Ground mechanism). Consider an expert $r$ whose private signal is $\expsig_r$ and report $\exprep_r$. Her peer is an expert, has private signal $\expsig$ ,and report $\exprep=\expsig$. Then the payment of agent $r$ is:
\begin{equation*}
    R_{ecgm}(\exprep_r,\exprep) = \log\left(\frac{\exppost(\exprep_r)\cdot \exppost(\expele)}{P}+\frac{(1-\exppost(\exprep_r))(1-\exppost(\expele))}{1-P}\right).
\end{equation*}

Now we take expert $r$'s perspective. While $r$ does not know her peer's private signal $\expele$, she may calculate a posterior distribution of $\expele$ based on her own private signal $\expsig_r$, i.e.
\begin{align*}
    \pdfs[\expele\mid \expSig_r=\expsig_r] = & f_0(\expele)\Pr[Y=0\mid \expSig_r=\expsig_r]+f_1(\expele)\Pr[Y=1\mid \expSig_r=\expsig_r]\\
    = & f_0(\expele)(1-\exppost(\expsig_r))+f_1(\expele)\cdot\exppost(\expsig_r),
\end{align*}

Moreover, $\exppost(\expele)$ can be written as 
\begin{equation*}
     \exppost(\expele) =\frac{P\cdot f_1(\expele)}{(1-P)f_0(\expele)+P\cdot f_1(\expele)}
\end{equation*}
Therefore, considering  the expected payment of  expert $r$ in the continuous mechanism can be written as 
\begin{align*}
 &\exppay_{exp}(\expsig_r,\exprep_r)\\
 = & \int_{\expset} \pdfs[\expele\mid \expSig_r=\expsig_r] \log\left(\frac{\exppost(\exprep_r)\cdot \exppost(\expele)}{P}+\frac{(1-\exppost(\exprep_r))(1-\exppost(\expele))}{1-P}\right) {\rm d}\expele\\
    = & \int_{\expset} (f_0(\expele)(1-\exppost(\expsig_r))+f_1(s)\cdot\exppost(\expsig_r)) \log \left(\frac{\exppost(\exprep_r)\cdot  f_1(\expele)+(1-\exppost(\exprep_r))f_0(\expsig)}{(1-P)f_0(\expele)+P \cdot f_1(\expele)}  \right){\rm d}\expele.
\end{align*}

\subsubsection{Expert's Expected Payment in Discrete Mechanism}
Now we consider expert $r$'s expected payment in discrete mechanism (i.e. Shifted Peer Prediction mechanism). We will consider group $l$ non-expert case w.l.o.g. The expert has private signal $\expsig_r$ and report $\nonirep{l}_r$. The peer, a group $l$ non-expert, has private signal $\nonisig{l}$ and truthful report ($\nonirep{l}=\nonisig{l}$).

The payment of $r$ is 
\begin{equation*}
    R_{sppm}^{l}(\nonirep{l}_r, \nonirep{l})=\PSI{l}(\nonirep{l}, \SSIdist{l} (\nonirep{l}_r)) - \PSI{l}(\nonirep{l}, \SSIdist{l})
\end{equation*}

Similarly, we can compute $r$'s posterior distribution towards $\nonisig{l}=\noniele{l}_j$:

\begin{align*}
    \Pr[\noniSig{l}=\noniele{l}_j\mid \expSig_r=\expsig_r]=&\Pr[\noniSig{l}=\noniele{l}_j\mid Y=0]\cdot\Pr[Y=0\mid \expSig_r=\expsig_r]\\
    +&\Pr[\noniSig{l}=\noniele{l}_j\mid Y=1]\cdot\Pr[Y=1\mid \expSig_r=\expsig_r]\\
    = & \nonidist{l}_0(j)\cdot(1-\exppost(\expsig_r))+\nonidist{l}_1(j)\cdot\exppost(\expsig_r). 
\end{align*}

The expected payment of expert $r$ in the group l non-expert mechanism can be written as 
\begin{align*}
    \nonipay{l}_{exp}(\expsig_r,\nonirep{l}_r) = &\sum_{j=1}^{\nonisetnum{l}}  \Pr[\noniSig{l}=\noniele{l}_j\mid \expSig_r=\expsig_r] \left(\PSI{l}(\noniele{l}_j, \SSIdist{l} (\nonirep{l}_r)) - \PSI{l}(\noniele{l}_j, \SSIdist{l})\right)\\
    = &\sum_{j=1}^{\nonisetnum{l}}  \left(\nonidist{l}_0(j)\cdot(1-\exppost(\expsig_r))+\nonidist{l}_1(j)\cdot\exppost(\expsig_r)\right) \left(\PSI{l}(\noniele{l}_j, \SSIdist{l} (\nonirep{l}_r)) - \PSI{l}(\noniele{l}_j, \SSIdist{l})\right).
\end{align*}

\subsubsection{Non-expert's Expected Payment}
Now we consider a group $l$ (w.l.o.g) non-expert $r$'s expected payment in each mechanism. we assume that $\nonisig{l}_r=\noniele{l}_j$. Her peer is an expert with private signal $\expsig$, and report $\exprep=\expsig$. The payment of $r$ is:

\begin{equation*}
    R_{ecgm}(\exprep_r, \exprep) = \log \left(\frac{\exppost(\exprep_r)\cdot \exppost(\expele)}{P}+\frac{(1-\exppost(\exprep_r))(1-\exppost(\expele))}{1-P}\right).
\end{equation*}
And $r$'s posterior distribution of $\expSig$ is

\begin{align*}
    \pdfs[\expele\mid \noniSig{l}_r=\noniele{l}_j]=&f_0(\expele)\Pr[Y=0\mid \noniSig{l}_r=\noniele{l}_j] + f_1(\expele)\Pr[Y=1\mid \noniSig{l}_r=\noniele{l}_j]\\
    = & f_0(\expele)(1-\nonipost(\noniele{l}_j))+f_1(\expele)\cdot \nonipost(\noniele{l}_j). 
\end{align*}
Therefore, the expected payment of group $l$ non-expert $r$ in the continuous mechanism can be written as 
\begin{align*}
     &\exppay_{non-l}(\nonisig{l}_r=\noniele{l}_j, \exprep_r) \\
    = & \int_{\expset}  \pdfs[\expele\mid \noniSig{l}_r=\noniele{l}_j] \log \left(\frac{\exppost(\exprep_r)\cdot \exppost(\expele)}{P}+\frac{(1-\exppost(\exprep_r))(1-\exppost(\expele))}{1-P}\right) {\rm d}\expele\\
    = & \int_{\expset} \left(f_0(\expele)(1-\nonipost(\noniele{l}_j))+f_1(\expele)\cdot \nonipost(\noniele{l}_j)\right) \left(\frac{\exppost(\exprep_r)\cdot  f_1(\expele)+(1-\exppost(\exprep_r))f_0(\expele)}{(1-P)f_0(\expele)+P \cdot f_1(\expele)}  \right){\rm d}\expele.
\end{align*}

Then we consider $r$'s expected payment in group $l$ non-expert mechanism (same as her private information space). Assume $r$'s report to be $\nonirep{l}_r$. Her peer is also a group $l$ non-expert, and has a private signal $\nonisig{l}$, and reports $\nonirep{l}=\nonisig{l}$. The payment of $r$ is:

\begin{equation*}
    R_{sppm}^{l}(\nonirep{l}_r=\noniele{l}_j, \nonirep{l}) = \PSI{l}(\nonisig{l}, \SSIdist{l} (\nonirep{l}_r=\noniele{l}_j)) - \PSI{l}(\nonisig{l}, \SSIdist{l})
\end{equation*}
$r$'s posterior distribution of $\noniSig{l}$ is 

\begin{align*}
    \Pr[\noniSig{l}=\noniele{l}_k\mid \noniSig{l}_r=\noniele{l}_j]=&\Pr[\noniSig{l}=\noniele{l}_k\mid Y=0]\cdot\Pr[Y=0\mid \noniSig{l}_r=\noniele{l}_j]\\
    +&\Pr[\noniSig{l}=\noniele{l}_k\mid Y=1]\cdot\Pr[Y=1\mid \noniSig{l}_r=\noniele{l}_j]\\
    = & \nonidist{l}_0(k)\cdot(1-\nonipost(\noniele{l}_j))+\nonidist{l}_1(k)\cdot\nonipost(\noniele{l}_j). 
\end{align*}

The expected payment of group 1 non-expert $r$ in the group 1 non-expert mechanism can be written as 
\begin{align*}
    &\nonipay{l}_{non-1}(\nonisig{l}_r=\noniele{l}_j,\nonirep{l}_r) \\
    = &\sum_{k=1}^{\nonisetnum{l}}  \Pr[\noniSig{l}=\noniele{l}_k\mid \noniSig{l}_r=\noniele{l}_j] \left(\PSI{l}(\noniele{l}_k, \SSIdist{l} (\nonirep{l}_r)) - \PSI{l}(\noniele{l}_k, \SSIdist{l}))\right)\\
    = &\sum_{k=1}^{\nonisetnum{l}}  \left(\nonidist{l}_0(k)\cdot(1-\nonipost(\noniele{l}_j))+\nonidist{l}_1(k)\cdot\nonipost(\noniele{l}_j)\right) \left(\PSI{l}(\noniele{l}_k, \SSIdist{l} (\nonirep{l}_r)) - \PSI{l}(\noniele{l}_k, \SSIdist{l}))\right).
\end{align*}

Similarly, the expected payment of $r$ in group 2 non-expert mechanism (with report $\nonirep{h}_r)$ can be written as 
\begin{align*}
    &\nonipay{h}_{non-1}(\nonisig{l}_r=\noniele{l}_j,\nonirep{h}_r) \\
    = &\sum_{k=1}^{\nonisetnum{h}}  \Pr[\noniSig{h}=\noniele{h}_k\mid \noniSig{l}_r=\noniele{l}_j] \left(\PSI{h}(\noniele{h}_k, \SSIdist{h} (\nonirep{h}_r)) - \PSI{h}(\noniele{h}_k, \SSIdist{h}))\right)\\
    = &\sum_{k=1}^{\nonisetnum{h}}  \left(\nonidist{h}_0(k)\cdot(1-\nonipost(\noniele{l}_j))+\nonidist{h}_1(k)\cdot\nonipost(\noniele{l}_j)\right) \left(\PSI{h}(\noniele{h}_k, \SSIdist{h} (\nonirep{h}_r)) - \PSI{h}(\noniele{h}_k, \SSIdist{h}))\right).
\end{align*}

\subsection{Proof of Theorem \ref{thm:composite}}
\begin{thmbis}{thm:composite}
Given Assumption~\ref{asm:ci} to \ref{asm:ps}, the Composite Elicitation mechanism is \emph{\newtruthful}. 
\end{thmbis}
Note that in the appendix we add Assumption~\ref{asm:ps} for a wider range of proper scoring rule for SPPM. 

\subsubsection{Proof Sketch}
To prove \newtruthfulness, it suffices prove the following two types of truthfulness for every agent, given that all other agents report truthfully.

 \textbf{1.~Interior truthfulness.} If the agent's report is in the same $\Theta^l$ as her private signal, her expected payment is uniquely maximized when she reports truthfully.  

\textbf{2.~Exterior Truthfulness.} If the agent's report is in a different $\Theta^l$, her expected payment will not exceed the payment under truthful report (and equal only for signals with same posterior distribution).

The proof proceeds with giving explicit  formulas of the  coefficients $\{\noniA{l}_k, \noniB{l}_k\}$, then proving that  both interior and exterior truthfulness hold for all agents.
This will be achieved by taking advantage of the second part of Assumption~\ref{asm:wise} to construct some special cases of signal, with which we bridge the gap between experts and non-experts.  For any $\noniele{l}_k\in \Theta^l$, let $\expelek{l}{k}\in\Theta^0$  denote a signal that is guaranteed by Assumption~\ref{asm:wise} (2), that is, $\exppost(\expelek{l}{k})=\nonipost(\noniele{l}_k)$, which is used to construct the linear coefficients. The proof relies on calculating and comparing expected payments of different groups of agents. While it is hard to directly compare expected payment of different groups, $\expelek{l}{k}$ provides a special case, where different expected payments will be equal or comparable to each other. In this way, we can convert all different expected payments into one group and compare them.  
.

\textbf{Note: For convenience, In the following writings, we denote $\exppost(\expsig_r)$ as $\exppost_r$ if the private signal is not emphasized. }

\subsubsection{Explicit form of Linear Transformation Coefficients}
We first give the explicit form of the linear transformation coefficients. In the following proof, we'll show that they lead to a \newtruthful{}mechanism. 

For $\noniset{l}$-reporters, 
\begin{enumerate}
    \item If $\nonidist{l}_0(k)\neq\nonidist{l}_1(k)$, 
    \begin{equation}
        \label{eq:noniA}
        \noniA{l}_k=\left.\frac{\frac{{\rm d}\exppay_{exp}(\expsig_r, \expsig_r)}{{\rm d}\exppost(\expsig_r)}}{\frac{\partial \nonipay{l}_{exp}(\expsig_r,\noniele{l}_k)}{\partial \exppost(\expsig_r)}}\right|_{\expsig_r = \expelek{l}{k}}
    \end{equation}
    \begin{equation}
    \label{eq:noniB}
        \noniB{l}_k = \exppay_{non-l}(\noniele{l}_k,\expelek{l}{k}) - \noniA{l}_k\cdot \nonipay{l}_{non-l}(\noniele{l}_k,\noniele{l}_k)
    \end{equation}
    where $\expelek{l}{k}$ satisfies $\exppost(\expelek{l}{k})=\nonipost(\noniele{l}_k)$. 
    \item If $\nonidist{l}_0(k)=\nonidist{l}_1(k)$, $\noniA{l}_k=\noniB{l}_k=0$. 
\end{enumerate}

The explicit form here involve the derivatives of expected payments. However, these derivatives have explicit forms which involves only integration. We give the explicit form of derivative in Equation~\ref{eq:exp-partial} and~\ref{eq:noni-partial}.

\subsubsection{Helping Lemmas}
In this section we introduce and prove several lemmas that are necessary in the proof. 

First we show that given the assumption of informative prior, expert's posterior PDF $\exppost(\expele)$ is strict monotone to $\expele$. 

\begin{lemma}
\label{lemma:exp-mono}
Given the Informative Prior (Assumption~\ref{asm:ip}), expert's posterior distribution to ground truth $\exppost(\expele)$ is strictly monotone to $\expele$. 
\end{lemma}
\begin{proof}
We turn to prove that if $\exppost(\expele)$ is not strictly monotone, then the Informative Prior does not hold. 
Consider two agents $i$ and $j$ with private signals $\expsig_i,\expsig_j\in\expset$ and satisfy $\expsig_i\neq \expsig_j$, and $\exppost(\expsig_i)=\exppost(\expsig_j)$. Since $\exppost(\expele)$ is not strictly monotone, such $\expsig_i$ and $\expsig_j$ exists. 
Then we consider the posterior distribution of some expert's private signal $\expele$ from $i$ and $j$ separately. We write their PDFs as following.
\begin{align*}
    \pdfs[\expele\mid \expSig_i=\expsig_i]=&(1-\exppost(\expsig_i))\cdot f_0(\expele) + \exppost(\expsig_i)\cdot f_1(\expele). \\
    \pdfs[\expele\mid \expSig_i=\expsig_j]=&(1-\exppost(\expsig_j))\cdot f_0(\expele) + \exppost(\expsig_j)\cdot f_1(\expele).
\end{align*}
Note that since $\exppost(\expsig_i)=\exppost(\expsig_j)$, we have 
\begin{equation*}
    \pdfs[\expele\mid \expSig_i=\expsig_i] = \pdfs[\expele\mid \expSig_i=\expsig_j]
\end{equation*}
for every $\expele\in\expset$. Therefore, the two posterior distribution of $\expele$ are equal. This contradicts the informative prior assumption for experts. 
Therefore, with informative prior, such $\expsig_i$ and $\expsig_j$ does not exist, and $\exppost(\expele)$ is strictly monotone to $\expele$. 
\end{proof}

Then we introduce an equivalent condition of $\nonidist{l}_0(k)=\nonidist{l}_1(k)$.

\begin{lemma}
\label{lemma:noni-neq}
For any $l\in [\noninum]$ and any $k\in[\nonisetnum{l}]$, $\nonidist{l}_0(k)=\nonidist{l}_1(k)$ if and only if 
\begin{equation*}
    \frac{\partial \nonipay{l}_{exp}(\expsig_r, \noniele{l}_k)}{\partial \exppost_r} = 0.
\end{equation*}
\end{lemma}

\begin{proof}
From Equation \ref{eq:SSexp} we know that 
\begin{align*}
    \nonipay{l}_{exp}(\expsig_r, \noniele{l}_k)=&\sum_{j=1}^{\nonisetnum{l}}  \left(\nonidist{l}_0(j)\cdot(1-\exppost(\expsig_r))+\nonidist{l}_1(j)\cdot\exppost(\expsig_r)\right) \left(\PSI{l}(\noniele{l}_j, \SSIdist{l} (\noniele{l}_k)) - \PSI{l}(\noniele{l}_j, \SSIdist{l})\right).
\end{align*}
We calculate$ \frac{\partial  \nonipay{l}_{exp}(\expsig_r, \noniele{l}_k)}{\partial \exppost_r}$: 
\begin{equation}
\label{eq:noni-partial}
     \frac{\partial  \nonipay{l}_{exp}(\expsig_r, \noniele{l}_k)}{\partial \exppost_r} = \sum_{j=1}^{\nonisetnum{l}}  \left(\nonidist{l}_1(j)-\nonidist{l}_0(j)\right) \left(\PSI{l}(\noniele{l}_j, \SSIdist{l} (\noniele{l}_k)) - \PSI{l}(\noniele{l}_j, \SSIdist{l})\right).
\end{equation}
First we prove that when $\nonidist{l}_0(k)=\nonidist{l}_1(k)$, this derivative equals to 0. When $\nonidist{l}_0(k)=\nonidist{l}_1(k)$, the posterior distribution satisfies:
\begin{equation*}
    \nonipost(\noniele{l}_k)=\frac{P\cdot \nonidist{l}_1(k)}{(1-P)\nonidist{l}_0(k)+P\cdot \nonidist{l}_1(k)} = P
\end{equation*}
The posterior of the ground truth is the same as the prior distribution. In this case, for the posterior signal distribution $\SSIdist{l} (\noniele{l}_k)$ we have:
\begin{align*}
    (\SSIdist{l} (\noniele{l}_k))_j=& \Pr[\noniSig{l}=\noniele{l}_j\mid \noniSig{l}_r=\noniele{l}_k]\\
    = &(1-\nonipost(\noniele{l}_k))\cdot \nonidist{l}_0(j)+\nonipost(\noniele{l}_k)\cdot \nonidist{l}_1(j)\\
    = & (1-P)\cdot \nonidist{l}_0(j)+P\cdot \nonidist{l}_1(j)\\
    =&\Pr[\noniSig{l}=\noniele{l}_j]\\
    =& (\SSIdist{l})_j
\end{align*}
This means that when $\nonidist{l}_0(k)=\nonidist{l}_1(k)$, the post signal distribution equals to the prior signal distribution. In this case, 
\begin{align*}
     \frac{\partial  \nonipay{l}_{exp}(\expsig_r, \noniele{l}_k)}{\partial \exppost_r} = & \sum_{j=1}^{\nonisetnum{l}}  \left(\nonidist{l}_1(j)-\nonidist{l}_0(j)\right) \left(\PSI{l}(\noniele{l}_j, \SSIdist{l} (\noniele{l}_k)) - \PSI{l}(\noniele{l}_j, \SSIdist{l})\right).\\
     = & \sum_{j=1}^{\nonisetnum{l}}  \left(\nonidist{l}_1(j)-\nonidist{l}_0(j)\right) \left(\PSI{l}(\noniele{l}_j, \SSIdist{l}) - \PSI{l}(\noniele{l}_j, \SSIdist{l})\right)\\
     = &0.
\end{align*}
Then we prove that when $\nonidist{l}_0(k)\neq\nonidist{l}_1(k)$, the derivative is non-zero. We construct a function $\nonipay{l}(p)$, where $\nonipay{l}(P)=0$ and $\nonipay{l}(\nonipost(\noniele{l}_k))=\frac{\partial  \nonipay{l}_{exp}(\expsig_r, \noniele{l}_k)}{\partial \exppost_r}$. Then we prove that $\nonipay{l}(p)$ is strictly monotone. Since when $\nonidist{l}_0(k)\neq\nonidist{l}_1(k), \nonipost(\noniele{l}_k)\neq P$, we show that $\frac{\partial  \nonipay{l}_{exp}(\expsig_r, \noniele{l}_k)}{\partial \exppost_r} \neq 0$. 

We defined $\nonipay{l}(p)$ as follows: 
\begin{equation*}
    \nonipay{l}(p) = \sum_{j=1}^{\nonisetnum{l}} \left( \nonidist{l}_1(j)- \nonidist{l}_0(j)\right) \left(\PSI{l}(\noniele{l}_j, \ssidist{l}(p)) - \PSI{l}(\noniele{l}_j, \SSIdist{l})\right),
\end{equation*}
where $\ssidist{l}(p)$ is a distribution on $\noniset{l}$. Formally, $\ssidist{l}(p) = (\ssidisti{l}_1(p),\ssidisti{l}_2(p), \cdots, \ssidisti{l}_{\nonisetnum{l}}(p))$, where  $\ssidisti{l}_{i}(p) = (1-p)\cdot \nonidist{l}_0(j)+p\cdot \nonidist{l}_1(j)$. 
When $p=P$, $\ssidist{l}(p) = \SSIdist{l}$, and $\nonipay{l}(P)=0$. When $p=\nonipost(\noniele{l}_k)$,  $\ssidist{l}(p) = \SSIdist{l} (\noniele{l}_k)$, and $\nonipay{l}(\nonipost(\noniele{l}_k))=\frac{\partial  \nonipay{l}_{exp}(\expsig_r, \noniele{l}_k)}{\partial \exppost_r}$. 

Now we prove that $\nonipay{l}(p)$ is strictly monotone by showing that its derivative is strictly positive. We consider the derivative of the proper scoring rule by calculating the derivative of its equivalent form in lemma \ref{lemma:ps}: 
\begin{align*}
    \frac{\partial \PSI{l}(\noniele{l}_j,\ssidist{l}(p))}{\partial p}
   =&\sum_{i=1}^{\nonisetnum{l}} \frac{\partial G^{l}(\ssidist{l}(p))}{\partial \ssidisti{l}_i} \frac{\partial \ssidisti{l}_i}{\partial p} + \sum_{i=1}^{\nonisetnum{l}}  \frac{\partial^2 G^{l}(\ssidist{l}(p))}{\partial \ssidisti{l}_j \partial\ssidisti{l}_i}\frac{\partial \ssidisti{l}_i}{\partial p}\\
   - & \sum_{k=1}^{\nonisetnum{l}} \sum_{i=1}^{\nonisetnum{l}} \frac{\partial^2 G^{l}(\ssidist{l}(p))}{\partial \ssidisti{l}_k \partial\ssidisti{l}_i}\frac{\partial \ssidisti{l}_i}{\partial p}\cdot \ssidisti{l}_k - \sum_{k=1}^{\nonisetnum{l}}  \frac{\partial G^{l}(\ssidist{l}(p))}{\partial \ssidisti{l}_k} \frac{\partial \ssidisti{l}_k}{ \partial p}\\
    = & \sum_{i=1}^{\nonisetnum{l}}  \frac{\partial^2 G^{l}(\ssidist{l}(p))}{\partial \ssidisti{l}_j \partial\ssidisti{l}_i}\frac{\partial \ssidisti{l}_i}{\partial p} - \sum_{k=1}^{\nonisetnum{l}} \sum_{i=1}^{\nonisetnum{l}} \frac{\partial^2 G^{l}(\ssidist{l}(p))}{\partial \ssidisti{l}_k \partial\ssidisti{l}_i}\frac{\partial \ssidisti{l}_i}{\partial p}\cdot \ssidisti{l}_k
\end{align*}

Note that since the Hessian Matrix of $G^{l}$ exists, we use the gradient of $G^{l}$ rather than subgradient. 

Then by Assumption \ref{asm:ps} (assumption of proper scoring rule), $\frac{\partial^2 G^{l}(\ssidist{l}(p))}{\partial \ssidisti{l}_i\partial \ssidisti{l}_j}\neq 0$ if and only if $i=j$. Therefore, we have 
\begin{align*}
     \frac{{\rm d} \nonipay{l}(p)}{{\rm d} p} = & \sum_{j=1}^{\nonisetnum{l}} \left( \nonidist{l}_1(j)- \nonidist{l}_0(j)\right) \left(\sum_{i=1}^{\nonisetnum{l}}  \frac{\partial^2 G^{l}(\ssidist{l}(p))}{\partial \ssidisti{l}_j \partial\ssidisti{l}_i}\frac{\partial \ssidisti{l}_i}{\partial p} - \sum_{k=1}^{\nonisetnum{l}} \sum_{i=1}^{\nonisetnum{l}} \frac{\partial^2 G^{l}(\ssidist{l}(p))}{\partial \ssidisti{l}_k \partial\ssidisti{l}_i}\frac{\partial \ssidisti{l}_i}{\partial p}\cdot \ssidisti{l}_k\right),\\
     =&\sum_{j=1}^{\nonisetnum{l}} \left( \nonidist{l}_1(j)- \nonidist{l}_0(j)\right) \left( \frac{\partial^2 G^{l}(\ssidist{l}(p))}{(\partial \ssidisti{l}_j)^2}\frac{\partial \ssidisti{l}_j}{\partial p} - \sum_{i=1}^{\nonisetnum{l}} \frac{\partial^2 G^{l}(\ssidist{l}(p))}{ (\partial\ssidisti{l}_i)^2}\frac{\partial \ssidisti{l}_i}{\partial p}\cdot \ssidisti{l}_i \right )\\
     = & \sum_{j=1}^{\nonisetnum{l}} \left( \nonidist{l}_1(j)- \nonidist{l}_0(j)\right)\frac{\partial^2 G^{l}(\ssidist{l}(p))}{(\partial \ssidisti{l}_j)^2}\frac{\partial \ssidisti{l}_j}{\partial p} \\
     - & \left(\sum_{j=1}^{\nonisetnum{l}} \left( \nonidist{l}_1(j)- \nonidist{l}_0(j)\right)\right)
     \left(\sum_{i=1}^{\nonisetnum{l}} \frac{\partial^2 G^{l}(\ssidist{l}(p))}{ (\partial\ssidisti{l}_i)^2}\frac{\partial \ssidisti{l}_i}{\partial p}\cdot \ssidisti{l}_i\right)
\end{align*}
Note that the second term are the product of two sums, and the first sum $\sum_{j=1}^{\nonisetnum{l}} \left( \nonidist{l}_1(j)- \nonidist{l}_0(j)\right)=0$. Therefore according to the definition of $\ssidisti{l}_j)$, 
\begin{align*}
    \frac{{\rm d} \nonipay{l}(p)}{{\rm d} p} = & \sum_{j=1}^{\nonisetnum{l}} \left( \nonidist{l}_1(j)- \nonidist{l}_0(j)\right)\frac{\partial^2 G^{l}(\ssidist{l}(p))}{(\partial \ssidisti{l}_j)^2}\frac{\partial \ssidisti{l}_j}{\partial p}\\
    =& \sum_{j=1}^{\nonisetnum{l}} \left( \nonidist{l}_1(j)- \nonidist{l}_0(j)\right)^2\frac{\partial^2 G^{l}(\ssidist{l}(p))}{(\partial \ssidisti{l}_j)^2}
\end{align*}
By assumption \ref{asm:ps} and convexity of $G^{l}$, for any $i=1,2,\cdots,\nonisetnum{l}, \frac{\partial^2 G^{l}(\ssidist{l}(p))}{(\partial \ssidisti{l}_j)^2} > 0$.  And since $\nonidist{l}_0(k)\neq \nonidist{l}_1(k)$, every term of $\frac{{\rm d} \nonipay{l}(p)}{{\rm d} p}$ is non-negative and at least one term is strictly positive. Therefore, $$ \frac{{\rm d} \nonipay{l}(p)}{{\rm d} p} > 0.$$
Therefore, $\nonipay{l}(p)$ is strictly monotone.
Since when $\nonidist{l}_0(k)\neq \nonidist{l}_1(k), \nonipost(\noniele{l}_k)\neq P$, we show that $\frac{\partial  \nonipay{l}_{exp}(\expsig_r, \noniele{l}_k)}{\partial \exppost_r} \neq 0$.
\end{proof}

The following lemma shows that $\nonipay{l}_{exp}(\expsig_r, \expsig_r) - \noniA{l}_k\cdot \nonipay{l}_{exp}(\expsig_r,\noniele{l}_k)$ is minimized when $\expsig_r$ satisfies $\exppost(\expsig_r)=\nonipost(\noniele{l}_k)$. This lemma is used in several lemmas in the proof part. 

\begin{lemma}
\label{lemma:noni-min}
For any $l\in[\noninum]$, given $\nonidist{l}_0(k)\neq \nonidist{l}_1(k)$,let $\expelek{l}{k}$ satisfies $\exppost(\expelek{l}{k})=\nonipost(\noniele{l}_k)$. Then for any $\expsig_r\in \expset$, we have 
\begin{equation*}
    \exppay_{exp}(\expsig_r, \expsig_r) - \noniA{l}_k\cdot \nonipay{l}_{exp}(\expsig_r,\noniele{l}_k) \ge \exppay_{exp}(\expelek{l}{k} \expelek{l}{k}) - \noniA{l}_k\cdot \nonipay{l}_{exp}(\expelek{l}{k},\noniele{l}_k)
\end{equation*}
The equation holds if and only if $\expsig_r=\expelek{l}{k}$. 
\end{lemma}
\begin{proof}
We proof Lemma~\ref{lemma:noni-min} with the same technique of Lemma~\ref{lemma:noni-neq}. We regard $\exppay_{exp}(\expsig_r, \expsig_r) - \noniA{l}_k\nonipay{l}_{exp}(\expsig_r,\noniele{l}_k)$ as a function of $\exppost_r$, and look into its derivatives. We will show that this function's derivatives equals to 0 when $\exppost_r=\nonipost(\noniele{l}_k)$, and is strictly monotone. Therefore, the function is uniquely minimized when $\exppost_r=\nonipost(\noniele{l}_k)$. 

Consider the form of $ \exppay_{exp}(\expsig_r, \expsig_r)$ and $\nonipay{l}_{exp}(\expsig_r,\noniele{l}_k)$:
\begin{align*}
    &\exppay_{exp}(\expsig_r, \expsig_r)\\
    = &\int_{\expset} \left((1-\exppost(\expsig_r))\cdot f_0(\expele) +\exppost(\expsig_r)\cdot f_1(\expele)\right)  \log \left( \frac{(1-\exppost(\expsig_r))\cdot f_0(\expele) +\exppost(\expsig_r)\cdot f_1(\expele)}{(1-P)\cdot f_0(\expele) +P\cdot f_1(\expele)}\right) {\rm d}\expele\\
    &\nonipay{l}_{exp}(\expsig_r, \noniele{l}_k)=\sum_{j=1}^{\nonisetnum{l}}  \left(\nonidist{l}_0(j)\cdot(1-\exppost(\expsig_r))+\nonidist{l}_1(j)\cdot\exppost(\expsig_r)\right) \left(\PSI{l}(\noniele{l}_j, \SSIdist{l} (\noniele{l}_k)) - \PSI{l}(\noniele{l}_j, \SSIdist{l})\right)
\end{align*}
Note that both $\exppay_{exp}$ and $\nonipay{l}_{exp}$ can be regarded as a function of $\exppost(\expsig_r)$. We construct the following functions:
\begin{align*}
    \exppay^*(p)=& \int_{\expset} \left((1-p)\cdot f_0(\expele) +p\cdot f_1(\expele)\right)  \log \left( \frac{(1-p)\cdot f_0(\expele) +p\cdot f_1(\expele)}{(1-P)\cdot f_0(\expele) +P\cdot f_1(\expele)}\right) {\rm d}\expele\\
    \nonipay{*}(p)=&\sum_{j=1}^{\nonisetnum{l}}  \left(\nonidist{l}_0(j)\cdot(1-p)+\nonidist{l}_1(j)\cdot p\right) \left(\PSI{l}(\noniele{l}_j, \SSIdist{l} (\noniele{l}_k)) - \PSI{l}(\noniele{l}_j, \SSIdist{l})\right).\\
    \expminussemi(p)=&\exppay^*(p)-\noniA{l}_k\cdot \nonipay{*}(p). 
\end{align*}
Note that $\nonipay{*}(p)$ is different from $\nonipay{l}(p)$ in Lemma~\ref{lemma:noni-neq}. 

Then our goal is to prove that $\expminussemi(p)$ is uniquely minimized on $p=\nonipost(\noniele{l}_k)$. We consider the derivative of $\exppay^*(p)$:

\begin{align*}
    \frac{{\rm d} \exppay^*(p)}{{\rm d} p} = & \int_{\expset} \left(f_1(\expele)-f_0(\expele)\right)  \log \left( \frac{(1-p)\cdot f_0(\expele) +p\cdot f_1(\expele)}{(1-P)\cdot f_0(\expele)+  P\cdot f_1(\expele)}\right) {\rm d}\expele \\
    + & \int_{\expset} \left((1-p)\cdot f_0(\expele) +p\cdot f_1(\expele)\right)  \frac{f_1(\expele)-f_0(\expele)}{(1-p)\cdot f_0(\expele) +p\cdot f_1(\expele)} {\rm d}\expele
\end{align*}
Note the the second term
\begin{align*}
    &\int_{\expset} \left((1-p)\cdot f_0(\expele) +p\cdot f_1(\expele)\right)  \frac{f_1(\expele)-f_0(\expele)}{(1-p)\cdot f_0(\expele) +p\cdot f_1(\expele)} {\rm d}\expele \\
    = & \int_{\expset}(f_1(\expele)-f_0(\expele)) {\rm d}\expele\\
    = & 0. 
\end{align*}
Therefore, 
\begin{align}
\label{eq:exp-partial}
    \frac{{\rm d} \exppay^*(p)}{{\rm d} p} = & \int_{\expset} \left(f_1(\expele)-f_0(\expele)\right)  \log \left( \frac{(1-p)\cdot f_0(\expele) +p\cdot f_1(\expele)}{(1-P)\cdot f_0(\expele) +P\cdot f_1(\expele)}\right) {\rm d}\expele.
\end{align}
Then we consider the derivative of $\nonipay{*}(p)$:

\begin{align*}
    \frac{{\rm d} \nonipay{*}(p)}{{\rm d} p}
    =&
    \sum_{j=1}^{\nonisetnum{l}}  \left(\nonidist{l}_1(j)-\nonidist{l}_0(j)\right) \left(\PSI{l}(\noniele{l}_j, \SSIdist{l} (\noniele{l}_k)) - \PSI{l}(\noniele{l}_j, \SSIdist{l})\right)
\end{align*}
Note that according to Lemma~\ref{lemma:noni-neq} since $\nonidist{l}_0(k)\neq \nonidist{l}_1(k)$, this derivative is nonzero. 

Now we show that $\frac{{\rm d}\expminussemi(p)}{{\rm d}p}=0$ when $\exppost_r=\nonipost(\noniele{l}_k)$. Note that 
$ \frac{{\rm d} \exppay^*(p)}{{\rm d} p} =  \frac{{\rm d} \exppay_{exp}(\expsig_r,\expsig_r)}{{\rm d} \exppost_r}$, and $\frac{{\rm d} \nonipay{*}(p)}{{\rm d} p}=\frac{\partial \nonipay{l}_{exp}(\expsig_r,\noniele{l}_k)}{\partial \exppost_r}$, and according to the definition of $\noniA{l}_k$ in Equation~\ref{eq:noniA},
\begin{align*}
    \left.\frac{{\rm d}\expminussemi(p)}{{\rm d}p}\right|_{p=\nonipost(\noniele{l}_k)} = &\left.\frac{{\rm d} \exppay_{exp}(\expsig_r,\expsig_r)}{{\rm d} \exppost_r} \right|_{\expele_r=\expelek{l}{k}}\\
    - & \left.\frac{\frac{{\rm d}\exppay_{exp}(\expsig_r, \expsig_r)}{{\rm d}\exppost_r}}{\frac{\partial \nonipay{l}_{exp}(\expsig_r,\noniele{l}_k)}{\partial \exppost_r}}\right|_{\expele_r=\expelek{l}{k}} \cdot
    \left.\frac{\partial \nonipay{l}_{exp}(\expsig_r,\noniele{l}_k)}{\partial \exppost_r} \right|_{\expele_r=\expelek{l}{k}}\\
    = & 0
\end{align*}
Then we prove that $\frac{{\rm d}^2\expminussemi(p)}{{\rm d}p^2}>0$ for every $p\in[0,1]$. In this way, $\frac{{\rm d}\expminussemi(p)}{{\rm d}p}=0$ is strictly monotone, and \expminussemi(p) is uniquely minimized in $p=\nonipost(\noniele{l}_k)$.
We first consider the second derivative of  $\exppay^*(p)$:
\begin{align*}
    \frac{{\rm d}^2 \exppay^*(p)}{{\rm d} p^2} = & \int_{\expset} \left(f_1(\expele)-f_0(\expele)\right)  \frac{f_1(\expele)-f_0(\expele)}{(1-p)\cdot f_0(\expele) +p\cdot f_1(\expele)} {\rm d}\expele\\
    = & \int_{\expset}\frac{(f_1(\expele)-f_0(\expele))^2}{(1-p)\cdot f_0(\expele) +p\cdot f_1(\expele)} {\rm d}\expele
\end{align*}

According to Assumption~\ref{asm:ip}, $f_0(\expele)$ and $f_1(\expele)$ do not equal \alev, otherwise different realizations of private signal will lead to same posterior distribution of expert's signal to be $f_0(\expele)$. And according to Assumption~\ref{asm:wise}, $f_0(\expele)$ and $f_1(\expele)$ are both strictly positive. Therefore, this second derivative is strictly positive. 
On the other hand, $\frac{{\rm d}^2 \nonipay{*}(p)}{{\rm d} p^2}=0$. Therefore, we have proved that $\expminussemi(p)$ is uniquely minimized at $p=\nonipost(\noniele{l}_k)$. Moreover, according to Lemma~\ref{lemma:exp-mono}, $\exppost(\expele)$ is monotone to $\expsig_r$. This implies that $\expminussemi(\exppost(\expsig_r))$ is uniquely minimized at $\expsig_r=\expelek{l}{k}$. 
Therefore, we have 
\begin{equation*}
    \exppay_{exp}(\expsig_r, \expsig_r) - \noniA{l}_k\nonipay{l}_{exp}(\expsig_r,\noniele{l}_k) \ge \exppay_{exp}(\expelek{l}{k}, \expelek{l}{k}) - \noniA{l}_k\nonipay{l}_{exp}(\expelek{l}{k},\noniele{l}_k)
\end{equation*}
The equation holds if and only if $\expsig_r=\expelek{l}{k}$. 
\end{proof}

\subsubsection{Expert  Part}
In expert part, we'll first prove the interior truthfulness, and them prove the exterior truthfulness. The interior truthfulness of expert comes from the non-negativity of cross entropy. 

\begin{lemma}[Expert Interior Truthfulness]
\label{lemma: exp_exp}
For any expert $r$, and for any $\expsig_r, \exprep_r\in\expset$ we have 
\begin{equation*}
    \exppay_{exp}(\expsig_r,\expsig_r)\ge \exppay_{exp}(\expsig_r, \exprep_r).
\end{equation*}
The equality holds if and only if $\exprep_r=\expsig_r$.
\end{lemma}
\begin{proof}
According to Equation \ref{eq:CCexp}, we have 
\begin{align*}
    &\exppay_{exp}(\expsig_r, \exprep_r)\\ 
    = &\int_{\expset} \pdfs[\expele\mid \expSig_r=\expsig_r]R_{ecgm}(\exprep_r, \expele){\rm d}\expele\\
    = & \int_{\expset} \left((1-\exppost_r)\cdot f_0(\expele) +\exppost_r\cdot f_1(\expele)\right)  \log \left( \frac{(1-\exppost(\exprep_r))\cdot f_0(\expele) +\exppost(\exprep_r)\cdot f_1(\expele)}{(1-P)\cdot f_0(\expele) +P\cdot f_1(\expele)}\right) {\rm d}\expele. 
\end{align*}
Then, we try to rewrite $\exppay_{exp}$ using the following three functions:
\begin{itemize}
    \item $f_1^*(\expele) = (1-\exppost_r)\cdot f_0(\expele) +\exppost_r\cdot f_1(\expele)$
    \item $f_2^*(\expele) = (1-\exppost(\exprep_r))\cdot f_0(\expele) +\exppost(\exprep_r)\cdot f_1(\expele)$
    \item $f_3^*(\expele) = (1-P)\cdot f_0(\expele) +P\cdot f_1(\expele)$
\end{itemize}

We have 
\begin{equation*}
    \exppay_{exp}(\expsig_r, \exprep_r) = \int_{\expset} f_1^*(\expele) \log \left( \frac{f_2^*(\expele)}{f_3^*(\expele)}\right){\rm d}\expele. 
\end{equation*}
Considering that $f_0$ and $f_1$ are both PDFs, we have $\int_{\expset} f_0(\expele){\rm d}\expele=1$ and $\int_{\expset} f_1(\expele){\rm d}\expele=1$. According to the definition of  $f_1^*(\expele), f_2^*(\expele),f_3^*(\expele)$, we have 
\begin{equation*}
   \int_{\expset} f_1^*(\expele){\rm d}\expele=\int_{\expset} f_2^*(\expele){\rm d}\expele=\int_{\expset} f_3^*(\expele){\rm d}\expele=1.
\end{equation*}
Also note that $f_1^*(\expele), f_2^*(\expele),f_3^*(\expele)$ are always positive, we know that $f_1^*(\expele), f_2^*(\expele),f_3^*(\expele)$ are all PDF of some random variables. Then, by applying the continuous version of Lemma~\ref{lem:lsr}, we have, 
\begin{align*}
    \exppay_{exp}(\expsig_r, \exprep_r) = &  \int_{\expset} f_1^*(\expele) \log \left( \frac{f_2^*(\expele)}{f_3^*(\expele)}\right){\rm d}\expele\\
    \le & \int_{\expset} f_1^*(\expele) \log \left( \frac{f_1^*(\expele)}{f_3^*(\expele)}\right){\rm d}\expele\\
    = & \exppay_{exp}(\expsig_r,\expsig_r). 
\end{align*}
Then we consider the strictness of the mechanism. Notice that $f_1^*(\expele)$ and $f_2^*(\expele)$ are posterior distribution of $\expSig$ conditioning on $\expele_r$ and $\exprep_r$. Therefore, with Assumption~\ref{asm:ip}, $f_1^*(\expele)=f_2^*(\expele)$ \alev if and only if $\expsig_r=\exprep_r$. Therefore, the equality holds if and only if $\expsig_r=\exprep_r$, and the strictness holds. 
\end{proof}

Then we consider the exterior truthfulness for experts. In the special case that $\nonidist{l}_0(k)=\nonidist{l}_1(k)$, reporting discrete signal gets a 0 expected payment, and we only need to prove expert's truthful report's expected payment is non-negative. In the case that $\nonidist{l}_0(k)\neq\nonidist{l}_1(k)$, we first show that in a special point $\expelek{l}{k}$, expected payment of truthful report and a certain discrete report ($\noniele{i}_k$) is equal. Since in Lemma~\ref{lemma:noni-min} we show that the difference of two expected payment (truthful minus discrete) is minimized on $\expsig_r=\expelek{l}{k}$, the exterior truthfulness holds. 
\begin{lemma}[Expert Exterior Truthfulness]
\label{lemma:exp-noni}
For any expert $r$, and for any $\expsig_r\in \expset, l\in[\noninum], k\in[\nonisetnum{l}]$, we have 
\begin{equation*}
    \exppay_{exp}(\expsig_r,\expsig_r) \ge \noniA{l}_k\cdot \nonipay{l}_{exp}(\expsig_r, \noniele{l}_k)+\noniB{l}_k
\end{equation*}
Moreover, the equation holds if and only if $\expsig_r=\expelek{l}{k}$, where $\expelek{l}{k}$ satisfies $\exppost_r^k=\exppost(\expelek{l}{k})=\nonipost(\noniele{l}_k)$.
\end{lemma}
\begin{proof}
First we consider the case that $\nonidist{l}_0(k)=\nonidist{l}_1(k)$. In this case, $\noniA{l}_k=\noniB{l}_k=0$. So we only need to prove that $\exppay_{exp}(\expsig_r,\expsig_r)\ge 0$. 

We use the same technique as that in Lemma \ref{lemma: exp_exp}: Let $f_1^*(\expele) = (1-\exppost_r)\cdot f_0(\expele) +\exppost_r\cdot f_1(\expele)$ and $f_3^*(\expele) = (1-P)\cdot f_0(\expele) +P\cdot f_1(\expele)$, we can rewrite $\exppay_{exp}(\expsig_r,\expsig_r)$ as
\begin{equation*}
    \exppay_{exp}(\expsig_r,\expsig_r)=\int_{\expset} f_1^*(\expele) \log \left( \frac{f_1^*(\expele)}{f_3^*(\expele)}\right){\rm d}\expele. 
\end{equation*}
Since both $f_1^*$ and $f_3^*$ are PDFs on $\expset$, according to Lemma~\ref{lem:lsr},
\begin{align*}
    \exppay_{exp}(\expsig_r, \expsig_r) = &  \int_{\expset} f_1^*(\expele) \log \left( \frac{f_1^*(\expele)}{f_3^*(\expele)}\right){\rm d}\expele\\
    = & -\int_{\expset} f_1^*(\expele) \log \left( \frac{f_3^*(\expele)}{f_1^*(\expele)}\right){\rm d}\expele\\
    \ge & \int_{\expset} f_1^*(\expele) \log \left( \frac{f_1^*(\expele)}{f_1^*(\expele)}\right){\rm d}\expele\\
    = & 0.
\end{align*}
And the equation holds if and only if $\exppost(\expsig_r)=P$. On the other hand, Since $\nonidist{l}_0(k)=\nonidist{l}_1(k)$, we have 
\begin{equation*}
    \exppost(\expelek{l}{k})=\nonipost(\noniele{l}_k)=P=\exppost(\expsig_r).
\end{equation*}
    Therefore, exterior truthfulness holds for this case, and the equation holds if and only if $\expsig_r=\expelek{l}{k}$. 
    
    Then we consider the case that $\nonidist{l}_0(k)\neq\nonidist{l}_1(k)$. In this case, according to Equation \ref{eq:noniA} and \ref{eq:noniB}:
    \begin{equation*}
        \noniA{l}_k=\left.\frac{\frac{{\rm d}\exppay_{exp}(\expsig_r, \expsig_r)}{{\rm d}\exppost(\expsig_r)}}{\frac{\partial \nonipay{l}_{exp}(\expsig_r,\noniele{l}_k)}{\partial \exppost(\expsig_r)}}\right|_{\expsig_r = \expelek{l}{k}}
    \end{equation*}
    \begin{equation*}
    \noniB{l}_k =
       \exppay_{non-l}(\noniele{l}_k,\expelek{l}{k}) - \noniA{l}_k\cdot \nonipay{l}_{non-l}(\noniele{l}_k,\noniele{l}_k)
    \end{equation*}
    We consider a special case of expert $r$. In this case, $r$'s private signal $\expsig_r=\expelek{l}{k}$ satisfies $\exppost_r^k=\exppost(\expelek{l}{k})=\nonipost(\noniele{l}_k)$, i.e. the posterior prediction of $r$ for ground truth $Y$ is exactly the posterior prediction of a group $l$ non-expert receiving private signal $\noniele{l}_k$. According to Assumption~\ref{asm:wise}, such $\expelek{l}{k}$ exists. In this case, the expected payment of $r$ also equals to the expected payment of this group 1 non-expert:
    \begin{align*}
    \exppay_{exp}(\expelek{l}{k}, \expelek{l}{k}) = &\int_{\expset} \left((1-\exppost_r^k)\cdot f_0(\expele) +\exppost_r^k\cdot f_1(\expele)\right)R_{ecgm}(\expelek{l}{k}, \expele){\rm d}\expele\\
    = & \int_{\expset} \left((1-\nonipost(\noniele{l}_k))\cdot f_0(\expele) +\nonipost(\noniele{l}_k)\cdot f_1(\expele)\right)R_{ecgm}(\expelek{l}{k}, \expele){\rm d}\expele\\
    =&\exppay_{non-l}(\nonisig{l}=\noniele{l}_k,\expelek{l}{k}).
    \end{align*}
    \begin{align*}
        \nonipay{l}_{exp}(\expelek{l}{k}, \noniele{l}_k)=&\sum_{j=1}^{\nonisetnum{l}} \left((1-\exppost_r^k)\cdot \nonidist{l}_0(j)+\exppost_r^k\cdot \nonidist{l}_1(j)\right)R_{sppm}^{l}(\noniele{l}_k, \noniele{l}_j)\\
        =&\sum_{j=1}^{\nonisetnum{l}} \left((1-\nonipost(\noniele{l}_k))\cdot \nonidist{l}_0(j)+\nonipost(\noniele{l}_k)\cdot \nonidist{l}_1(j)\right)R_{sppm}^{l}(\noniele{l}_k, \noniele{l}_j)\\
        =&\nonipay{l}_{non-l}(\nonisig{l}=\noniele{l}_k, \noniele{l}_k).
    \end{align*}
    Therefore, with the definition of $\noniB{l}_k$, we have 
    \begin{equation*}
        \exppay_{exp}(\expelek{l}{k}, \expelek{l}{k}) - \noniA{l}_k\nonipay{l}_{exp}(\expelek{l}{k},\noniele{l}_k)=\noniB{l}_k.
    \end{equation*}
    Now we find a special case of $\expsig_r=\expelek{l}{k}$ that satisfies the equality. Then we only need to prove that $\exppay_{exp}(\expsig_r, \expsig_r) - \noniA{l}_k\nonipay{l}_{exp}(\expsig_r,\noniele{l}_k)$ is uniquely minimized when $\expsig_r=\expelek{l}{k}$.
    This is proved in Lemma~\ref{lemma:noni-min}. 
    Therefore,  for all $s_r\in\expset$ the expert's exterior truthfulness for group 1 non-experts holds, and the equation holds if and only if $\expsig_r=\expelek{l}{k}$. 
\end{proof}

\subsubsection{Non-expert Part}
In this part we prove interior and exterior truthfulness for non-expert groups. W.l.o.g, we'll select one specific group (i.e group $l$ non-expert) to present our proof. In exterior truthfulness towards other non-expert groups where we need another non-expert group, we use group $h$ non-expert. 

Before we start the interior and exterior truthfulness, we first propose a lemma that will be used many times in the proof. This lemma shows that the expected payment of a non-expert reporting continuous answer is maximized when she continuous report has the same posterior distribution of ground truth with her private signal. 
\begin{lemma}
\label{lemma:noni-max}
For any group $l$ non-expert $r$, and any $k\in[\nonisetnum{l}], \exprep_r\in\expset$, we have 
\begin{equation*}
        \exppay_{non-l}(\noniele{l}_k,\expelek{l}{k})\ge \exppay_{non-l}(\noniele{l}_k,\exprep_r),
    \end{equation*}
    where $\expelek{l}{k}$ satisfies $\exppost(\expelek{l}{k})=\nonipost(\noniele{l}_k)$. The equality holds if and only if $\exprep=\expelek{l}{k}$. 
\end{lemma}
\begin{proof}
According to Equation~\ref{eq:CCsemi}, 
\begin{align*}
    \exppay_{non-l}(\noniele{l}_k,\expelek{l}{k})=&
    \pdfs[\expele\mid \noniSig{l}_r=\nonisig{l}_r]R_{ecgm}(\expelek{l}{k}, \expele){\rm d}\expele
    \\
     =&\int_{\expset}  \left((1-\nonipost(\noniele{l}_k))\cdot f_0(\expele) +\nonipost(\noniele{l}_k)\cdot f_1(\expele)\right)R_{ecgm}(\expelek{l}{k}, \expele){\rm d}\expele\\
     =&\int_{\expset} \left((1-\exppost(\expelek{l}{k}))\cdot f_0(\expele) +\exppost(\expelek{l}{k})\cdot f_1(\expele)\right)R_{ecgm}(\expelek{l}{k}, \expele){\rm d}\expele\\
     =&\exppay_{exp}(\expelek{l}{k}, \expelek{l}{k}).\\
     \exppay_{non-l}(\noniele{l}_k,\exprep_r)
     =&\exppay_{exp}(\expelek{l}{k}, \exprep_r).\\
\end{align*}
Therefore, we only need to prove that 
\begin{equation*}
    \exppay_{exp}(\expelek{l}{k}, \expelek{l}{k}) \ge \exppay_{exp}(\expelek{l}{k}, \exprep_r).
\end{equation*}
Note that this is proved in Lemma~\ref{lemma: exp_exp}. Therefore the original inequality holds. The equality holds if and only if the equality in Lemma~\ref{lemma: exp_exp} holds, i.e. $\exprep=\expelek{l}{k}$. 
\end{proof}

Now we start from the interior truthfulness for group 1 non-expert part. We need to consider the special cases of  $\nonidist{l}_0(j)=\nonidist{l}_1(j)$ and $\nonidist{l}_0(k)=\nonidist{l}_1(k)$. We will deal with different cases separately since $(\noniA{l}_i, \noniB{l}_i)_{i=j,k}$ are defined differently . But we share the same technique in all cases. We use the definition of$(\noniA{l}_i, \noniB{l}_i)_{i=j,k}$ , or the special $\expelek{l}{k}$ to convert continuous-report expected payments into discrete-report expected payments, and compare them in the same domain. 
\begin{lemma}[Non-expert Interior Truthfulness]
\label{lemma:non-int}
For any group $l$ non-expert $r$, and for any $j,k\in[\nonisetnum{l}]$, we have 
\begin{equation*}
    \noniA{l}_k\cdot\nonipay{l}_{non-l}(\noniele{l}_k, \noniele{l}_k)+\noniB{l}_k\ge \noniA{l}_j\cdot\nonipay{l}_{non-l}(\noniele{l}_k, \noniele{l}_j)+\noniB{l}_j.
\end{equation*}
The equality holds if and only if $j=k$. 
\end{lemma}

\begin{proof}
In this proof we use $\expelek{l}{k}$ to denote a special expert signal such that $\exppost(\expelek{l}{k})=\nonipost(\noniele{l}_k)$. And we use $\expelek{l}{j}$ to denote a special expert signal such that $\exppost(\expelek{l}{j})=\nonipost(\noniele{l}_j)$.
We process this proof in three different cases:
\begin{enumerate}
    \item $\nonidist{l}_0(j)=\nonidist{l}_1(j)$ and $\nonidist{l}_0(k)=\nonidist{l}_1(k)$
    \item $\nonidist{l}_0(j)=\nonidist{l}_1(j)$ but $\nonidist{l}_0(k)\neq\nonidist{l}_1(k)$ \item $\nonidist{l}_0(j)\neq\nonidist{l}_1(j)$
\end{enumerate}

\textbf{Case 1} $\nonidist{l}_0(j)=\nonidist{l}_1(j)$ and $\nonidist{l}_0(k)=\nonidist{l}_1(k)$
Note in this case $\noniA{l}_j=\noniB{l}_j=\noniA{l}_k=\noniB{l}_k=0. $ Therefore $\noniA{l}_k\cdot\nonipay{l}_{non-l}(\noniele{l}_k, \noniele{l}_k)+\noniB{l}_k$ and $ \noniA{l}_j\cdot\nonipay{l}_{non-l}(\noniele{l}_k, \noniele{l}_j)+\noniB{l}_j$ equals to zero. 

Then we show that such case only exists when $j=k$ according to Assumption~\ref{asm:ip}. Consider group $l$ non-expert $r$'s posterior distribution of another group $l$ non-expert's private signal $\noniSig{l}$:
\begin{align*}
    \Pr[\noniSig{l}=\noniele{l}_i\mid \noniSig{l}_r=\noniele{l}_k] =& (1-\nonipost(\noniele{l}_k)) \nonidist{l}_0(i) + \nonipost(\noniele{l}_k)\cdot \nonidist{l}_1(i). \\
   \Pr[\noniSig{l}=\noniele{l}_i\mid \noniSig{l}_r=\noniele{l}_j] =& (1-\nonipost(\noniele{l}_j)) \nonidist{l}_0(i) + \nonipost(\noniele{l}_j)\cdot \nonidist{l}_1(i).
\end{align*}
Note that when $\nonidist{l}_0(k)=\nonidist{l}_1(k)$, 
\begin{equation*}
    \nonipost(\noniele{l}_k) = \Pr[Y=1\mid \noniSig{l}=\noniele{l}_k]=\frac{P\cdot \nonidist{l}_1(k)}{(1-P)\nonidist{l}_0(k)+P\cdot \nonidist{l}_1(k)}=P.
\end{equation*}
Similarly, $\nonipost(\noniele{l}_j)=P.$Therefore, we have $\Pr[\noniSig{l}=\noniele{l}_i\mid \noniSig{l}_r=\noniele{l}_k] = \Pr[\noniSig{l}=\noniele{l}_i\mid \noniSig{l}_r=\noniele{l}_j]$ for every $i\in [\nonisetnum{l}]$. According to Assumption~\ref{asm:ip}, same posterior distributions lead to same private signal. Therefore $j=k$. 

\textbf{Case 2} $\nonidist{l}_0(j)=\nonidist{l}_1(j)$ but $\nonidist{l}_0(k)\neq\nonidist{l}_1(k)$.In this part the right side of the inequality ($\noniA{l}_j\cdot\nonipay{l}_{non-l}(\noniele{l}_k, \noniele{l}_j)+\noniB{l}_j$) equals to 0. Recall the definition of $\noniB{l}_k$ in Equation~\ref{eq:noniB}:
\begin{equation*}
        \noniB{l}_k = \exppay_{non-l}(\noniele{l}_k,\expelek{l}{k}) - \noniA{l}_k\cdot \nonipay{l}_{non-l}(\noniele{l}_k,\noniele{l}_k)
    \end{equation*}
    where $\expelek{l}{k}$ satisfies $\exppost(\expelek{l}{k})=\nonipost(\noniele{l}_k)$.
Therefore, we only need to prove that 
\begin{equation*}
    \exppay_{non-l}(\noniele{l}_k,\expelek{l}{k})>0.
\end{equation*}
From Lemma~\ref{lem:lsr}, with the same technique of Lemma~\ref{lemma: exp_exp} we have 
\begin{align*}
    &\exppay_{non-l}(\noniele{l}_k,\expelek{l}{k})\\
    = & \int_{\expset}\left((1-\nonipost(\noniele{l}_k))\cdot f_0(\expele) +\nonipost(\noniele{l}_k)\cdot f_1(\expele)\right)  \log \left( \frac{(1-\nonipost(\noniele{l}_k))\cdot f_0(\expele) +\nonipost(\noniele{l}_k)\cdot f_1(\expele)}{(1-P)\cdot f_0(\expele) +P\cdot f_1(\expele)}\right) {\rm d}\expele\\
   \ge & \int_{\expset}\left((1-\nonipost(\noniele{l}_k))\cdot f_0(\expele) +\nonipost(\noniele{l}_k)\cdot f_1(\expele)\right)  \log \left( \frac{(1-P)\cdot f_0(\expele) +P\cdot f_1(\expele)}{(1-P)\cdot f_0(\expele) +P\cdot f_1(\expele)}\right) {\rm d}\expele\\
    = & 0.
\end{align*}
Moreover, since $\nonidist{l}_0(k)\neq\nonidist{l}_1(k)$, $\nonipost(\noniele{l}_k)\neq P$, the equality does not hold. Therefore, $\exppay_{non-l}(\noniele{l}_k,\expelek{l}{k})>0.$

\textbf{Case 3}  $\nonidist{l}_0(j)\neq\nonidist{l}_1(j)$. First we show that this case we can convert the original inequality to 
\begin{equation*}
    \exppay_{non-l}(\noniele{l}_k,\expelek{l}{k})\ge\noniA{l}_j\cdot\nonipay{l}_{non-l}(\noniele{l}_k,\noniele{l}_j)+\noniB{l}_j.
\end{equation*}
When $\nonidist{l}_0(k)\neq\nonidist{l}_1(k)$, we expand $\noniB{l}_k$ as in case 2 and get what we want. 
\begin{equation*}
        \noniB{l}_k = \exppay_{non-l}(\noniele{l}_k,\expelek{l}{k}) - \noniA{l}_k\cdot \nonipay{l}_{non-l}(\noniele{l}_k,\noniele{l}_k). 
    \end{equation*}
When $\nonidist{l}_0(k)=\nonidist{l}_1(k)$, the left side of the original inequality 
\begin{equation*}
    \noniA{l}_k\cdot \nonipay{l}_{non-l}(\noniele{l}_k,\noniele{l}_k) + \noniB{l}_k = 0.
\end{equation*}
We show that in this case $ \exppay_{non-l}(\noniele{l}_k,\expelek{l}{k})=0$ as well, and we can still convert the inequality. 
Note that since $\nonidist{l}_0(k)=\nonidist{l}_1(k)$, $\exppost(\expelek{l}{k})=\nonipost(\noniele{l}_k)=P$. Therefore, 
\begin{align*}
    &\exppay_{non-l}(\noniele{l}_k,\expelek{l}{k})\\
    = & \int_{\expset}\left((1-\nonipost(\noniele{l}_k))\cdot f_0(\expele) +\nonipost(\noniele{l}_k)\cdot f_1(\expele)\right)  \log \left( \frac{(1-\nonipost(\noniele{l}_k))\cdot f_0(\expele) +\nonipost(\noniele{l}_k)\cdot f_1(\expele)}{(1-P)\cdot f_0(\expele) +P\cdot f_1(\expele)}\right) {\rm d}\expele\\
    = & \int_{\expset}\left((1-\nonipost(\noniele{l}_k))\cdot f_0(\expele) +\nonipost(\noniele{l}_k)\cdot f_1(\expele)\right)  \log \left( \frac{(1-P)\cdot f_0(\expele) +P\cdot f_1(\expele)}{(1-P)\cdot f_0(\expele) +P\cdot f_1(\expele)}\right) {\rm d}\expele\\
    = & 0.
\end{align*}

Then we will start from the right side  ($\noniA{l}_j\cdot\nonipay{l}_{non-l}(\noniele{l}_k,\noniele{l}_j)+\noniB{l}_j$) and finally reach the left ($\exppay_{non-l}(\noniele{l}_k,\expelek{l}{k})$). First expand the $\noniB{l}_j$ according to Equation~\ref{eq:noniB}, and the right side of the inequality equals to
\begin{equation*}
    \noniA{l}_j\cdot\nonipay{l}_{non-l}(\noniele{l}_k,\noniele{l}_j)+\exppay_{non-l}(\noniele{l}_j,\expelek{l}{j}) - \noniA{l}_j\cdot \nonipay{l}_{non-l}(\noniele{l}_j,\noniele{l}_j).
\end{equation*}
Then we look at the difference between $\nonipay{l}_{non-l}(\noniele{l}_k,\noniele{l}_j)$ and $\nonipay{l}_{non-l}(\noniele{l}_j,\noniele{l}_j)$.
\begin{align*}
   \nonipay{l}_{non-l}(\noniele{l}_k,\noniele{l}_j) =&\sum_{i=1}^{\nonisetnum{l}} \left((1-\nonipost(\noniele{l}_k))\cdot \nonidist{l}_0(i)+\nonipost(\noniele{l}_k)\cdot \nonidist{l}_1(i)\right)R_{sppm}^{l}(\noniele{l}_j, \noniele{l}_i)\\
    \nonipay{l}_{non-l}(\noniele{l}_j,\noniele{l}_j) =&\sum_{i=1}^{\nonisetnum{l}} \left((1-\nonipost(\noniele{l}_j))\cdot \nonidist{l}_0(i)+\nonipost(\noniele{l}_j)\cdot \nonidist{l}_1(i)\right)R_{sppm}^{l}(\noniele{l}_j, \noniele{l}_i)
\end{align*}
Therefore, 
\begin{align*}
   &\nonipay{l}_{non-l}(\noniele{l}_k,\noniele{l}_j) -  \nonipay{l}_{non-l}(\noniele{l}_j,\noniele{l}_j) \\
   = & \sum_{i=1}^{\nonisetnum{l}} \left((\nonipost(\noniele{l}_j)-\nonipost(\noniele{l}_k))\cdot \nonidist{l}_0(i)+(\nonipost(\noniele{l}_k) - \nonipost(\noniele{l}_j))\cdot \nonidist{l}_1(i)\right)R_{sppm}^{l}(\noniele{l}_j, \noniele{l}_i)\\
    = & (\nonipost(\noniele{l}_k) - \nonipost(\noniele{l}_j))\sum_{i=1}^{\nonisetnum{l}} \left(\nonidist{l}_1(i)-\nonidist{l}_0(i)\right)R_{sppm}^{l}(\noniele{l}_j, \noniele{l}_i)
\end{align*}
Recall the derivative of $\nonipay{l}_{exp}$ (Equation~\ref{eq:noni-partial}):
\begin{equation*}
     \frac{\partial  \nonipay{l}_{exp}(\expsig_r, \noniele{l}_j)}{\partial \exppost_r} = \sum_{i=1}^{\nonisetnum{l}}  \left(\nonidist{l}_1(i)-\nonidist{l}_0(i)\right) R_{sppm}^{l}(\noniele{l}_j, \noniele{l}_i). 
\end{equation*}
Expand $\noniA{l}_j$ according to Equation~\ref{eq:noniA}, and we get 
\begin{align*}
    &\noniA{l}_j\cdot\left(\nonipay{l}_{non-l}(\noniele{l}_k,\noniele{l}_j) -  \nonipay{l}_{non-l}(\noniele{l}_j,\noniele{l}_j)\right) \\
    = & \left.\frac{\frac{{\rm d}\exppay_{exp}(\expsig_r, \expsig_r)}{{\rm d}\exppost_r}}{\frac{\partial  \nonipay{l}_{exp}(\expsig_r, \noniele{l}_j)}{\partial \exppost_r}}\right|_{\exppost_r=\nonipost(\noniele{l}_j)} \cdot (\nonipost(\noniele{l}_k) - \nonipost(\noniele{l}_j))\cdot \frac{\partial  \nonipay{l}_{exp}(\expsig_r, \noniele{l}_j)}{\partial \exppost_r}\\
    = &(\nonipost(\noniele{l}_k) - \nonipost(\noniele{l}_j))\cdot  \left.\frac{{\rm d}\exppay_{exp}(\expsig_r, \expsig_r)}{{\rm d}\exppost_r}\right|_{\exppost_r=\nonipost(\noniele{l}_j)}.
\end{align*}
From Equation~\ref{eq:exp-partial} (the derivative of $\exppay_{exp}$), we have 
\begin{align*}
    & (\nonipost(\noniele{l}_k) - \nonipost(\noniele{l}_j))\cdot  \left.\frac{{\rm d}\exppay_{exp}(\expsig_r, \expsig_r)}{{\rm d}\exppost_r}\right|_{\exppost_r=\nonipost(\noniele{l}_j)}\\
    = & (\nonipost(\noniele{l}_k) - \nonipost(\noniele{l}_j))\cdot  \left.\frac{{\rm d} \exppay^*(p)}{{\rm d} p}\right|_{p=\nonipost(\noniele{l}_j)}\\
    = &(\nonipost(\noniele{l}_k) - \nonipost(\noniele{l}_j))\cdot  \int_{\expset} \left(f_1(\expele)-f_0(\expele)\right)  \log \left( \frac{(1-\nonipost(\noniele{l}_j))\cdot f_0(\expele) +\nonipost(\noniele{l}_j)\cdot f_1(\expele)}{(1-P)\cdot f_0(\expele) +P\cdot f_1(\expele)}\right) {\rm d}\expele.
\end{align*}
Now we split the above term into the difference of two terms, one is exactly $\exppay_{non-l}(\noniele{l}_j, \expelek{l}{j})$, and  the other is smaller than $\exppay_{non-l}(\noniele{l}_k, \expelek{l}{k})$. Note that 
\begin{align*}
   (\nonipost(\noniele{l}_k) - \nonipost(\noniele{l}_j))\cdot \left(f_1(\expele)-f_0(\expele)\right) = & \left((1-\nonipost(\noniele{l}_k) )\cdot f_0(\expele) + \nonipost(\noniele{l}_k) \cdot f_1(\expele) \right)\\
   - & \left((1-\nonipost(\noniele{l}_j))\cdot f_0(\expele) + \nonipost(\noniele{l}_j)\cdot f_1(\expele) \right)
\end{align*}
Therefore, 
\begin{align*}
    &(\nonipost(\noniele{l}_k) - \nonipost(\noniele{l}_j))\cdot  \int_{\expset} \left(f_1(\expele)-f_0(\expele)\right)  \log \left( \frac{(1-\nonipost(\noniele{l}_j))\cdot f_0(\expele) +\nonipost(\noniele{l}_j)\cdot f_1(\expele)}{(1-P)\cdot f_0(\expele) +P\cdot f_1(\expele)}\right) {\rm d}\expele
    \\ 
    = & \int_{\expset}\left((1-\nonipost(\noniele{l}_k) )\cdot f_0(\expele) + \nonipost(\noniele{l}_k) \cdot f_1(\expele) \right)\log \left( \frac{(1-\nonipost(\noniele{l}_j))\cdot f_0(\expele) +\nonipost(\noniele{l}_j)\cdot f_1(\expele)}{(1-P)\cdot f_0(\expele) +P\cdot f_1(\expele)}\right) {\rm d}\expele\\
    - & \int_{\expset} \left((1-\nonipost(\noniele{l}_j))\cdot f_0(\expele) + \nonipost(\noniele{l}_j)\cdot f_1(\expele) \right)\log \left( \frac{(1-\nonipost(\noniele{l}_j))\cdot f_0(\expele) +\nonipost(\noniele{l}_j)\cdot f_1(\expele)}{(1-P)\cdot f_0(\expele) +P\cdot f_1(\expele)}\right) {\rm d}\expele\\
    = & \exppay_{non-l}(\noniele{l}_k, \expelek{l}{j}) - \exppay_{non-l}(\noniele{l}_j, \expelek{l}{j}).
\end{align*}
Therefore, going back to the right side,we have 
\begin{align*}
    &\noniA{l}_j\cdot\nonipay{l}_{non-l}(\noniele{l}_k,\noniele{l}_j)+\exppay_{non-l}(\noniele{l}_j,\expelek{l}{j}) - \noniA{l}_j\cdot \nonipay{l}_{non-l}(\noniele{l}_j,\noniele{l}_j) \\
    = & \exppay_{non-l}(\noniele{l}_j,\expelek{l}{j}) + \exppay_{non-l}(\noniele{l}_k, \expelek{l}{j}) - \exppay_{non-l}(\noniele{l}_j, \expelek{l}{j})\\
    = & \exppay_{non-l}(\noniele{l}_k, \expelek{l}{j}).
\end{align*}
Note that from Lemma~\ref{lemma:noni-max} we have 
\begin{equation*}
    \exppay_{non-l}(\noniele{l}_k, \expelek{l}{j})\le \exppay_{non-l}(\noniele{l}_k, \expelek{l}{k}).
\end{equation*}
and the equality holds if and only if $j=k$. Therefore, the original inequality 
\begin{equation*}
    \exppay_{non-l}(\noniele{l}_k, \expelek{l}{j})\ge\noniA{l}_j\cdot\nonipay{l}_{non-l}(\noniele{l}_k,\noniele{l}_j)+\noniB{l}_j.
\end{equation*}
holds, and the equality holds if and only if $j=k$.

\end{proof}

Then we consider the non-expert's exterior truthfulness towards experts. In this case we also need to consider the special case of $\nonidist{l}_0(k)=\nonidist{l}_1(k)$. When $\nonidist{l}_0(k)\neq\nonidist{l}_1(k)$ we can convert $\nonipay{l}_{non-l}$ into $\exppay_{non-l}$, and then take advantage of Lemma~\ref{lemma:noni-max}. When $\nonidist{l}_0(k)=\nonidist{l}_1(k)$, we prove that the inequality can still be converted to the Lemma~\ref{lemma:noni-max} form. 
\begin{lemma}[Non-expert Exterior Truthfulness-Experts]
For any $l\in[\noninum]$ and any group $l$ non-expert $r$, and for any $k\in[\nonisetnum{l}], \exprep_r\in\expset$, we have
\begin{equation*}
    \noniA{l}_k\cdot \nonipay{l}_{non-l}(\noniele{l}_k,\noniele{l}_k)+\noniB{l}_k\ge \exppay_{non-l}(\noniele{l}_k,\exprep_r).
\end{equation*}
The equation holds if and only if $\exprep_r=\expelek{l}{k}$, where $\exppost(\expelek{l}{k})=\nonipost(\noniele{l}_k).$
\end{lemma}
\begin{proof}
We process this proof in two different cases: $\nonidist{l}_0(k)\neq\nonidist{l}_1(k)$ and $\nonidist{l}_0(k)=\nonidist{l}_1(k)$.

\textbf{Case 1} $\nonidist{l}_0(k)\neq\nonidist{l}_1(k)$  From the definition of $\noniB{l}_k$ we have:
\begin{equation*}
        \noniB{l}_k = \exppay_{non-l}(\noniele{l}_k,\expelek{l}{k}) - \noniA{l}_k\cdot \nonipay{l}_{non-l}(\noniele{l}_k,\noniele{l}_k)
    \end{equation*}
    where $\expelek{l}{k}$ satisfies $\exppost(\expelek{l}{k})=\nonipost(\noniele{l}_k)$. Therefore, we only need to prove that 
    \begin{equation*}
        \exppay_{non-l}(\noniele{l}_k,\expelek{l}{k})\ge \exppay_{non-l}(\noniele{l}_k,\exprep_r).
    \end{equation*}
    Note that this is proved in Lemma~\ref{lemma:noni-max}, and the equality holds if and only if $\exprep_r=\expelek{l}{k}$. 
    
    \textbf{Case 2} $\nonidist{l}_0(k)=\nonidist{l}_1(k)$. In this case since $\noniA{l}_k=\noniB{l}_k=0$, we only need to prove that $\exppay_{non-l}(\noniele{l}_k,\exprep_r)\le 0$. Since we have proved that $ \exppay_{non-l}(\noniele{l}_k,\exprep_r)\le \exppay_{non-l}(\noniele{l}_k,\expelek{l}{k}) $ in Lemma~\ref{lemma:noni-max} (and the equality holds if and only if $\exprep_r=\expelek{l}{k}$ ), we only need to show that $\exppay_{non-l}(\noniele{l}_k,\expelek{l}{k})=0$ . Note that when $\nonidist{l}_0(k)=\nonidist{l}_1(k)$, $\nonipost(\noniele{l}_k)=P$. Therefore, \begin{align*}
    &\exppay_{non-l}(\noniele{l}_k,\expelek{l}{k})\\
    = & \int_{\expset}\left((1-\nonipost(\noniele{l}_k))\cdot f_0(\expele) +\nonipost(\noniele{l}_k)\cdot f_1(\expele)\right)  \log \left( \frac{(1-\nonipost(\noniele{l}_k))\cdot f_0(\expele) +\nonipost(\noniele{l}_k)\cdot f_1(\expele)}{(1-P)\cdot f_0(\expele) +P\cdot f_1(\expele)}\right) {\rm d}\expele\\
    = & \int_{\expset}\left((1-\nonipost(\noniele{l}_k))\cdot f_0(\expele) +\nonipost(\noniele{l}_k)\cdot f_1(\expele)\right)  \log \left( \frac{(1-P)\cdot f_0(\expele) +P\cdot f_1(\expele)}{(1-P)\cdot f_0(\expele) +P\cdot f_1(\expele)}\right) {\rm d}\expele\\
    = & 0.
\end{align*}
    Therefore, 
    \begin{equation*}
        \exppay_{non-l}(\noniele{l}_k,\exprep_r)\le \exppay_{non-l}(\noniele{l}_k,\expelek{l}{k})=0. 
    \end{equation*}
    and the equality holds if and only if $\exprep_r=\expelek{l}{k}$. 
\end{proof}

Then we prove non-expert's exterior truthfulness towards another non-expert's group. The proof shares a similar process and techniques with non-expert's interior truthfulness. 
\begin{lemma}[Non-expert Exterior Truthfulness-non-experts]
For any $l,h\in [\noninum]$, and for any group $l$ non-expert $r$, and for any $k\in[\nonisetnum{l}], j\in[\nonisetnum{h}]$, we have

\begin{equation*}
    \noniA{l}_k\cdot \nonipay{l}_{non-l}(\noniele{l}_k,\noniele{l}_k)+\noniB{l}_k\ge \noniA{h}_j\cdot \nonipay{h}_{non-l}(\noniele{l}_k,\noniele{h}_j)+\noniB{h}_j.
\end{equation*}
The equality holds if and only if $\nonipost(\noniele{l}_k)=\nonipost(\noniele{h}_j)$. 
\end{lemma}
\begin{proof}
In this proof we use $\expelek{l}{k}$ to denote a special expert signal such that $\exppost(\expelek{l}{k})=\nonipost(\noniele{l}_k)$. And we use $\expelek{h}{j}$ to denote a special expert signal such that $\exppost(\expelek{h}{j})=\nonipost(\noniele{h}_j)$. {\em Note that in this proof $\expelek{l}{k}$ and $\expelek{h}{j}$ are corresponding to different non-expert groups.}
We process this proof in three different cases:
\begin{enumerate}
    \item $\nonidist{h}_0(j)=\nonidist{h}_1(j)$ and $\nonidist{l}_0(k)=\nonidist{l}_1(k)$
    \item $\nonidist{h}_0(j)=\nonidist{h}_1(j)$ but $\nonidist{l}_0(k)\neq\nonidist{l}_1(k)$ \item $\nonidist{h}_0(j)\neq\nonidist{}_1(j)$
\end{enumerate}

\textbf{Case 1} $\nonidist{h}_0(j)=\nonidist{h}_1(j)$ and $\nonidist{l}_0(k)=\nonidist{l}_1(k)$. 
In this case both $\noniA{l}_k\cdot \nonipay{l}_{non-l}(\noniele{l}_k,\noniele{l}_k)+\noniB{l}_k$ and $ \noniA{h}_j\cdot \nonipay{h}_{non-l}(\noniele{l}_k,\noniele{h}_j)+\noniB{h}_j$ equals to 0. Similar to Lemma~\ref{lemma:non-int},  we have $\nonipost(\noniele{l}_k)=\nonipost(\noniele{h}_j)$.

\textbf{Case 2} $\nonidist{h}_0(j)=\nonidist{h}_1(j)$ but $\nonidist{l}_0(k)\neq\nonidist{l}_1(k)$. In this case $\noniA{h}_j\cdot \nonipay{h}_{non-l}(\noniele{l}_k,\noniele{h}_j)+\noniB{h}_j$ equals to zero. 
Recall the definition of $\noniB{l}_k$ in Equation~\ref{eq:noniB}:
\begin{equation*}
        \noniB{l}_k = \exppay_{non-l}(\noniele{l}_k,\expelek{l}{k}) - \noniA{l}_k\cdot \nonipay{l}_{non-l}(\noniele{l}_k,\noniele{l}_k)
    \end{equation*}
    where $\expelek{l}{k}$ satisfies $\exppost(\expelek{l}{k})=\nonipost(\noniele{l}_k)$.
Therefore, we only need to prove that 
\begin{equation*}
    \exppay_{non-l}(\noniele{l}_k,\expelek{l}{k})> 0.
\end{equation*}
This has been proved in the case 2 of Lemma~\ref{lemma:non-int} (interior truthfulness of non-expert). 

\textbf{Case 3} 
$\nonidist{h}_0(j)\neq\nonidist{h}_1(j)$
First we show that this case we can convert the original inequality to 
\begin{equation}
\label{eq:case3}
    \exppay_{non-l}(\noniele{l}_k,\expelek{l}{k})\ge \noniA{h}_j\cdot \nonipay{h}_{non-l}(\noniele{l}_k,\noniele{h}_j)+\noniB{h}_j.
\end{equation}
When $\nonidist{l}_0(k)\neq\nonidist{l}_1(k)$, we expand $\noniB{l}_k$ as in case 2 and get what we want. 
\begin{equation*}
        \noniB{l}_k = \exppay_{non-l}(\noniele{l}_k,\expelek{l}{k}) - \noniA{l}_k\cdot \nonipay{l}_{non-l}(\noniele{l}_k,\noniele{l}_k) .
    \end{equation*}
When $\nonidist{l}_0(k)=\nonidist{l}_1(k)$
\begin{equation*}
    \noniA{l}_k\cdot \nonipay{l}_{non-l}(\noniele{l}_k,\noniele{l}_k) + \noniB{l}_k = 0.
\end{equation*}
We show that in this case $ \exppay_{non-l}(\noniele{l}_k,\expelek{l}{k})$ as well. Note that since $\nonidist{l}_0(k)=\nonidist{l}_1(k)$, $\exppost(\expelek{l}{k})=\nonipost(\noniele{l}_k)=P$. Therefore, 
\begin{align*}
    &\exppay_{non-l}(\noniele{l}_k,\expelek{l}{k})\\
    = & \int_{\expset}\left((1-\nonipost(\noniele{l}_k))\cdot f_0(\expele) +\nonipost(\noniele{l}_k)\cdot f_1(\expele)\right)  \log \left( \frac{(1-\nonipost(\noniele{l}_k))\cdot f_0(\expele) +\nonipost(\noniele{l}_k)\cdot f_1(\expele)}{(1-P)\cdot f_0(\expele) +P\cdot f_1(\expele)}\right) {\rm d}\expele\\
    = & \int_{\expset}\left((1-\nonipost(\noniele{l}_k))\cdot f_0(\expele) +\nonipost(\noniele{l}_k)\cdot f_1(\expele)\right)  \log \left( \frac{(1-P)\cdot f_0(\expele) +P\cdot f_1(\expele)}{(1-P)\cdot f_0(\expele) +P\cdot f_1(\expele)}\right) {\rm d}\expele\\
    = & 0.
\end{align*}

Then we change inequality's form into the form of Lemma~\ref{lemma:noni-min}. Recall the definition of $\noniB{h}_j$:
\begin{equation*}
        \noniB{h}_j = \exppay_{non-h}(\noniele{h}_j,\expelek{h}{j}) - \noniA{h}_j\cdot \nonipay{h}_{non-h}(\noniele{h}_j,\noniele{h}_j)
    \end{equation*}
    where $\expelek{h}{j}$ satisfies $\exppost(\expelek{h}{j})=\nonipost(\noniele{h}_j)$.
Therefore, expanding $\noniB{l}_k$ and $\noniB{h}_j$,  the equality can be written as
\begin{equation*}
     \exppay_{non-l}(\noniele{l}_k,\expelek{l}{k}) - \noniA{h}_j\cdot \nonipay{h}_{non-l}(\noniele{l}_k,\noniele{h}_j)\ge  \exppay_{non-h}(\noniele{h}_j,\expelek{h}{j}) - \noniA{h}_j\cdot \nonipay{h}_{non-h}(\noniele{h}_j,\noniele{h}_j).
\end{equation*}
We use the Assumption~\ref{asm:wise} and rewrite the expected payments to an equivalent form of expert's expected payments. Then by Lemma~\ref{lemma:noni-min} we prove the inequality. 

For the group $l$ non-expert part, since $\expelek{l}{k}$ satisfies $\exppost(\expelek{l}{k})=\nonipost(\noniele{l}_k)$, we have 
\begin{align*}
     \exppay_{non-l}(\noniele{l}_k,\expelek{l}{k}) =&\int_{\expset}  \left((1-\nonipost(\noniele{l}_k))\cdot f_0(\expele) +\nonipost(\noniele{l}_k)\cdot f_1(\expele)\right)R_{ecgm}(\expelek{l}{k}, \expele){\rm d}\expele\\
     =&\int_{\expset} \left((1-\exppost(\expelek{l}{k}))\cdot f_0(\expele) +\exppost(\expelek{l}{k})\cdot f_1(\expele)\right)R_{ecgm}(\expelek{l}{k}, \expele){\rm d}\expele\\
     =&\exppay_{exp}(\expelek{l}{k}, \expelek{l}{k}).\\
     \nonipay{h}_{non-l}(\noniele{l}_k,\noniele{h}_j) =&
     \sum_{i=1}^{\nonisetnum{h}}\left((1-\nonipost(\noniele{l}_k))\cdot \nonidist{h}_0(i) +\nonipost(\noniele{l}_k)\cdot \nonidist{h}_1(i) \right) R_{sppm}^{h}(\noniele{h}_j,\noniele{h}_i)\\
     = & \sum_{i=1}^{\nonisetnum{h}}\left((1-\exppost(\expelek{l}{k}))\cdot \nonidist{h}_0(i) +\exppost(\expelek{l}{k})\cdot \nonidist{h}_1(i) \right) R_{sppm}^{h}(\noniele{h}_j,\noniele{h}_i)\\
     = & \nonpay_{exp}(\expelek{l}{k}, \noniele{h}_j). 
\end{align*}
For the group $h$ non-expert part, since $\expelek{h}{j}$ satisfies $\exppost(\expelek{h}{j})=\nonipost(\noniele{h}_j)$, we have 
\begin{align*}
    \exppay_{non-h}(\noniele{h}_j,\expelek{h}{j})=
     &\int_{\expset}  \left((1-\nonipost(\noniele{h}_j))\cdot f_0(\expele) +\nonipost(\noniele{h}_j)\cdot f_1(\expele)\right)R_{ecgm}(\expelek{h}{j}, \expele){\rm d}\expele\\
     =&\int_{\expset} \left((1-\exppost(\expelek{h}{j}))\cdot f_0(\expele) +\exppost(\expelek{h}{j})\cdot f_1(\expele)\right)R_{ecgm}(\expelek{h}{j}, \expele){\rm d}\expele\\
     =&\exppay_{exp}(\expelek{h}{j}, \expelek{h}{j}).\\
     \nonpay_{non-h}(\noniele{h}_j,\noniele{h}_j) =&
     \sum_{i=1}^{\nonisetnum{h}}\left((1-\nonipost(\noniele{h}_j))\cdot \nonidist{h}_0(i) +\nonipost(\noniele{h}_j)\cdot \nonidist{h}_1(i) \right) R_{sppm}^{h}(\noniele{}_j,\noniele{h}_i)\\
     = & \sum_{i=1}^{\nonisetnum{h}}\left((1-\exppost(\expelek{h}{j}))\cdot \nonidist{h}_0(i) +\exppost(\expelek{h}{j})\cdot \nonidist{h}_1(i) \right) R_{sppm}^{h}(\noniele{}_j,\noniele{h}_i)\\
     = & \nonipay{h}_{exp}(\expelek{h}{j}, \noniele{h}_j). 
\end{align*}
Therefore, we can further rewrite the original inequality into the expert's expected payment form: 

\begin{equation*}
   \exppay_{exp}(\expelek{l}{k}, \expelek{l}{k}) - \noniA{h}_j\cdot \nonipay{h}_{exp}(\expelek{l}{k}, \noniele{h}_j)\ge  \exppay_{exp}(\expelek{h}{j}, \expelek{h}{j}) - \noniA{h}_j\cdot \nonipay{h}_{exp}(\expelek{h}{j}, \noniele{h}_j).
\end{equation*}

From Lemma~\ref{lemma:noni-min}, we know that this inequality holds, and equality holds if and only if $\exppost(\expelek{l}{k})=\exppost(\expelek{h}{j})$. There for the original inequality also holds, and the equality holds if and only if $\nonipost(\noniele{l}_k)=\nonipost(\noniele{h}_j)$. 

\end{proof}

\section{Running Example}
\label{sec:runexp}
In this section we give a example of how two mechanisms calculates the payments of different agents. 
\subsection{Setting}
Consider the weather prediction problem. Let $Y=0$  and $Y=1$ denote the weather be rainy and  sunny, respectively. Let the prior distribution be uniform, i.e.~$P=0.5$.  Let $\expset=[-100,100]$ be experts' signal space.  Let $f_0\sim N(-100,100^2)$ (respectively, $f_1\sim N(100,100^2)$) denote the restriction of the Gaussian whose mean is $-100$ (respectively, $100$) and whose variance is $100^2$ on $\expset$. 

Suppose group 1 non-experts  receive  signals about the intervals of weather. Let $\noniset{1}=\{[-100,-50),[-50,0),[0,50),[50,100]\}$ denoting {Very Likely Rainy, Likely Rainy, Likely Sunny, Very Likely Sunny}. When $Y=0$, let the signal distribution be $\nonidist{1}_0 =(0.3,0.4,0.2,0.1)$; and when $Y=1$, let the signal distribution be $\nonidist{1}_1=(0.1,0.2,0.4,0.3)$. 

Suppose non-experts group 2 receive signals about a more coarse set of intervals of weather.  Let $\noniset{2}=\{\text{Rainy, Sunny}\}$. When $Y=0$, let the signal distribution be $\nonidist{2}_0 =(0.3,0.7)$; and when $Y=1$, let the signal distribution be  $\nonidist{2}_1 =(0.7,0.3)$. %

We list three pairs of agents of each signal space in CEM in Table~\ref{tab:ce}. Note that there may be more than two agents in each signal space and we only focus on these six agents.

\subsection{Composite Elicitation Mechanism}
\begin{table}[ht]
\centering
\begin{tabular}{|l|l|l|l|}
\hline
Agents         & Report domain & Report       & Posterior distribution on ground truth  \\ \hline
$\expreporterA$  & $\expset$     & $20$         & $\exppost_A^0\approx 0.60$ \\ \hline
$\expreporterB$  & $\expset$     & $70$         & $\exppost_B^0\approx 0.80$ \\ \hline
$\semireporterA$ & $\noniset{1}$    & $[-100,-50)$ & $\nonipost^{1}_A=0.25$       \\ \hline
$\semireporterB$ & $\noniset{1}$    & $[0,50)$     & $\nonipost^{1}_B=0.75$       \\ \hline
$\nonreporterA$  & $\noniset{2}$     & $up$         & $\nonipost^{2}_A=0.75$        \\ \hline
$\nonreporterB$  & $\noniset{2}$     & $down$       & $\nonipost^{2}_B=0.25$        \\ \hline
\end{tabular}
\caption{Notations, report domains and reports of six reporters. \label{tab:ce}}
\end{table}
This section we give a running example of how composite elicitation mechanism compute payment of different reporters. 

We first compute the agents' posterior distribution on ground truth. We take $\semireporterA$ as an example. 
\begin{equation*}
 \nonipost_A^1=\frac{0.5\times 0.1}{0.5\times 0.3+0.5\times 0.1}=0.25. 
\end{equation*}
Other's posterior distribution are in the fourth column of Table~\ref{tab:ce}. 

Then we compute the payment of $\expset$ reporters. Since $\expreporterA$ and $\expreporterB$ are the peers of each other, and both of their payments are their ECGM mutual payments. 
\begin{equation*}
    Pay(\expreporterA)=Pay(\expreporterB)=R_{ecgm}(20,70)= \log \left( \frac{0.6\times 0.8}{0.5} + \frac{0.4\times0.2}{0.5} \right) = 0.22. 
\end{equation*}
For $\noniset{1}$ and $\noniset{2}$ reporters, we need to compute their linear transformation. We first compute payment for $\noniset{2}$ reporters for simplicity.  Here we used the linear transformation coefficients that mentioned in the proof.
According to the setting, $\nonidist{2}_0(1)=0.75, \nonidist{2}_1(1)=0.25, \nonidist{2}_0(2)=0.25, \nonidist{2}_1(2)=0.75$. Therefore, both cases are non-trivial. We use formula~\ref{eq:noniA} and \ref{eq:noniB} to calculate the coefficients. 
\begin{equation*}
    \noniA{2}_1=\left.\frac{\frac{{\rm d}\exppay_{exp}(\expsig_r, \expsig_r)}{{\rm d}\exppost(\expsig_r)}}{\frac{\partial \nonipay{2}_{exp}(\expsig_r,down)}{\partial \exppost(\expsig_r)}}\right|_{\exppost(\expsig_r) = \nonipost(down)}
    \approx 0.45 
\end{equation*}
\begin{equation*}
    \noniB{2}_1 = \exppay_{non-2}(down,\expelek{2}{1}) - \noniA{2}_1\cdot \nonipay{2}_{non-2}(down, down) \approx 8\times10^{-5}. 
\end{equation*}
\begin{equation*}
    \noniA{2}_2=\left.\frac{\frac{{\rm d}\exppay_{exp}(\expsig_r, \expsig_r)}{{\rm d}\exppost(\expsig_r)}}{\frac{\partial \nonipay{2}_{exp}(\expsig_r,up)}{\partial \exppost(\expsig_r)}}\right|_{\exppost(\expsig_r) = \nonipost(up)} \approx 0.45 
\end{equation*}
\begin{equation*}
    \noniB{2}_2 = \exppay_{non-2}(up,\expelek{2}{2}) - \noniA{2}_1\cdot \nonipay{2}_{non-2}(up,up) \approx 8\times10^{-5}. 
\end{equation*}
Then we compute their SPPM mutual payment. For simplicity, we use log scoring rule in SPPM. 
\begin{equation*}
    R_{sppm}^{2}(up, down) = \log \left(\frac{\Pr[\noniSig{2}=down\mid \noniSig{2}_r=up]}{\Pr[\noniSig{2}=down]}\right) =0.375.  
\end{equation*}
Then we can calculate the payment for both $\noniset{2}$ reporter:
\begin{align*}
    Pay(\nonreporterA)=\noniA{2}_1\cdot R_{sppm}^{2}(up, down) + \noniB{2}_1 \approx & 0.17\\
    Pay(\nonreporterB)=\noniA{2}_2\cdot R_{sppm}^{2}(up, down) + \noniB{2}_2 \approx & 0.17
\end{align*}
$\noniset{1}$ reporters takes the similar procedure of computing linear transformation and SPPM mutual payment. We first compute the linear transformation that results in 
\begin{table}[ht]
\centering
\begin{tabular}{ll}
$\semiA_1\approx0.69$ & $\semiB_1\approx1.2\times10^{-4}$  \\
$\semiA_3\approx0.68$ & $\semiB_3\approx-2.2\times10^{-5}$
\end{tabular}
\end{table}
And the mutual payment between two agents are
\begin{equation*}
R_{sppm}^{1}([-100.-50), [0,50)) \approx -0.18. 
\end{equation*}
Therefore, we can calculate the payment of $\noniset{1}$ reporters:
\begin{align*}
     Pay(\semireporterA)=\semiA_1\cdot R_{sppm}^{1}([-100.-50), [0,50)) + \semiB_1 \approx & -0.12\\
    Pay(\semireporterB)=\semiA_3\cdot R_{sppm}^{1}([-100.-50), [0,50)) + \semiB_3 \approx & -0.11
\end{align*}

\subsection{Mutual-Information-Based Mechanism}
Let $\PSI{0}, \PSI{1}, \PSI{2}$ be the log scoring rule, and let the coefficients be
\begin{align*}
    \alpha^l_i(|\nonirepset{0}|,|\nonirepset{1}|,|\nonirepset{2}|)=\frac{1}{|\nonirepset{0}|+|\nonirepset{1}|+|\nonirepset{2}|}
\end{align*}
for $0\le l,i\le 2$. For each agent, the mechanism will u.a.r select three other agents (i.e. peers) in $\nonirepset{0},\nonirepset{1},\nonirepset{2}$ respectively in order to calculate the payment. We assume the peers for each agent are in Table~\ref{tab:mipeer}.
\begin{table}[ht]
\centering
\begin{tabular}{|l|l|l|l|}
\hline
Reporter & $\nonirepset{0}$ peer & $\nonirepset{1}$ peer & $\nonirepset{2}$ peer \\ \hline
$\expreporterA$ & $\expreporterB$ & $\semireporterB$ & $\nonreporterB$ \\ \hline
$\expreporterB$  & $\expreporterA$ & $\semireporterA$ & $\nonreporterA$ \\ \hline
$\semireporterA$ & $\expreporterB$ & $\semireporterB$ & $\nonreporterB$ \\ \hline
$\semireporterB$ & $\expreporterA$ & $\semireporterA$ & $\nonreporterA$ \\ \hline
$\nonreporterA$  & $\expreporterB$ & $\semireporterB$ & $\nonreporterB$ \\ \hline
$\nonreporterB$  & $\expreporterA$ & $\semireporterA$ & $\nonreporterA$ \\ \hline
\end{tabular}
\caption{Peers of six reporters. \label{tab:mipeer}}
\end{table}

Note that with Assumption~\ref{asm:ci}, and when we use log scoring rule for mutual-information-based mechanism, $R_{l_1 l_2}(s^{l_1},s^{l_2})$ will become similar with the payment function of ECGM (Equation~\ref{eq:ecg}).

For agent $\expreporterA$, the payment is calculated as
{\normalsize
\begin{align*}
    Pay(\expreporterA)=&\frac{\log\left(\Pr[\expSig=70\mid\expSig=20]\right)-\log\left(\Pr[\expSig=70]\right)}{|\nonirepset{0}|+|\nonirepset{1}|+|\nonirepset{2}|} \\
    +&\frac{\log\left(\Pr[\noniSig{1}=[0,50)\mid\expSig=20]\right)-\log\left(\Pr[\noniSig{1}=[0,50)]\right)}{|\nonirepset{0}|+|\nonirepset{1}|+|\nonirepset{2}|} \\
    +&\frac{\log\left(\Pr[\noniSig{2}=down\mid\expSig=20]\right)-\log\left(\Pr[\noniSig{2}=down]\right)}{|\nonirepset{0}|+|\nonirepset{1}|+|\nonirepset{2}|} \\
    =&\frac{1}{6}\log\left(\sum_{y\in\{0,1\}}\frac{\Pr[Y=y\mid \expSig=70]\Pr[Y=y\mid\expSig=20]}{\Pr[Y=y]}\right) \\
    +&\frac{1}{6}\log\left(\sum_{y\in\{0,1\}}\frac{\Pr[Y=y\mid \expSig=70]\Pr[Y=y\mid\noniSig{1}=[0,50)]}{\Pr[Y=y]}\right) \\
    +&\frac{1}{6}\log\left(\sum_{y\in\{0,1\}}\frac{\Pr[Y=y\mid \expSig=70]\Pr[Y=y\mid\noniSig{2}=down]}{\Pr[Y=y]}\right) \\
    =&\frac{1}{6}\log\left(\frac{\exppost^0_A\exppost^0_B}{P}+\frac{(1-\exppost^0_A)(1-\exppost^0_B)}{1-P}\right) \\
    +&\frac{1}{6}\log\left(\frac{\exppost^0_A\nonipost^{1}_B}{P}+\frac{(1-\exppost^0_A)(1-\nonipost^{1}_B)}{1-P}\right) \\
    +&\frac{1}{6}\log\left(\frac{\exppost^0_A\nonipost^{2}_B}{P}+\frac{(1-\exppost^0_A)(1-\nonipost^{2}_B)}{1-P}\right) \\
    \approx& 0.017
\end{align*}
}
Similarly, 
\begin{align*}
    Pay(\expreporterB)=&\frac{1}{6}\log\left(\frac{\exppost^0_B\exppost^0_A}{P}+\frac{(1-\exppost^0_B)(1-\exppost^0_A)}{1-P}\right) \\
    +&\frac{1}{6}\log\left(\frac{\exppost^0_B\nonipost^{1}_A}{P}+\frac{(1-\exppost^0_B)(1-\nonipost^{1}_A)}{1-P}\right) \\
    +&\frac{1}{6}\log\left(\frac{\exppost^0_B\nonipost^{2}_B}{P}+\frac{(1-\exppost^0_B)(1-\nonipost^{2}_A)}{1-P}\right) \\
    \approx& 0.003
\end{align*}
For $\semireporterA$ or $\semireporterB$, the payment can be derived by replacing $\exppost^0_A$ by $\nonipost^{1}_A$ in $Pay(\expreporterA)$ or $\exppost^0_B$ by $\nonipost^{1}_B$ in $Pay(\expreporterB)$ respectively. The payments are
$$Pay(\semireporterA)\approx -0.070$$
$$Pay(\semireporterB)\approx 0.005$$
For $\nonreporterA$ or $\nonreporterB$, the payment can be derived by replacing $\exppost^0_A$ by $\nonipost^{2}_A$ in $Pay(\expreporterA)$ or $\exppost^0_B$ by $\nonipost^{2}_B$ in $Pay(\expreporterB)$ respectively. The payments are
$$Pay(\nonreporterA)\approx 0.033$$
$$Pay(\nonreporterB)\approx -0.028$$

\section{Proof of Main Theorem of Mutual-Information-Based Mechanism}
\label{proofmi}
\textbf{Theorem \ref{thm:mi}} \emph{Given Assumption~\ref{asm:ci} and \ref{asm:ip}. 
Assume $|\nonirepset{l}|\ge3$ for all $0\le l\le n$.
If for $0\le l_1<l_2\le n$ and $0\le i\le n$, the following conditions holds for $t_0\ge2,\cdots,t_n\ge 2$:
$$ \alpha^{l_1}_i(t_0,\cdots,t_n) = \alpha^{l_2}_i(t_0,\cdots,t_{l_1}-1, \cdots,t_{l_2}+1,\cdots,t_n)$$
then MIBM is \newtruthful.}

\begin{proof}
For convenience, when we say $\exprep_t$ or $\expsig_t$ in $\pdfs[\cdot]$, the corresponding random variable is $\expSig_t$. Similarly, for $l\in[\noninum]$, when we say $\nonirep{l}_t$ or $\nonisig{l}_t$ in $\Pr[\cdot]$, the corresponding random variable is $\noniSig{l}_t$. First, we will prove the \newtruthfulness{} for expert, then we will prove the group $l$ non-expert cases. \\

\noindent\textbf{Expert Part} For expert $r$, we pick an expert $j^0\in\nonirepset{0}\setminus\{r\}$ u.a.r., and pick agents $j^i\in\nonirepset{i},1\le i\le\noninum$ u.a.r. Given that all other agents report truthfully, which means $\nonirep{i}_{j^i}=\nonisig{i}_{j^i}$ for $0\le i\le\noninum$, and conditioning on expert $r$ receives continuous signal $\expsig_r$, if she chooses to report $\exprep_r\in\expset$, the expected payment will be
{\normalsize
\begin{align*}
    &\alpha^{0}_0\left(|\nonirepset{0}|,\cdots,|\nonirepset{\noninum}|\right)\int_{\expset}\pdfs[\exprep_{j^0}\mid\expsig_r]\left(\PSI{0}\left(\exprep_{j^0},\SSIdist{0}(\exprep_r)\right)-\PSI{0}\left(\exprep_{j^0},\SSIdist{0}\right)\right){\rm d}\exprep_{j^0} \\
    +&\sum_{i=1}^{\noninum}\alpha^{0}_i\left(|\nonirepset{0}|,\cdots,|\nonirepset{\noninum}|\right)\sum_{\nonirep{i}_{j^i}\in\noniset{i}}\Pr[\nonirep{i}_{j^i}\mid\expsig_r] \left(\PSI{i}\left(\nonirep{i}_{j^i},\SSIdist{0}(\exprep_r)\right)-\PSI{i}\left(\nonirep{i}_{j^i},\SSIdist{i}\right)\right)
\end{align*}
}
Since all the proper scoring rules are strictly proper and we assume the informative prior (Assumption~\ref{asm:ip}), so the expected payment is uniquely maximized when $\exprep_r=\expsig_r$. Now we need to prove that if she chooses to report in different signal spaces, she won't get higher expected payment, and she will get the maximized expected payment if and only if the posterior signal distributions are the same as her private signal's ones.

Conditioning on expert $r$ gets her signal $\expsig_r$ and others report truthfully, she can calculate the expected payment when she reports truthfully,
{\normalsize
\begin{align*}
    &\alpha^{0}_0\left(|\nonirepset{0}|,\cdots,|\nonirepset{\noninum}|\right) \\
    \times&\underset{\text{EE}_{exp}}{\underbrace{\int_{\expset}\pdfs[\exprep_{j^0}\mid\expsig_r]\left(\PSI{0}\left(\exprep_{j^0},\SSIdist{0}(\expsig_r)\right)-\PSI{0}\left(\exprep_{j^0},\SSIdist{0}\right)\right){\rm d}\exprep_{j^0}}} \\
    +&\sum_{i=1}^{\noninum}\alpha^{0}_i\left(|\nonirepset{0}|,\cdots,|\nonirepset{\noninum}|\right) \\
    \times&\underset{\text{EN}^i_{exp}}{\underbrace{\sum_{\nonirep{i}_{j^i}\in\noniset{i}}\Pr[\nonirep{i}_{j^i}\mid\expsig_r] \left(\PSI{i}\left(\nonirep{i}_{j^i},\SSIdist{0}(\expsig_r)\right)-\PSI{i}\left(\nonirep{i}_{j^i},\SSIdist{i}\right)\right)}}
\end{align*}
}

If she chooses to report answer $\noniele{h}_t\in\noniset{h}$, $h\in[n]$, note that the mechanism will regard $\noniele{h}_t$ as a realization of random variable $\noniSig{h}$. When she moves from $\exprepset$ to $\nonirepset{h}$, the sizes of the two sets will change, and the expected payment will be
{\normalsize
\begin{align*}
    &\alpha^{h}_0\left(|\nonirepset{0}|-1,\cdots,|\nonirepset{h}|+1,\cdots|\nonirepset{\noninum}|\right) \\
    \times&\underset{\text{N}^{h}\text{E}_{exp}}{\underbrace{\int_{\expset}\pdfs[\exprep_{j^0}\mid\expsig_r]\left(\PSI{0}\left(\exprep_{j^0},\SSIdist{h}(\noniele{h}_t)\right)-\PSI{0}\left(\exprep_{j^0},\SSIdist{0}\right)\right){\rm d}\exprep_{j^0}}} \\
    +&\sum_{i=1}^{\noninum}\alpha^{h}_i\left(|\nonirepset{0}|-1,\cdots,|\nonirepset{h}|+1,\cdots|\nonirepset{\noninum}|\right) \\
    \times&\underset{\text{N}^{h}\text{N}^i_{exp}}{\underbrace{\sum_{\nonirep{i}_{j^i}\in\noniset{i}}\Pr[\nonirep{i}_{j^i}\mid\expsig_r] \left(\PSI{i}\left(\nonirep{i}_{j^i},\SSIdist{h}(\noniele{h}_t)\right)-\PSI{i}\left(\nonirep{i}_{j^i},\SSIdist{i}\right)\right)}}
\end{align*}
}
By the definition of strictly proper scoring rule, we know that, for $h,i\in[n]$
$$\text{N}^{h}\text{E}_{exp}\le \text{EE}_{exp}$$
$$\text{N}^{h}\text{N}^i_{exp}\le \text{EN}^i_{exp}$$
Since the coefficients satisfy
\begin{align*}
    \alpha^{0}_i\left(|\nonirepset{0}|,\cdots,|\nonirepset{\noninum}|\right) &= \alpha^{h}_i\left(|\nonirepset{0}|-1,\cdots,|\nonirepset{h}|+1,\cdots|\nonirepset{\noninum}|\right)
\end{align*}
for $0\le i\le\noninum$, so
{\normalsize
\begin{align*}
    &\alpha^{0}_0\left(|\nonirepset{0}|,\cdots,|\nonirepset{\noninum}|\right)\text{EE}_{exp} \\
    +&\sum_{i=1}^{\noninum}\alpha^{0}_i\left(|\nonirepset{0}|,\cdots,|\nonirepset{\noninum}|\right)\text{EN}^i_{exp} \\
    \ge&\alpha^{h}_0\left(|\nonirepset{0}|-1,\cdots,|\nonirepset{h}|+1,\cdots|\nonirepset{\noninum}|\right)\text{N}^{h}\text{E}_{exp} \\
    +&\sum_{i=1}^{\noninum}\alpha^{h}_i\left(|\nonirepset{0}|-1,\cdots,|\nonirepset{h}|+1,\cdots|\nonirepset{\noninum}|\right)\text{N}^{h}\text{N}^i_{exp}
\end{align*}
}
which means expert $r$ won't get higher expected payment if she chooses to report answer in $\noniset{h}$. And the equality holds if and only if $\SSIdist{0}(\expsig_r)=\SSIdist{h}(\noniele{h}_t)$.

We've finished the proof of the expert part. The proof of the group $l$ non-expert part is similar to the proof of the expert part by using the same idea and inequality. \\

\noindent\textbf{Group $l$ Non-expert Part} For group $l$ non-expert $r$, we pick an group $l$ non-expert $j^l\in\nonirepset{l}\setminus\{r\}$ u.a.r., and pick agents $j^i\in\nonirepset{i},0\le i\neq l\le\noninum$ u.a.r. Given that all other agents report truthfully, which means $\nonirep{i}_{j^i}=\nonisig{i}_{j^i}$ for $0\le i\le\noninum$, and conditioning on group $l$ non-expert $r$ receives signal $\nonisig{l}_r$, if she chooses to report $\nonirep{l}_r\in\noniset{l}$, the expected payment will be
{\normalsize
\begin{align*}
    &\alpha^{l}_0\left(|\nonirepset{0}|,\cdots,|\nonirepset{\noninum}|\right)\int_{\expset}\pdfs[\exprep_{j^0}\mid\nonisig{l}_r]\left(\PSI{0}\left(\exprep_{j^0},\SSIdist{l}(\nonirep{l}_r)\right)-\PSI{0}\left(\exprep_{j^0},\SSIdist{0}\right)\right){\rm d}\exprep_{j^0} \\
    +&\sum_{i=1}^{\noninum}\alpha^{l}_i\left(|\nonirepset{0}|,\cdots,|\nonirepset{\noninum}|\right)\sum_{\nonirep{i}_{j^i}\in\noniset{i}}\Pr[\nonirep{i}_{j^i}\mid\nonisig{l}_r] \left(\PSI{i}\left(\nonirep{i}_{j^i},\SSIdist{l}(\nonirep{l}_r)\right)-\PSI{i}\left(\nonirep{i}_{j^i},\SSIdist{i}\right)\right)
\end{align*}
}
Since all the proper scoring rules are strictly proper and we assume the informative prior (Assumption~\ref{asm:ip}), so the expected payment is uniquely maximized when $\nonirep{l}_r=\nonisig{l}_r$, which means she can get the highest expected payment when she reports truthfully, and the expected payment is
{\normalsize
\begin{align*}
    &\alpha^{l}_0\left(|\nonirepset{0}|,\cdots,|\nonirepset{\noninum}|\right)\underset{\text{N}^l\text{E}_{non-l}}{\underbrace{\int_{\expset}\pdfs[\exprep_{j^0}\mid\nonisig{l}_r]\left(\PSI{0}\left(\exprep_{j^0},\SSIdist{l}(\nonisig{l}_r)\right)-\PSI{0}\left(\exprep_{j^0},\SSIdist{0}\right)\right){\rm d}\exprep_{j^0}}} \\
    +&\sum_{i=1}^{\noninum}\alpha^{l}_i\left(|\nonirepset{0}|,\cdots,|\nonirepset{\noninum}|\right)\underset{\text{N}^l\text{N}^i_{non-l}}{\underbrace{\sum_{\nonirep{i}_{j^i}\in\noniset{i}}\Pr[\nonirep{i}_{j^i}\mid\nonisig{l}_r] \left(\PSI{i}\left(\nonirep{i}_{j^i},\SSIdist{l}(\nonisig{l}_r)\right)-\PSI{i}\left(\nonirep{i}_{j^i},\SSIdist{i}\right)\right)}}
\end{align*}
}

If she chooses to report answer $\expele_t\in\expset$, note that the mechanism will regard $\expele_t$ as a realization of random variable $\expSig$. When she moves from $\nonirepset{l}$ to $\exprepset$, the expected payment will be
{\normalsize
\begin{align*}
    &\alpha^{0}_0\left(|\nonirepset{0}|+1,\cdots,|\nonirepset{l}|-1,\cdots|\nonirepset{\noninum}|\right) \\
    \times&\underset{\text{EE}_{non-l}}{\underbrace{\int_{\expset}\pdfs[\exprep_{j^0}\mid\nonisig{l}_r]\left(\PSI{0}\left(\exprep_{j^0},\SSIdist{0}(\expele_t)\right)-\PSI{0}\left(\exprep_{j^0},\SSIdist{0}\right)\right){\rm d}\exprep_{j^0}}} \\
    +&\sum_{i=1}^{\noninum}\alpha^{0}_i\left(|\nonirepset{0}|+1,\cdots,|\nonirepset{l}|-1,\cdots|\nonirepset{\noninum}|\right) \\
    \times&\underset{\text{EN}^i_{non-l}}{\underbrace{\sum_{\nonirep{i}_{j^i}\in\noniset{i}}\Pr[\nonirep{i}_{j^i}\mid\nonisig{l}_r] \left(\PSI{i}\left(\nonirep{i}_{j^i},\SSIdist{0}(\expele_t)\right)-\PSI{i}\left(\nonirep{i}_{j^i},\SSIdist{i}\right)\right)}}
\end{align*}
}
By the definition of strictly proper scoring rule, we know that, for $i\in[n]$
$$\text{EE}_{non-l}\le \text{N}^l\text{E}_{non-l}$$
$$\text{EN}^i_{non-l}\le {N^l}{N^i}_{non-l}$$
Since the coefficients satisfy
\begin{align*}
    \alpha^{l}_i\left(|\nonirepset{0}|,\cdots,|\nonirepset{\noninum}|\right) &= \alpha^{0}_i\left(|\nonirepset{0}|+1,\cdots,|\nonirepset{l}|-1,\cdots|\nonirepset{\noninum}|\right)
\end{align*}
for $0\le i\le\noninum$, so
{\normalsize
\begin{align*}
    &\alpha^{l}_0\left(|\nonirepset{0}|,\cdots,|\nonirepset{\noninum}|\right)\text{N}^l\text{E}_{non-l} \\
    +&\sum_{i=1}^{\noninum}\alpha^{l}_i\left(|\nonirepset{0}|,\cdots,|\nonirepset{\noninum}|\right)\text{N}^l\text{N}^i_{non-l} \\
    \ge&\alpha^{0}_0\left(|\nonirepset{0}|+1,\cdots,|\nonirepset{l}|-1,\cdots|\nonirepset{\noninum}|\right)\text{EE}_{non-l} \\
    +&\sum_{i=1}^{\noninum}\alpha^{0}_i\left(|\nonirepset{0}|+1,\cdots,|\nonirepset{l}|-1,\cdots|\nonirepset{\noninum}|\right)\text{EN}^i_{non-l}
\end{align*}
}
which means group $l$ non-expert $r$ won't get higher payment expected payment if she chooses to report answer in $\expset$. And the equality holds if and only if $\SSIdist{l}(\nonisig{l}_r)=\SSIdist{0}(\expele_t)$.

If she chooses to report answer $\noniele{h}_t\in\noniset{h}$, note that the mechanism will regard $\noniele{h}_t$ as a realization of random variable $\noniSig{h}$. Suppose $h<l$ here , when she moves from $\nonirepset{l}$ to $\nonirepset{h}$, the expected payment will be
{\normalsize
\begin{align*}
    &\alpha^{h}_0\left(|\nonirepset{0}|,\cdots,|\nonirepset{h}|+1,\cdots,|\nonirepset{l}|-1,\cdots|\nonirepset{\noninum}|\right)\\
    \times&\underset{\text{N}^h\text{E}_{non-l}}{\underbrace{\int_{\expset}\pdfs[\exprep_{j^0}\mid\nonisig{l}_r]\left(\PSI{0}\left(\exprep_{j^0},\SSIdist{h}(\noniele{h}_t)\right)-\PSI{0}\left(\exprep_{j^0},\SSIdist{0}\right)\right){\rm d}\exprep_{j^0}}} \\
    +&\sum_{i=1}^{\noninum}\alpha^{h}_i\left(|\nonirepset{0}|,\cdots,|\nonirepset{h}|+1,\cdots,|\nonirepset{l}|-1,\cdots|\nonirepset{\noninum}|\right)\\
    \times&\underset{\text{N}^h\text{N}^i_{non-l}}{\underbrace{\sum_{\nonirep{i}_{j^i}\in\noniset{i}}\Pr[\nonirep{i}_{j^i}\mid\nonisig{l}_r] \left(\PSI{i}\left(\nonirep{i}_{j^i},\SSIdist{h}(\noniele{h}_t)\right)-\PSI{i}\left(\nonirep{i}_{j^i},\SSIdist{i}\right)\right)}}
\end{align*}
}
By the definition of strictly proper scoring rule, we know that, for $h,i\in[n]$
$$\text{N}^h\text{E}_{non-l}\le \text{N}^l\text{E}_{non-l}$$
$$\text{N}^h\text{N}^i_{non-l}\le {N^l}{N^i}_{non-l}$$
Since the coefficients satisfy
\begin{align*}
    \alpha^{l}_i\left(|\nonirepset{0}|,\cdots,|\nonirepset{\noninum}|\right) &= \alpha^{h}_i\left(|\nonirepset{0}|,\cdots,|\nonirepset{h}|+1,\cdots,|\nonirepset{l}|-1,\cdots|\nonirepset{\noninum}|\right)
\end{align*}
for $0\le i\le\noninum$ and $1\le h\neq l\le\noninum$, so
{\normalsize
\begin{align*}
    &\alpha^{l}_0\left(|\nonirepset{0}|,\cdots,|\nonirepset{\noninum}|\right)\text{N}^l\text{E}_{non-l} \\
    +&\sum_{i=1}^{\noninum}\alpha^{l}_i\left(|\nonirepset{0}|,\cdots,|\nonirepset{\noninum}|\right)\text{N}^l\text{N}^i_{non-l} \\
    \ge&\alpha^{h}_0\left(|\nonirepset{0}|,\cdots,|\nonirepset{h}|+1,\cdots,|\nonirepset{l}|-1,\cdots|\nonirepset{\noninum}|\right)\text{N}^h\text{E}_{non-l} \\
    +&\sum_{i=1}^{\noninum}\alpha^{h}_i\left(|\nonirepset{0}|,\cdots,|\nonirepset{h}|+1,\cdots,|\nonirepset{l}|-1,\cdots|\nonirepset{\noninum}|\right)\text{N}^h\text{N}^i_{non-l}
\end{align*}
}
which means group $l$ non-expert $r$ won't get higher payment expected payment if she chooses to report answer in $\noniset{h}$. And the equality holds if and only if $\SSIdist{l}(\nonisig{l}_r)=\SSIdist{h}(\noniele{h}_t)$. For $h>l$ cases, just substitute all $\alpha^{h}_i\left(|\nonirepset{0}|,\cdots,|\nonirepset{h}|+1,\cdots,|\nonirepset{l}|-1,\cdots|\nonirepset{\noninum}|\right)$ with $\alpha^{h}_i\left(|\nonirepset{0}|,\cdots,|\nonirepset{l}|-1,\cdots,|\nonirepset{h}|+1,\cdots|\nonirepset{\noninum}|\right)$ then the proof remains the same.
\end{proof}

\section{Exclusion of Two Mechanisms}
Now we give a special case of setting to show that Composite Elicitation Mechanism and MI-based Mechanisms are exclusive, i.e. none of them is a subclass of the other. 

We consider the following setting:
$P=0.5$. Expert's signal domain $\expset=[0,1]$. $\nonidist{0}_0\sim N(0,1)$ restricted on $\expset$, and $\nonidist{0}_1\sim N(1,1)$ restricted on $\expset$. 

Group 1 non-expert's signal domain $\noniset{1}=\{\noniele{1}_1,\noniele{1}_2,\noniele{1}_3\}$. When $Y=0$, let the probability distribution be $\nonidist{1}_0=(0.6,0.2,0.2)$; and when $Y=1$, let the probability distribution be $\nonidist{1}_1=(0.2,0.5,0.3)$. 

Group 2 non-expert's signal domain $\noniset{2}=\{\noniele{2}_1,\noniele{2}_2\}$. When $Y=0$, let the probability distribution be $\nonidist{2}_0=(0.5,0.5)$; and when $Y=1$, let the probability distribution be $\nonidist{2}_1=(0.5,0.5)$. 

Now we consider the agents in the mechanisms. We list three pairs of agents of each signal space. Their reports and posterior distributions are in the Table~\ref{tab:exclusion}. 
\begin{table}[ht]
\centering
\begin{tabular}{|l|l|l|l|}
\hline
Agents           & Report domain & Report       & Posterior distribution on ground truth \\ \hline
$\expreporterA$  & $\expset$     & $0.5$        & $\exppost_A^0= 0.5$                      \\ \hline
$\expreporterB$  & $\expset$     & $0.5$        & $\exppost_B^0= 0.5$                \\ \hline
$\semireporterA$ & $\noniset{1}$    & $\noniele{1}_1$ & $\nonipost^{1}_A=0.25$                     \\ \hline
$\semireporterB$ & $\noniset{1}$    & $\noniele{1}_3$ & $\nonipost^{1}_B=0.6$                      \\ \hline
$\nonreporterA$  & $\noniset{2}$     & $\noniele{2}_1$  & $\nonipost^{2}_A=0.5$                       \\ \hline
$\nonreporterB$  & $\noniset{2}$     & $\noniele{2}_2$  & $\nonipost^{2}_B=0.5$                       \\ \hline
\end{tabular}
\caption{Reporters for the exclusion case \label{tab:exclusion}}
\end{table}
Now we show the exclusion of the two mechanism by showing that in this special case, the payment of two $\noniset{1}$ reporters are different in CEM, but same in MIBM. 
\subsection{Composite Elicitation Mechanism Payment}
We first compute the linear transformation coefficients:
\begin{table}[ht]
\centering
\begin{tabular}{ll}
$\semiA_1=2.36$ & $\semiB_1=0.25$  \\
$\semiA_3=3.15$ & $\semiB_3=-0.10$
\end{tabular}
\end{table}
Then we compute the SPP mutual payment: $R_{sppm}(\noniele{1}_1,\noniele{1}_3)=-0.10$. Therefore, the CEM payment for two group 1 non-experts are
\begin{align*}
    Pay_{ce}(\semireporterA)=\semiA_1\cdot R_{sppm}(\noniele{1}_1,\noniele{1}_3) + \semiB_1 \approx & 0.014\\
    Pay_{ce}(\semireporterB)=\semiA_3\cdot R_{sppm}(\noniele{1}_1,\noniele{1}_3) + \semiB_3 \approx & 0.42
\end{align*}
We can see that payment are different for two agents in CEM. 

\subsection{Mutual-Information-Based Mechanism Payment}
We just calculate the payments for $\semireporterA$ and $\semireporterB$ here. Since the reports for all experts (and all group 2 non-experts) are the same, it doesn't matter which peer the mechanism chooses. The coefficients are constants when $|\nonirepset{0}|,|\nonirepset{1}|,|\nonirepset{2}|$ are known. Since CEM uses log scoring rule, here we use log scoring rule as well.
\begin{align*}
    &Pay(\semireporterA)=Pay(\semireporterB) \\
    =&\alpha^1_{0}\left(2,2,2\right)\cdot0+\alpha^1_{1}\left(2,2,2\right)\log\left(\frac{0.25\cdot0.6}{0.5}+\frac{(1-0.25)\cdot(1-0.6)}{0.5}\right)+\alpha^1_{2}\left(2,2,2\right)\cdot0 \\
    \approx&-0.105\cdot\alpha^1_{1}\left(2,2,2\right)
\end{align*}
The MIBM payments for two group 1 non-experts are the same.

From the calculation above, we show that two $\noniset{1}$ reporter are paid differently in CEM, but same in MIBM. Therefore, the two mechanisms are exclusive to each other.

\section{Experiments}
In this section, we present results of experiments on synthetic data and on MTurk to verify the statistical efficiency of information elicitation. 
\label{sec:exp}

\subsection{Computer Configuration}
\label{appsec:confOfExp}
The computer configuration is as follows:
\begin{enumerate}
    \item Laptop: Lenovo Legion Y7000P
    \item Processor: Intel® Core™ i7-9750H @2.60GHz
    \item Graphics: NVIDIA® GeForce RTX™ 2060
    \item Memory: DDR4-2666 8GBx2
    \item Operating System: Windows 10 Home
\end{enumerate}

\subsection{Numerical Experiment Main Results}
\label{subsec:result}
We ran the numerical experiments in three different settings and due to the space constraint, we only present one of them here, and the other two can be found in Appendix~\ref{appsec:moreResult}. 

\noindent{\bf The High Variance Setting} is exactly the same as the one presented in Example~\ref{ex:sig}.

We use the {\em maximum a posteriori estimation (MAP)} to aggregate the reports from the agents to predict the ground truth. Five mechanisms are compared: CEM, MIBM, ECGM, the SPPM for group 1 non-experts (whose payment function is $R_{sppm}^{1}$) and the SPPM for group 2 non-experts (whose payment function is $R_{sppm}^{2}$). In the experiments we assume that all agents would like to be as truthful as possible, but are sometimes asked to report signals in a different report space. In the latter case we need to make some natural assumptions on their behavior described as follows.

\noindent $\bullet$ {\em CEM \& MIBM.} Since agents' reports are the same under CEM and MIBM, we merge the curves for CEM and MIBM.

\noindent $\bullet$ {\em ECGM Uniform Non-expert.} If the group $i$ non-expert receives signal $\noniele{1}_k$ or $\noniele{2}_k$, she will report a probability drawn from the uniform distribution over $\left(\frac{k-1}{m_i},\frac{k}{m_i}\right)$.

\noindent $\bullet$ {\em SPPM for Group 1 Non-expert.} Each expert converts her cardinal signal to the group-1-interval it falls into. Each group 2 non-expert  chooses a group-1-interval that is consistent with her cardinal signal u.a.r.

\noindent $\bullet$ {\em SPPM for Group 2 Non-expert.} Since the signal space of group 2 non-experts is the most coarse, each expert and group 1 non-expert simply converts her   signal to the group-1-interval it falls into.

\noindent\textbf{Results.} Figure~\ref{fig:variance} shows the results calculated over $10^4$ iterations with 95\% confidence interval.  The computer configuration is in Appendix~\ref{appsec:confOfExp}.   As can be seen from the figure, while the accuracy of MAP under the five mechanisms all converge to $1$, \textbf{CEM and MIBM have the highest accuracy among all mechanisms}. The same conclusion can be drawn in other settings we have tested, see Appendix~\ref{appsec:moreResult} for more details, in particular, ECGM works poorly in one of the settings, whose accuracy converges to $0$.

\begin{figure}[htp]
    \centering
    \includegraphics[trim=1cm 0.1cm 1.2cm 2.3cm, clip,width=\linewidth] {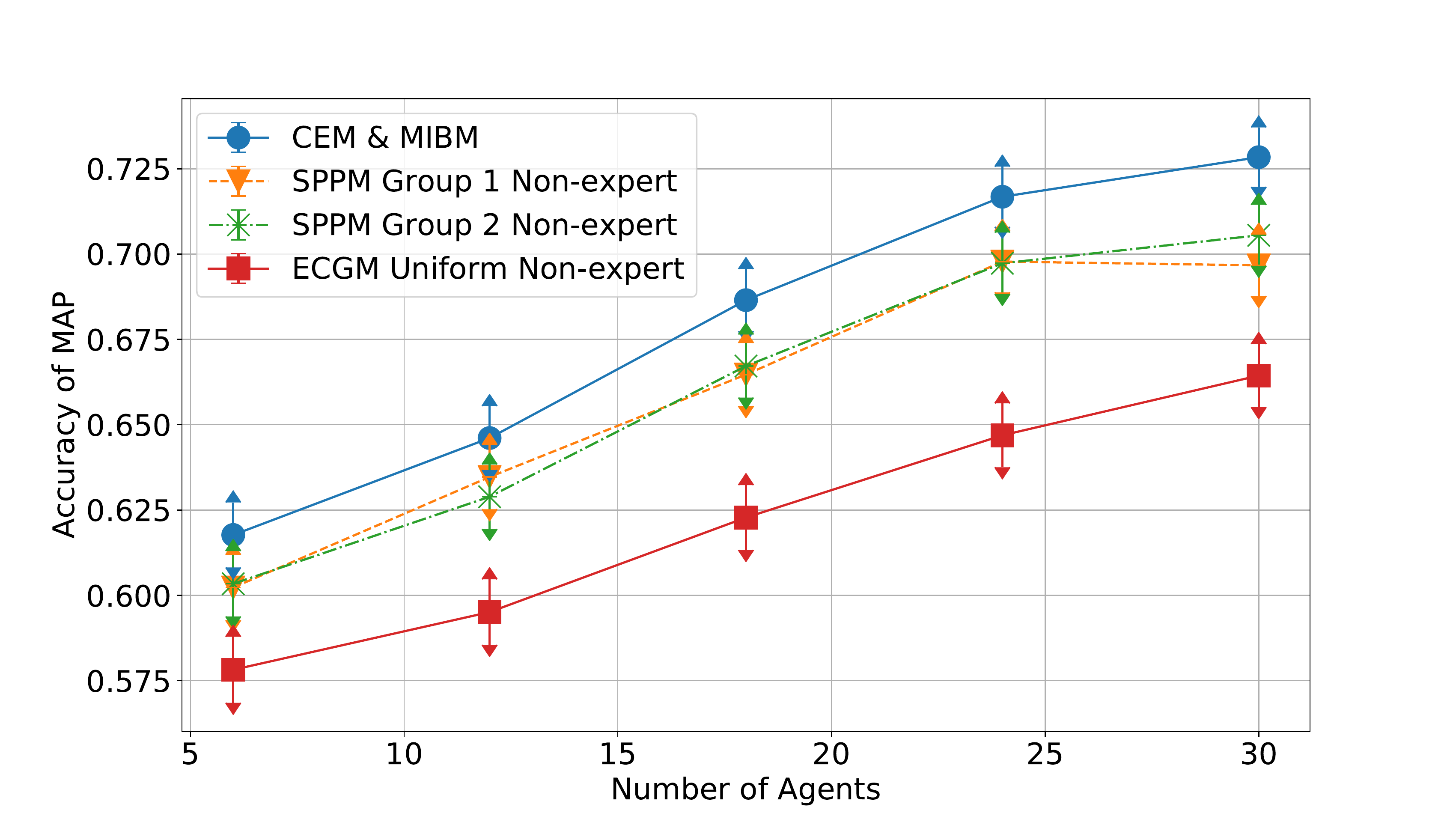}
    \caption{Accuracy of MAP   in the \textit{High Variance} setting.}
    \label{fig:variance}
\end{figure}

\subsection{More Numerical Experiment Results}
\label{appsec:moreResult}
In the \textit{standard} setting, we only change the variances of the signal distributions $f_0, f_1$ to $1$ from the \textit{High Variance} setting. In this scenario, we note from the results in Figure~\ref{fig:standard} \textbf{MAP accuracy for SPPM and ECGM converges to 1 slowly compared to CEM/MIBM}.

\begin{figure}[htp]
    \centering
    \includegraphics[width=0.7\linewidth] {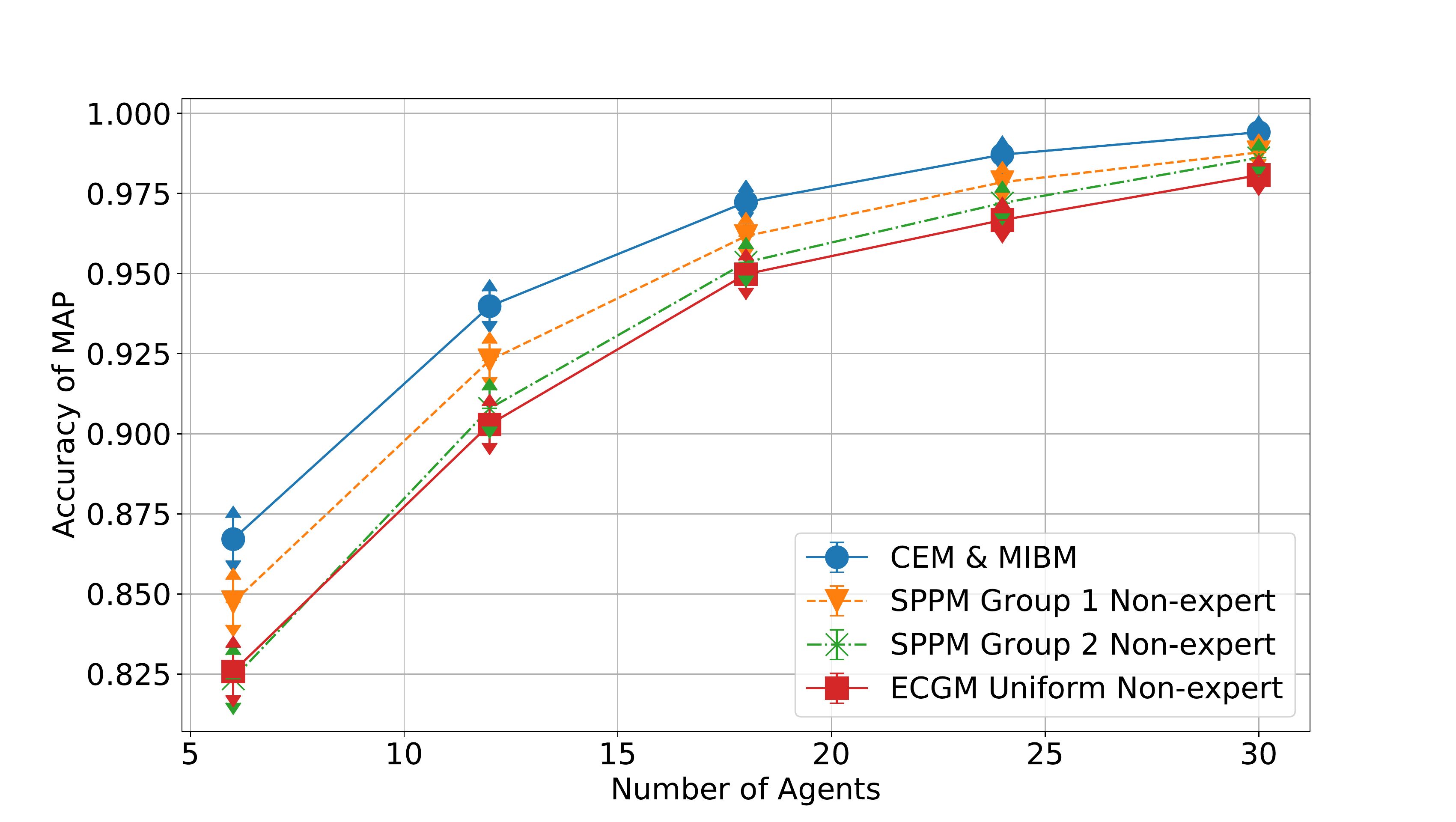}
    \caption{MAP aggregated prediction accuracy of the four mechanisms in the \textit{standard} setting.}
    \label{fig:standard}
\end{figure}

In the \textit{Prior 0.8} setting, we change the prior distribution to $\Pr[Y=1]=0.8$ from the \textit{standard} setting.
We see from the result in the Figure~\ref{fig:prior0.8} that under certain scenarios, it is possible for the ECGM mechanism to perform poorly in terms of aggregation accuracy where the accuracy never converges to 1.

\begin{figure}[htp]
    \centering
    \includegraphics[width=0.7\linewidth] {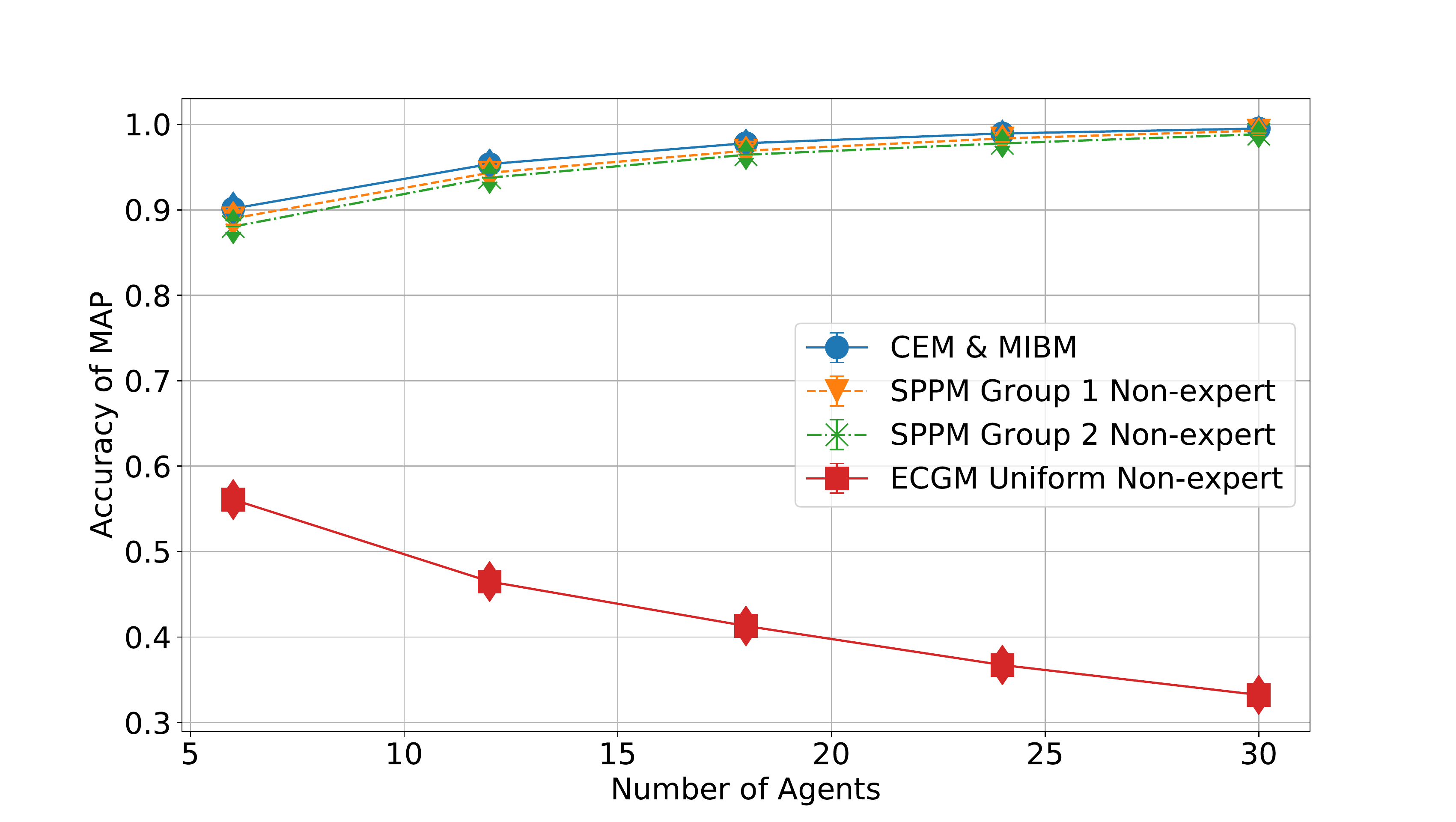}
    \caption{MAP aggregated prediction accuracy of the four mechanisms in the \textit{Prior 0.8} setting.}
    \label{fig:prior0.8}
\end{figure}

Our experiments indicate that CEM and MIBM performs consistently well with different types of hybrid data.

\subsection{MTurk Experiments}
\label{appsec:mturk}
We conducted a poll on  MTurk on Nov.~1 and Nov.~2, 2020, which was one day before the US 2020 Presidential election, about predictions for the election result. 
We asked every participant to answer {\em all} of the following three types of questions that roughly correspond to the report spaces of experts, group 1 non-experts, and group 2 non-experts, respectively: (1) Exact Answer - exactly predict how many electoral votes each candidate would get. (2) Multiple choice question - choose one of the four intervals the electoral college votes Trump would get: $[0,249), [250,270), [270,290], (290,538]$. (3) Binary choice - who will win? 

We also ask each Turker to identify her favorite answer type, i.e.~which of the three types of questions mentioned above best described her expertise. $707$ MTurkers participated in the poll, and the ratio of Turkers with different favorite answer types is presented in Table~\ref{tab:mturk_ratio}. 
\begin{table}[htp]
    \centering
    \begin{tabular}{c|c|c}
        Exact Answer  & Multiple Choice & Binary Choice  \\
         \hline
         25.2\% & 41.7\% &  33.1\%
    \end{tabular}
    \caption{\small Ratio of Turkers with different favorite answer types.}
    \label{tab:mturk_ratio}
\end{table}

Table~\ref{tab:mturk_ratio} indicates a hybrid crowd which motivates the study in this paper---notice that there is a large number of Turkers favoring each type of questions. 

For accuracy, we focus on the prediction of the exact electoral college votes. We first divide the participants into groups according to their most-preferred answer type,  and then compute the MSE of each group's predictions on the exact electoral college votes as a function of the ground truth (Figure~\ref{fig:mturk_president}), which evaluates how consistent the agents in the group are. Interestingly, for any ground truth electoral college that is $<310$, considering only the \textit{exact answerers} gives the lowest MSE (the ground truth turned out to be $232$). 
While this experiment does not consider agents' incentives and does not exactly correspond to the model studied in this paper, the message behind it aligns well with the main message of this paper---that  it can be beneficial to elicit different type of answers in a hybrid crowd.

\begin{figure}[htp]
    \centering
    \includegraphics[trim=1.8cm 1.4cm 3.5cm 2.25cm, clip,width=\linewidth]{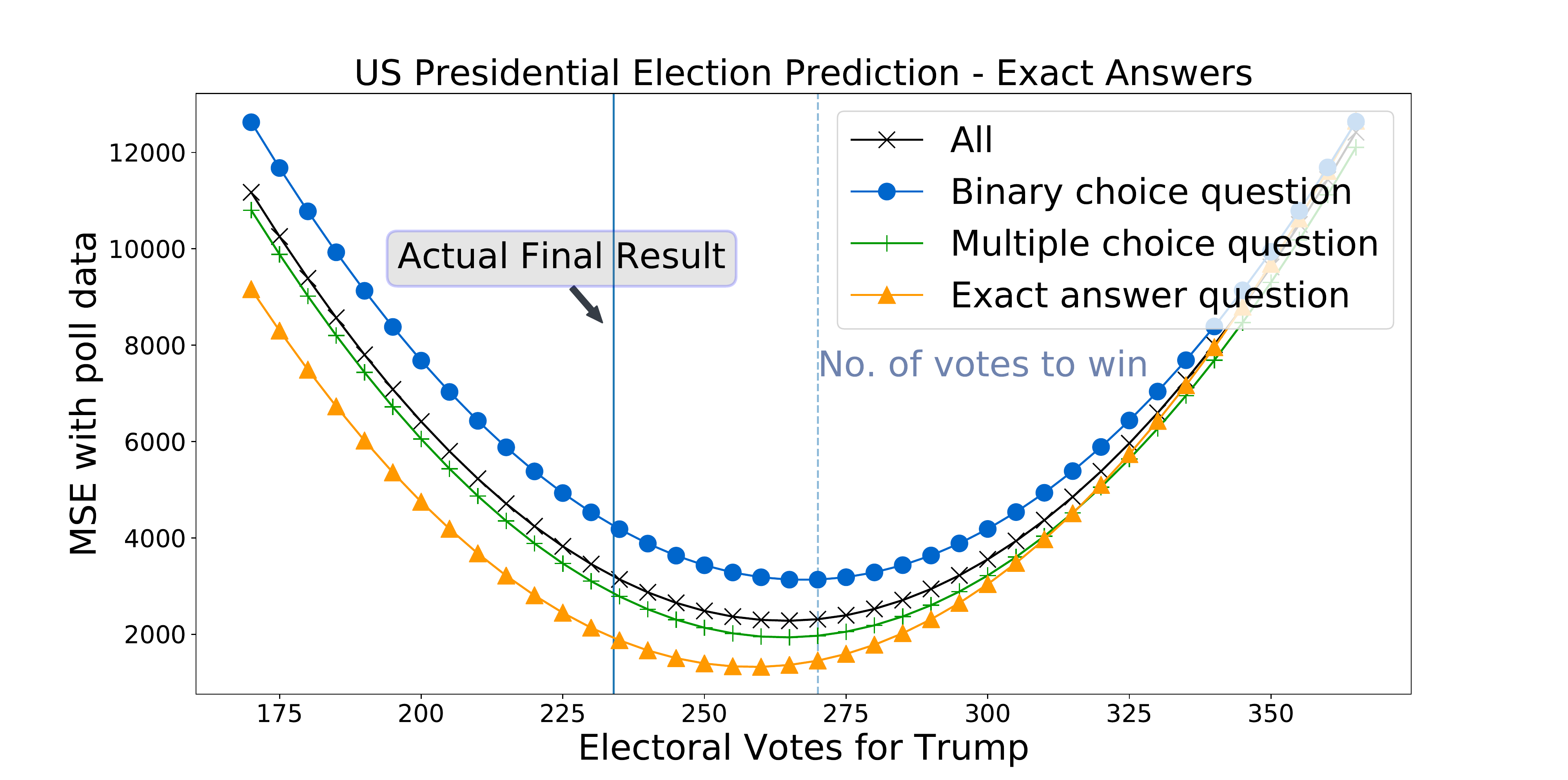}
    \caption{MSE for electoral college votes prediction.}
    \label{fig:mturk_president}
\end{figure}

\subsubsection{Setup}
We ran experiments on MTurk to see how real agents behaved when given multiple options (continuous or discrete) to express their opinion. For the initial experiment, we chose a topic that general people are expected to have varying level of knowledge about -- the 2020 US Presidential election. Since the question focus on the US, we recruited Turkers only from the US.

Question set:
\begin{itemize}
    \item Who do you think will win the 2020 US Presidential Election? (Binary choice question). 
    \\Options are Donald Trump and Joe Biden, the two leading candidates.
    \item Around how many electoral votes will Donald Trump get? (Multiple choice question)
    \\Options are $<250$ (large margin loss), $250-270$ (close loss), $270-290$ (close win), $>290$ (large margin win) 
    \item Exactly how many electoral votes will Donald Trump get? (Fully cardinal question)
    \\Expected answer is an integer $\in [0,538]$. 
\end{itemize}

\begin{figure}
    \centering
    \includegraphics[width=\textwidth]{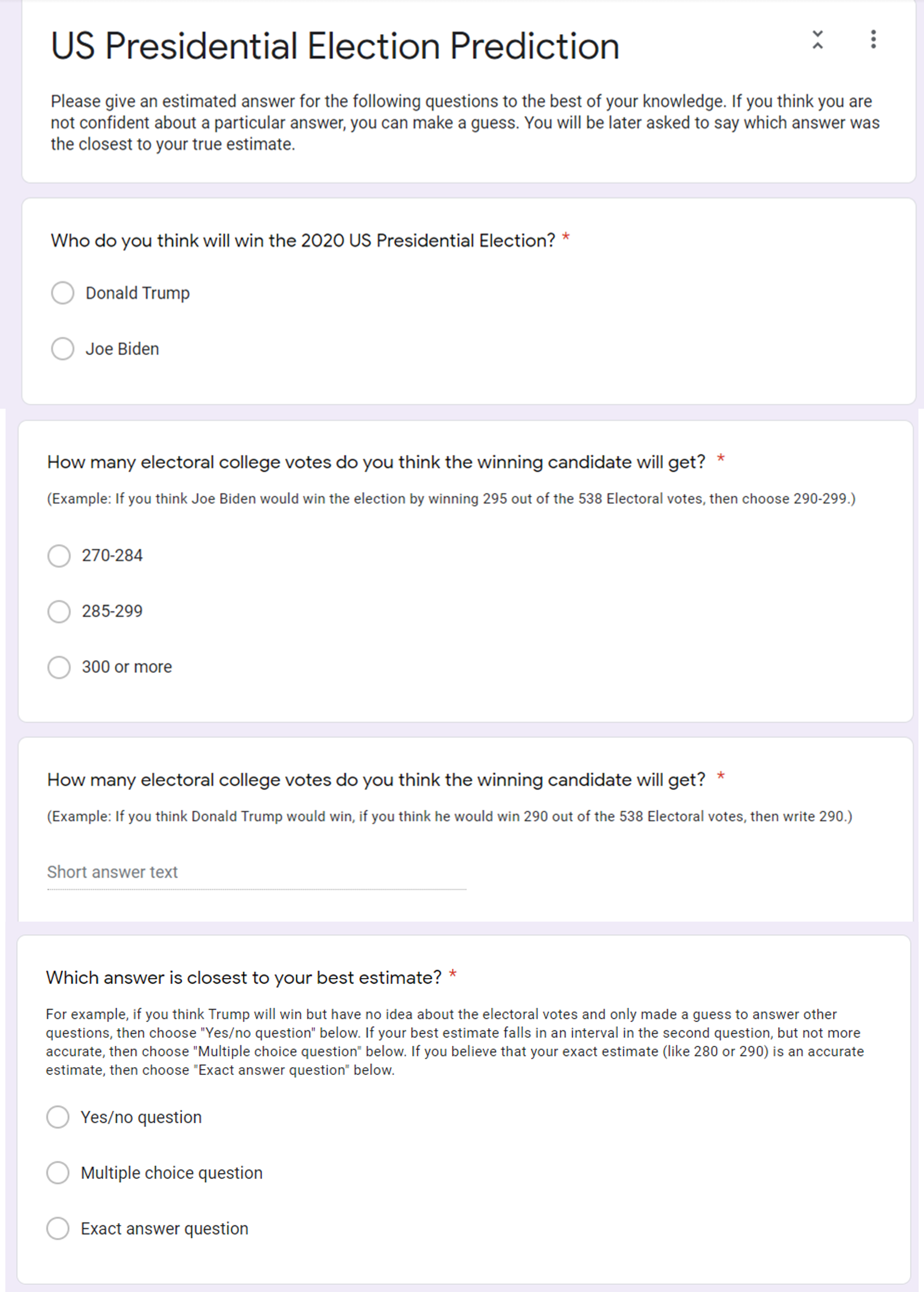}
    \caption{Exact Questions asked and accompanying instructions.}
    \label{fig:mturk_form}
\end{figure}

The actual environment that participants saw is shown in Figure~\ref{fig:mturk_form}. Every participant gives an answer to all three questions and then answer an additional question regarding which one is closest to their best estimate. For example, if they think Trump would win but have no idea about distribution of electoral votes and only made a guess to answer other questions, then they choose "Binary question". If their best estimate falls in an interval in the second question, but no more accurate, then they would choose "Multiple choice question". If they thought that their exact estimate is accurate, then they chose "Exact answer question".

Thus, this gives us an idea about how good a participant thinks their estimate is. And could allow us to estimate which question they would choose if they were to answer only one question, which is the premise of the theoretical work in this paper.

Each survey participant was given 0.3 USD for their participation. Our estimate was that for someone familiar with the topics (as was the goal) would take at most 3-5 minutes to finish the survey. To estimate how much time it might take to finish the survey, the survey was taken by members of the research group who had not seen the survey beforehand and we chose a higher estimate than the maximum taken by anyone in the group. Our experience is that the average completion time for such surveys was always less than the predicted maximum time. Thus we expected the average compensation from the survey to be higher than 7.25 USD, the minimum wage in the US, where all participants were recruited from. Later, based on samples of participants who reported their true completion time in the MTurk forum TurkerView, we predicted an average compensation of 7.89 USD per hour, which is higher than the aforementioned 7.25 USD.

\subsubsection{Results}
We recruited 707 participants in total for this experiment. However, we noticed some noise in that some participants would give an accurate answer inconsistent with their answer for the interval question. We considered that this noise might come from the fact that their best estimate is in the 'binary choice' or 'multiple choice' question, and they're just giving noisy answers for the 'exact answer'. We only consider it true noise  and filter out participants if they choose 'exact answer' as best estimate and give inconsistent answers. Removing participants who gave such noisy answers in at least one of the four questions, we remain with 660 data points. 

As mentioned before, the survey paid \$0.3 to all participants, and they were asked questions about four scenarios in total, the other questions were about US senatorial elections and Covid-19 spread in the US. %

We notice that different people have different levels of confidence about their best estimate. This heterogeneity among agents is something that we focus on in this work.

Now, for each question, rather than comparing an aggregate result to the ground truth (which we do not yet know for most questions yet), we do a plot for mean-square-error (MSE) as a function of possible ground truth values. In Figure~\ref{fig:president-exact}, we compare the MSE plot considering all agents against agents who respectively chose 'exact answer question', 'multiple choice question' and 'binary question' as their best estimate.
\begin{figure}[htp]
    \centering
    \includegraphics[width=0.7\linewidth]{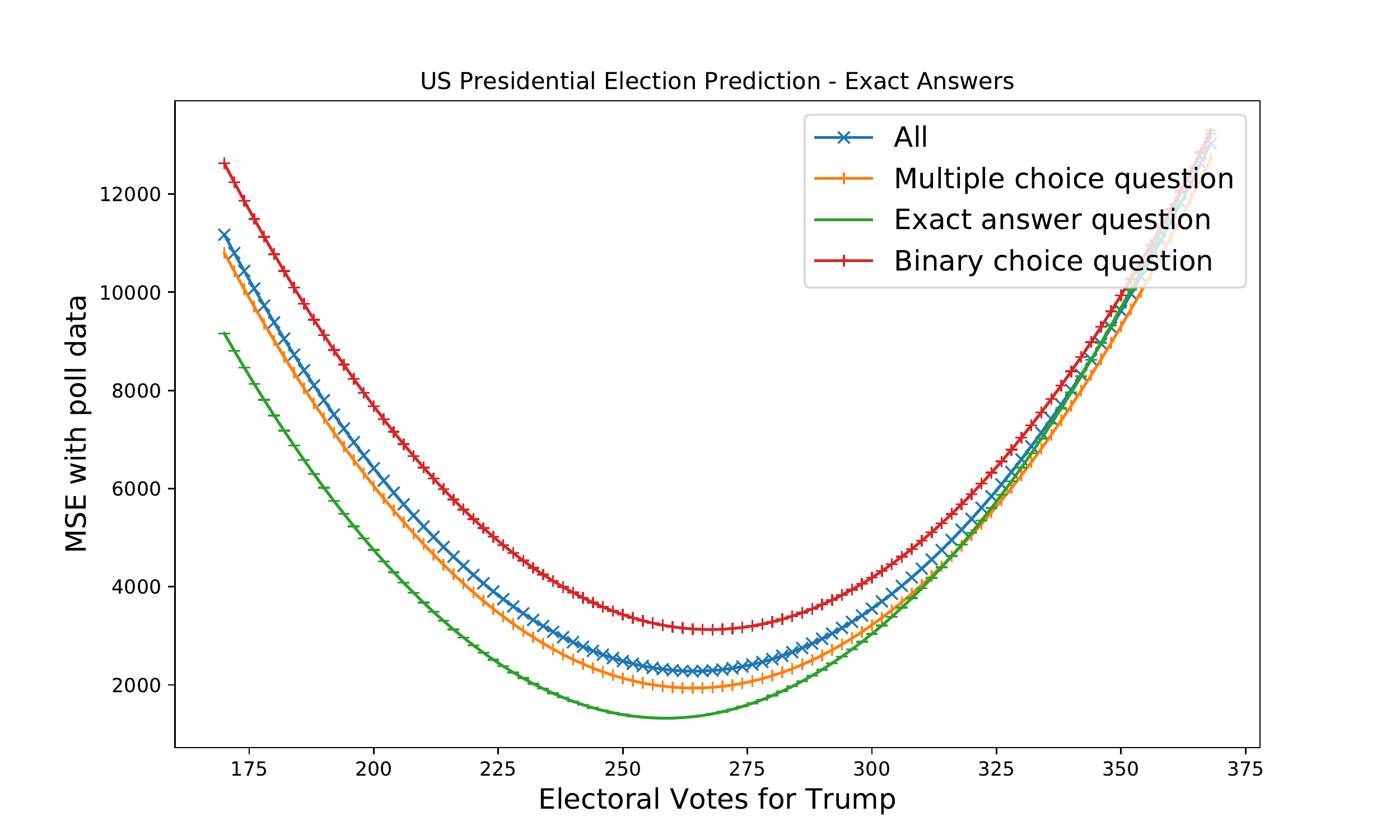}
    \caption{MSE-vs-possible ground truth value for Presidential election prediction}
    \label{fig:president-exact}
\end{figure}
We note that for participants who chose 'exact answer' or 'multiple choice', the MSE is better than the overall MSE for almost all plausible results. In particular, predictions from participants who chose binary question as their best estimate, has poor MSE compared to the two other groups with $p-value<0.01$. However, if we were only interested about prediction about the winner and not the margin, all three groups (and as a collective) has the same majority prediction. 

Additionally, although ``correct predictions" are not our goal at the moment, we can compare results to the only ground truth available to us now i.e. for the presidential election. as of now, the prediction using exact answer questions (lowest MSE is at 258 electoral votes for Donald Trump) is the closest to the presidential election result according to media declarations - Joe Biden wins with 217-248 electoral votes for Donald Trump.

This indicates that offering different types of questions to users and discriminate between them depending on what question they choose to answer can have its merit. For example, all responses can be aggregated to answer the ``who won" question, but it can be better to use the `expert' responses to predict the cardinal answer.

We also create similar plots to Figures ~\ref{fig:president-exact}, by using the ranged answer for everyone. We consider a mean value for each interval, so we can get a reasonable estimate for MSE. 
As we can see in Figure~\ref{fig:president-range}, using the ranged values for people who chose `exact answer' actually gives worse estimates for most values compared to people who actually chose multiple choice question. It might mean, that people in the `multiple choice' group mostly reported values that are closer to the mean of their choice of interval in expectation, thus their MSE does not vary much between the two methods. However, people in the `exact answer' group can have given more detailed answers which loses information in neglecting the full answer.
\begin{figure}[htp]
    \centering
    \includegraphics[width=0.7\linewidth]{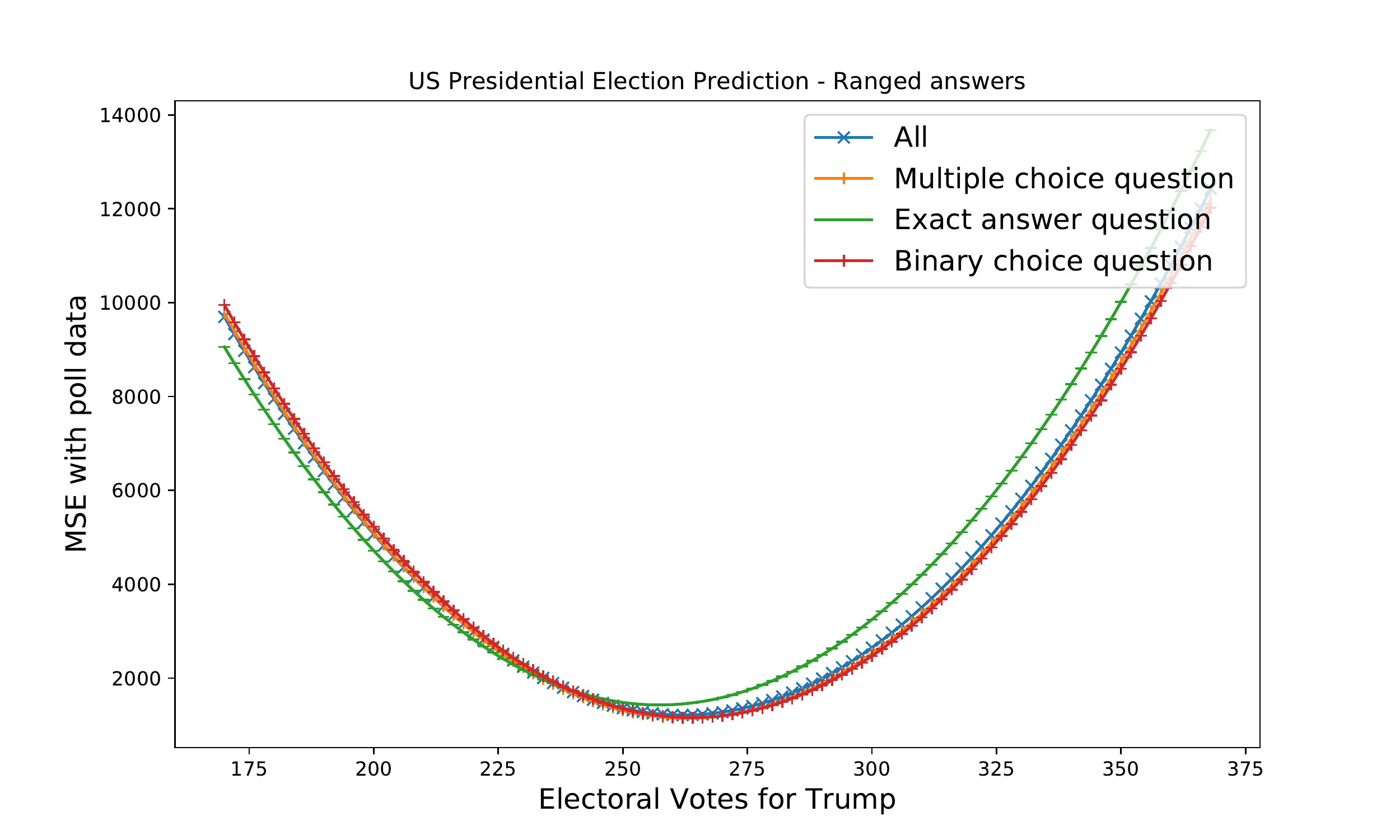}
    \caption{MSE-vs-possible ground truth value for prediction of new  Covid-19 cases in US on Nov 30 using range answers.}
    \label{fig:president-range}
\end{figure}

\end{document}